\documentclass[11pt]{article}
\usepackage[a4paper,top=3cm,bottom=2cm,left=2.5cm,right=2.5cm,marginparwidth=1.75cm]{geometry}
\usepackage[round]{natbib}

\usepackage{bm}
\usepackage{color}
\usepackage{booktabs} 
\usepackage{nicefrac}
\usepackage{float}
\usepackage{caption}
\usepackage{subcaption}
\usepackage{mathtools}
\usepackage{soul}
\mathtoolsset{showonlyrefs}
\usepackage{enumitem}
\usepackage{multirow}
\usepackage{scalerel}
\usepackage{url}
\usepackage{xcolor}
\usepackage[colorlinks = true,
            linkcolor = blue,
            urlcolor  = blue,
            citecolor = blue,
            anchorcolor = blue]{hyperref}

\usepackage{mathtools}
\usepackage{amsmath,mathrsfs,amsfonts,amssymb,amsthm}
\usepackage{thmtools,thm-restate}
\newtheorem{lemma}{Lemma}[section]
\newtheorem{theorem}{Theorem}[section]
\newtheorem{proposition}{Proposition}[section]

\newtheorem{definition}{Definition}[section]
\newtheorem{remark}{Remark}[section]

\usepackage{tcolorbox}
\usepackage{microtype}      
\usepackage{xcolor}         
\usepackage[normalem]{ulem}
\usepackage[none]{hyphenat}
\usepackage{breakcites}
\usepackage[noend]{algpseudocode}
\usepackage[ruled,vlined]{algorithm2e}

\SetCommentSty{mycommfont}
\SetKwInput{KwInput}{Input}                
\SetKwInput{KwOutput}{Output}  
\SetKwInput{KwInitialization}{Initialization}

\usepackage{dsfont}
\allowdisplaybreaks
\newcommand{\textBlue}[1]{{\leavevmode\color{blue}#1}} 

\newcommand{\lp}{\left(}
\newcommand{\rp}{\right)}
\newcommand{\lb}{\left[}
\newcommand{\rb}{\right]}
\newcommand{\lbp}{\left\{}
\newcommand{\rbp}{\right\}}
\newcommand{\lba}{\left\lvert}
\newcommand{\rba}{\right\rvert}
\newcommand{\lV}{\left\lVert}
\newcommand{\rV}{\right\rVert}
\newcommand{\mv}{\middle\vert}

\newcommand{\mcal}{\mathcal}

\newcommand{\mb}{\mathbf}

\newcommand{\mbb}{\mathbb}
\newcommand{\msf}{\mathsf}

\newcommand{\lan}{\langle}
\newcommand{\ran}{\rangle}

\newcommand{\eqDef}{\triangleq}
\newcommand{\diid}{\overset{\text{i.i.d.}}{\sim}}

\newcommand{\E}{\mathbb{E}}
\newcommand{\Var}{\mathsf{Var}}

\renewcommand{\Pr}{\mathbb{P}}

\newcommand{\Unif}{\mathrm{Unif}}

\newcommand{\KLD}[2]{{D}_{\msf{KL}}\left( #1\, \middle\Vert #2 \right)}
\newcommand{\TV}[2]{\lVert #1 - #2 \rVert_{\msf{TV}}}

\newcommand{\Probability}{\mathbb{P}}
\newcommand{\Indicator}{\mathds{1}}
\newcommand{\Expectation}{\mathbb{E}}
\DeclareMathAlphabet{\mathbsf}{OT1}{cmss}{bx}{n}
\DeclareMathAlphabet{\mathssf}{OT1}{cmss}{m}{sl}
\newcommand{\rvbx}{{\mathbsf{x}}} 
\newcommand{\rvbz}{{\mathbsf{z}}} 
\newcommand{\svbx}{{\boldsymbol{x}}} 
\newcommand{\svbu}{{\boldsymbol{u}}} 
\newcommand{\svbz}{{\boldsymbol{z}}} 
\newcommand{\svbp}{{\boldsymbol{p}}} 
\newcommand{\svbmu}{{\boldsymbol{\mu}}} 
\newcommand{\cX}{\mathcal{X}} 
\newcommand{\cU}{\mathcal{U}} 
\newcommand{\cZ}{\mathcal{Z}} 

\newcommand{\qPU}{q^{\texttt{pu}}}
\newcommand{\qSS}{q^{\texttt{ss}}}
\newcommand{\qMRC}{q^{\texttt{mrc}}}
\newcommand{\qMMRC}{q^{\texttt{mmrc}}}
\newcommand{\MRC}{\texttt{MRC}}
\newcommand{\MMRC}{\texttt{MMRC}}
\newcommand{\pMRC}{p^{\texttt{mrc}}}
\newcommand{\pMMRC}{p^{\texttt{mmrc}}}
\newcommand{\piMRC}{\pi^{\texttt{mrc}}}
\newcommand{\piMMRC}{\pi^{\texttt{mmrc}}}

\newcommand{\PrivUnit}{\texttt{PrivUnit}_2}
\newcommand{\SubsetSelection}{\texttt{Subset Selection}}
\newcommand{\sphere}{\mathbb{S}}

\newcommand{\norm}[1]{\left\|{#1}\right\|} 
\newcommand{\ltwo}[1]{\norm{#1}_2} 

\let\svthefootnote\thefootnote
\newcommand\freefootnote[1]{%
	\let\thefootnote\relax%
	\footnotetext{#1}%
	\let\thefootnote\svthefootnote%
}

\title{Optimal Compression of Locally Differentially Private Mechanisms}

\author{%
	Abhin Shah\footnote{Work done while A.S and W.C were interns at Google. J.B. provided the initial idea. P.K. and L.T. gave the conceptual and theoretical framework. A.S. and L.T. designed the algorithm. A.S., W.C., and L.T., devised the proofs. A.S. designed and performed the simulations with support from J.B.. A.S., W.C., P.K., and L.T. wrote the manuscript. J.B., P.K., and L.T., are listed alphabetically.}
	\\
	Massachusetts Institute of Technology\\
	\texttt{abhin@mit.edu} \\
	\and
	Wei-Ning Chen\footnotemark[1]
	\\
	Stanford University \\
	\texttt{wnchen@stanford.edu} \\
	\and
	Johannes Balle\footnotemark[1]
	\\
	Google Research \\
	\texttt{jballe@google.com}
	\and
	Peter Kairouz\footnotemark[1]
	\\
	Google Research \\
	\texttt{kairouz@google.com}
	\and
	Lucas Theis\footnotemark[1]
	\\
	Google Research \\
	\texttt{theis@google.com} 
}
\date{}
\begin{document}
\sloppy
\maketitle
\begin{abstract}
Compressing the output of $\varepsilon$-locally differentially private (LDP) randomizers naively leads to suboptimal utility. In this work, we demonstrate the benefits of using schemes that jointly compress and privatize the data using shared randomness. In particular, we investigate a family of schemes based on Minimal Random Coding \citep{HPHJ19} and prove that they offer optimal privacy-accuracy-communication tradeoffs. Our theoretical and empirical findings show that our approach can compress $\PrivUnit$ \citep{BDFKR2018} and $\SubsetSelection$ \citep{YB18}, the best known LDP algorithms for mean and frequency estimation, to the order of $\varepsilon$ bits of communication while preserving their privacy and accuracy guarantees.
\end{abstract}
\section{Introduction}
\label{sec:introduction}
Machine learning and data analytics are critical tools for designing better products and services. So far, these tools have been predominantly applied in datacenters on data that was curated from millions of users. However, centralized data collection and processing can expose individuals to privacy risks and organizations to legal risks if data is not properly managed. Indeed, increasing privacy concerns are fueling the demand for distributed learning and analytics systems that ensure that the underlying data remains private and secure. This is evident from the recent surge of interest in federated learning and analytics \citep[e.g.,][]{RM20,KM21}.

Designing private and efficient distributed learning and analytics systems involves addressing three main challenges: (a) preserving the privacy of the user's local data, (b) communicating the privatized data efficiently to a central server, and (c) achieving high accuracy on a task (e.g., mean or frequency estimation). Privacy is often achieved by enforcing $\varepsilon$-local differential privacy
($\varepsilon$-LDP) \citep{warner1965randomized, evfimievski2003limiting, dwork2006calibrating, kasiviswanathan2011can}, which guarantees that the outcome from a privatization mechanism will not release too much individual information statistically. Efficient communication, on the other hand, is achieved via compression and dimensionality reduction techniques  \citep{an2016distributed, alistarh17qsgd, wen2017terngrad, wang2018atomo, han2018distributed, han2018geometric, agarwal2018cpsgd, g2019vqsgd, barnes2020rtopk, CKO21}. 

Most existing works focus on addressing two of the three above-mentioned challenges, such as achieving good privacy-accuracy or good communication-accuracy tradeoffs separately. However, doing so can lead to suboptimal performance where all three desiderata are concerned. It is thus important to investigate the joint privacy-communication-accuracy tradeoffs when designing communication-efficient and private distributed algorithms. Under $\varepsilon$-LDP constraints, \cite{CKO20} presents minimax order-optimal mechanisms for frequency and mean estimation that require only $\varepsilon$ bits (independent of the underlying dimensionality of the problem) by using shared randomness\footnote{We assume that the encoder and the decoder can depend on a random quantity that both the server and the user have access to. See Section \ref{subsec:shared_randomness} for details.}.
However,  as noted by \cite{FT21}, the algorithms of \cite{CKO20} are not competitive in terms of accuracy with the best known schemes -- $\SubsetSelection$ for frequency estimation \citep{YB18} and $\PrivUnit$ for mean estimation \citep{BDFKR2018}.
Motivated by this fact, the present work addresses the following fundamental question: \emph{Can we attain the best known accuracy under $\varepsilon$-LDP while only using on the order of $\varepsilon$ bits of communication}? We answer this question affirmatively by leveraging a technique based on importance sampling called Minimal Random Coding \citep{HPHJ19,C08,SCV16}. 

\subsection{Our Contributions}
\label{subsec:contributions}
We first demonstrate that Minimal Random Coding ($\MRC$) can compress any $\varepsilon$-LDP mechanism in a near-lossless fashion using only on the order of $\varepsilon$ bits of communication (see Theorem \ref{theorem:mrc_accuracy_local}). We also prove that the resulting compressed mechanism is $2\varepsilon$-LDP (see Theorem \ref{theorem:mrc_pure_privacy}).
Thus, to achieve $\varepsilon$-LDP, one has to simulate an $\varepsilon/2$ mechanism and pay the corresponding penalty in accuracy. 
Similar to \cite{CKO20}, this approach can achieve the order optimal privacy-accuracy tradeoffs with about $\varepsilon$ bits of communication but is not competitive with the best known LDP schemes. However, we show that this approach is optimal if one is willing to accept approximate LDP with a small $\delta$ (see Theorem~\ref{theorem:mrc_approximate_privacy}).

To overcome the limitations of $\MRC$ in the pure LDP case, we present a modified version ($\MMRC$) such that the resulting compressed mechanism is $\varepsilon$-LDP (see Theorem \ref{theorem:mmrc_privacy}). We show that $\MMRC$ can simulate a large class of LDP mechanisms in a near-lossless fashion using only on the order of $\varepsilon$ bits of communication (see Theorem \ref{theorem:mmrc_accuracy_local_wrt_mrc} in conjunction with Theorem \ref{theorem:mrc_accuracy_local}).

While the class of LDP mechanisms $\MMRC$ can simulate includes the best-known schemes for mean and frequency estimation, $\MMRC$ (similar to $\MRC$) is biased for a fixed number of bits of communication. We show that $\MMRC$ simulating $\PrivUnit$ and $\SubsetSelection$ can be debiased (see Lemma \ref{lemma:mmrc_privunit_bias} and Lemma \ref{lemma:mmrc_ss_bias}), while preserving the corresponding accuracy guarantees (see Theorem \ref{thm:me_mmrc_pu} and Theorem \ref{thm:fe_mmrc_ss}).


Finally, we empirically demonstrate that $\MMRC$ achieves an accuracy comparable to $\PrivUnit$ and $\SubsetSelection$ (see Section \ref{subsec:sim_pu1} and Section \ref{subsec:sim_ss1})\footnote{The source code of our implementation is available at 
\url{https://tinyurl.com/rcc-dp}.} while only using about $\varepsilon$ bits.

We discuss interesting open problems in Section \ref{sec:conclusion} and defer all additional results and experiments to the Appendix.

\subsection{Related Work}
\label{subsec:related_works}
Recent works have examined approaches for compressing LDP schemes in the presence of shared randomness. 
When $\varepsilon \leq 1$, for frequency estimation, \cite{BS15} showed that a single bit is enough to simulate any LDP randomizer with (almost) no impact on its utility although with a large amount of shared randomness. Their result was improved upon, in terms of the amount of shared randomness required, by \cite{bassily2017practical}, \cite{BNS19}, and \cite{acharya2019communication}.

\cite{CKO20} generalized these methods to arbitrary $\varepsilon$'s, and provided order-optimal schemes for both frequency and mean estimation that only use on the order of $\varepsilon$ bits. However, their method is only order-optimal and cannot achieve the accuracy of the best known schemes: $\PrivUnit$ \citep{BDFKR2018} for mean estimation and $\SubsetSelection$ \citep{YB18}\footnote{$\SubsetSelection$ is similar to asymmetric RAPPOR \citep{erlingsson14rappor} in the sense that both have the same marginal distribution. Here, we focus on simulating $\SubsetSelection$.} 
for frequency estimation. We show how one can achieve the accuracy of these schemes with on the order of $\varepsilon$ bits of communication (when $\varepsilon \geq 1$). While we don't advocate large $\varepsilon$ (our methods work for $\varepsilon=1$ as well), we note that larger $\varepsilon$ are both of theoretical and practical interest 
since \textit{amplification via shuffling} can convert a local $\varepsilon > 1$ to a small central $\varepsilon$ \citep{erlingsson2019amplification,balle2019privacy, erlingsson2020encode}.

In the absence of shared randomness, \cite{girgis2021a}, \cite{girgis2021b}, \cite{CKO20} provided order-optimal mechanisms for frequency and mean estimation but their mechanisms do not achieve the best known accuracy. \citet{FT21} presented an approach for compressing $\varepsilon$-LDP schemes in a lossless fashion using a pseudorandom generator (PRG). Their approach, which relies on cryptographic hardness of the PRG, can compress $\SubsetSelection$ to $O(\ln d)$ bits and $\PrivUnit$ to $O(\varepsilon+\ln d)$ bits, where $d$ is the dimension of the underlying problem. 
Their approach, similar to ours, can achieve the privacy vs accuracy tradeoffs of the best known schemes, i.e., $\SubsetSelection$ and $\PrivUnit$. Nevertheless, their approach is designed to work without shared randomness, therefore requiring more bits than necessary if shared randomness is available, as in our work. 

Unlike previous work, our technique of compressing generic LDP schemes relies on Minimal Random Coding ($\MRC$), which was designed to simulate noisy channels. Several papers in information theory and related fields have studied the problem of efficiently simulating noisy channels over digital channels \cite[e.g.,][]{BS02,HJMR07,LE18} and proposed general solutions. In particular, these papers showed that any noisy channel can be simulated at a bit-rate which is close to the mutual information between the information available to the sender and the receiver. However, this result only holds if a shared source of randomness is available. Without such a source, the achievable rate has been shown to be close to Wyner's common information \citep{W75,C08}, which can be significantly larger than the mutual information \citep{XLC11}. While promising as a recipe for simulating arbitrary differentially private mechanisms, the general coding schemes discussed in these papers have not been analyzed for their effect on differential privacy guarantees. $\MRC$ \citep{HPHJ19}, which we analyze and build upon here, is one of these schemes and is also known as likelihood encoder in information theory \citep{C08,SCV16}. 

Finally, mean and frequency estimation under LDP constraints, two canonical problems in distributed learning and analytics, have been widely studied \citep{duchi2013local, nguyn2016collecting, BDFKR2018, wang2019locally, g2019vqsgd, erlingsson14rappor, BS15, kairouz16discrete, YB18, acharya2019hadamard}.
\section{Preliminaries}
\label{sec:preliminaries}
\subsection{Locally Differentially Private (LDP)}
Suppose $\svbx \in \cX$ is some user's data that must remain private.  A privatization mechanism $q$ is a randomized mapping that maps $\svbx \in \cX$ to $\svbz \in \cZ$ with probability $q(\svbz|\svbx)$ where $\cZ$ can be arbitrary. The user transmits $\svbz \sim q(\cdot|\svbx)$, i.e., a privatized version of $\svbx$ to the server. Further, $q$ is $\varepsilon$-LDP if 
\begin{align}
    \forall \svbx, \svbx' \in \cX, \svbz \in \cZ, ~ q(\svbz|\svbx) \leq e^{\varepsilon} q(\svbz|\svbx')  \label{eq:ldp}
\end{align}
and $q$ is $(\varepsilon, \delta)$-LDP if $\forall \svbx, \svbx' \in \cX, Z \subseteq \cZ,$
$$\sum_{\svbz \in Z} q(\svbz|\svbx) \leq e^{\varepsilon} \sum_{\svbz \in Z} q(\svbz|\svbx') + \delta.$$
Here, we focus on $\varepsilon$-LDP mechanisms where $\varepsilon \geq 1$.


\subsection{Shared Randomness}
\label{subsec:shared_randomness}
Here, we allow $\varepsilon$-LDP mechanisms to use \emph{shared randomness}. That is, $q$ can depend on a random variable $\svbu \in \cU$ that is known to both the user and the server (but $\svbu$ is independent of $\svbx$). The corresponding $\varepsilon$-LDP constraint is
$$\forall \svbx, \svbx' \in \cX, \svbz \in \cZ, \svbu \in \cU, ~ q(\svbz|\svbx, \svbu) \leq \exp(\varepsilon)q(\svbz|\svbx', \svbu). $$
The server wishes to reconstruct $\svbx$ from $\svbz$ and the corresponding estimator is allowed to implicitly depend on $\svbu$.
However, for simplicity, we suppress the dependence on $\svbu$ in our notation. In practice, shared randomness can be achieved via downlink communication, that is, the server generates $\svbu$ (e.g., a random seed) and communicates it to the user. Further, we note that such shared randomness can be established well before the advent of any private data\footnote{Quantifying the amount of such shared randomness required remains an open question. See Section \ref{sec:conclusion}.}.

\subsection{\texorpdfstring{PrivUnit$_2$}{PrivUnit}}
The $\PrivUnit$ mechanism $\qPU$, proposed by \cite{BDFKR2018}, is an $\varepsilon$-LDP sampling scheme when the input alphabet $\cX$ is the $d-$dimensional unit $\ell_2$ sphere $\sphere^{d-1}$. Formally, given a vector $\svbx \in \sphere^{d-1}$, $\PrivUnit$ draws a random vector $\svbz$ 
from a spherical cap $\{\svbz \in \sphere^{d-1} \mid
\lan \svbz, \svbx  \ran \ge \gamma\}$ with probability $p_0$ or from its complement $\{\svbz \in \sphere^{d-1} \mid \lan \svbz, \svbx \ran < \gamma\}$  with probability $1 - p_0$, 
where $\gamma \in [0, 1]$ and $p_0  \geq 1/2$ are parameters (depending on $\varepsilon$ and $d$) that trade accuracy and privacy (see Appendix \ref{appendix:privunit}). In other words, $\qPU$ is as follows:
\begin{align}
    \qPU(\svbz | \svbx)   =   
    \begin{cases}
        \dfrac{2 p_0}{A(1,d)I_{1-\gamma^2}(\frac{d-1}{2},\frac{1}{2})} \hspace{10pt} \text{if}\ \langle \rvbx, \rvbz \rangle \geq \gamma \\[10pt]
       \dfrac{2(1-p_0)}{2A(1,d)  -  A(1,d)I_{1-\gamma^2}(\frac{d-1}{2},\frac{1}{2})} \hspace{6pt} \text{otherwise}
    \end{cases}
\end{align}
where $A(1,d)$ denotes the area of $\sphere^{d-1}$ and $I_x(a,b)$ denotes the regularized incomplete beta function. The estimator of the $\PrivUnit$ mechanism (denoted by $\hat{\svbx}^{\texttt{pu}}$) is obtained by dividing every coordinate of $\svbz$ by $m_{\texttt{pu}}$ i.e., $\hat{\svbx}^{\texttt{pu}} \coloneqq \svbz/m_{\texttt{pu}}$ where 
\begin{align}
    m_{\texttt{pu}} \coloneqq  \frac{(1 - \gamma^2)^\alpha}{2^{d-2} (d - 1)}
  \left[\frac{p_0}{B(1; \alpha,\alpha) - B(\tau; \alpha,\alpha)}
     - \frac{1 - p_0}{B(\tau; \alpha, \alpha)}\right]
  \label{eq:m}
\end{align}
with $\alpha = (d-1)/2$, $\tau = (1+\gamma)/2$, and $B(x;\alpha,\beta)$ denoting the incomplete beta function. The estimator $\hat{\svbx}^{\texttt{pu}}$ is (a) unbiased i.e., $\Expectation[\hat{\svbx}^{\texttt{pu}} | \svbx] = \svbx$, (b) has order-optimal utility i.e., $\E[\ltwo{\hat{\svbx}^{\texttt{pu}} - \svbx}^2] = \Theta\lp
   \frac{d}{\min\lp\varepsilon, (e^\varepsilon-1)^2, d\rp}\rp$, and (c) achieves the best known constants for mean estimation. See Appendix \ref{appendix:privunit} for more details on $\PrivUnit$.


\subsection{Subset Selection}
The $\SubsetSelection$ mechanism $\qSS$, proposed by \cite{YB18}, is an $\varepsilon$-LDP sampling scheme when the input alphabet $\cX$ can take $d$ different values.
Without loss of generality, let 
$\cX \coloneqq \lbp e_1, e_2,...,e_d \rbp$, where $e_j \in \{0, 1\}^d$ is the $j^{th}$ standard unit vector, i.e., the one-hot encoding of $j$.
The output alphabet $\cZ$ is the set of all $d$-bit binary strings with Hamming weight $s \coloneqq \lceil \frac{d}{1+e^{\varepsilon}}\rceil$, i.e.,
\begin{align*}
\textstyle
\cZ = \lbp \svbz = (z^{(1)}, \cdots ,z^{(d)}) \in \{0, 1\}^d: \sum_{i=1}^d z^{(i)} = s \rbp.
\end{align*}
Given $\svbx \in \cX$,  $\SubsetSelection$ maps it to $\svbz \in \cZ$ with the following conditional probability:
\begin{align}\label{eq:q_ss}
    \qSS(\svbz|\svbx) \coloneqq \begin{cases}
    \frac{e^{\varepsilon}}{{ \binom{d-1}{s-1} }e^{\varepsilon}+{ \binom{d-1}{s} }} & \text{if}\ \svbz \in \cZ_{\svbx}
    \\[10pt]
    \frac{1}{{ \binom{d-1}{s-1} }e^{\varepsilon}+{ \binom{d-1}{s} }} & \text{if}\ \svbz \in \cZ \setminus \cZ_{\svbx}
    \end{cases}
\end{align}
where $\cZ_{\svbx} = \lbp \svbz = (z^{(1)}, \cdots ,z^{(d)}) \in \cZ : z^{(x)} = 1 \rbp$ is the set of elements in $\cZ$ with $1$ in the $x^{th}$ location. The estimator of the $\SubsetSelection$ mechanism (denoted by $\hat{\svbx}^{\texttt{ss}}$) is obtained by subtracting $b_{\texttt{ss}}$ from every component of $\svbz$ and dividing every component of the result by $m_{\texttt{ss}}$ i.e., $\hat{\svbx}^{\texttt{ss}} \coloneqq (\svbz - b_{\texttt{ss}})/m_{\texttt{ss}}$ where 
\begin{align}\label{eq:m_ss}
    m_{\texttt{ss}} \coloneqq \frac{s(d  -  s)(e^{\varepsilon}-1)}{(d  -  1)(s(e^{\varepsilon}  -  1)+d)},\  b_{\texttt{ss}} \coloneqq  \frac{s((s  -  1)e^{\varepsilon} +(d  -  s))}{(d  -  1)(s(e^{\varepsilon}  -  1)+d)}.
\end{align}
The estimator $\hat{\svbx}^{\texttt{ss}}$ is (a) unbiased i.e., $\Expectation[\hat{\svbx}^{\texttt{ss}} | \svbx] = \svbx$, (b) has optimal utility i.e., $\E[\ltwo{\hat{\svbx}^{\texttt{ss}} - \svbx}^2] = \Theta\lp \frac{d}{\min\lp e^{\varepsilon}, \lp e^{\varepsilon}-1 \rp^2, d \rp}\rp$ and (c) achieves the best known constants for frequency estimation. See Appendix \ref{appendix:ss} for more details on $\SubsetSelection$.
\section{Main Results}
\label{sec:main_results}
In this section, first, we describe the Minimal Random Coding algorithm for compressing any $\varepsilon$-LDP  mechanism and prove its order-optimal privacy-accuracy-communication tradeoffs. Then, we propose the Modified Minimal Random Coding algorithm for compressing any $\varepsilon$-LDP \emph{cap-based mechanism}\footnote{The family of cap-based mechanisms includes
$\PrivUnit$ and $\SubsetSelection$. See Definition \ref{def:cap}.} and prove that it achieves optimal privacy-accuracy-communication tradeoffs.

\subsection{Minimal Random Coding (\texorpdfstring{$\MRC$}{MRC})}
\label{subsec:mrc}
Consider an $\varepsilon$-LDP mechanism $q(\cdot|\svbx)$ that we wish to compress. Under $\MRC$, first, a number of candidates $\svbz_1, \cdots, \svbz_N$ are drawn from a fixed reference distribution $p(\cdot)$ (known to both the user and the server). This can be achieved via a pseudorandom number generator with a known seed. Next, the user transmits an index $K \in [N]$ to the server where $K$ is drawn according to some distribution $\piMRC(\cdot)$ such that $\svbz_K \sim q(\cdot|\svbx)$ approximately. The distribution $\piMRC$ is such that, $\forall k \in [N]$,  $\piMRC(k) \propto w(k)$ where $w(k) \coloneqq q(\svbz_k|\svbx) / p(\svbz_k)$ are the importance weights\footnote{We suppress dependence of $\piMRC$ \& $w$ on $\svbx$ for simplicity.} (see Algorithm \ref{alg:mrc}). To communicate the index $K$ of $\MRC$, $\log N$ bits are required.

\begin{algorithm}
\KwInput{$\varepsilon$-LDP mechanism $q(\cdot|\svbx)$, reference distribution $p(\cdot)$, number of candidates $N$}
Draw samples $\svbz_1 , \cdots, \svbz_N$ from $p(\svbz)$ \tcp*[h]{Using the shared source of randomness}\\
\For {$k \in \{1,\cdots, N\}$}
{
$w(k) \leftarrow q(\svbz_k|\svbx) / p(\svbz_k) $
}
$\piMRC(\cdot) \leftarrow w(\cdot) / \sum_{k}w(k)$\\
\KwOutput{$\piMRC(\cdot), \{\svbz_1 , \cdots, \svbz_N \}$}
\caption{$\MRC$}
\label{alg:mrc}
\end{algorithm}



Let $\qMRC$ denote the distribution of $\svbz_K$ where $K \sim \piMRC(\cdot)$.
The following theorem shows that when the number of candidates is exponential in $\varepsilon$, samples drawn from $\qMRC$ will be similar to samples drawn from $q(\cdot|\svbx)$ in terms of $\ell_2$ error. In other words, $\qMRC$ can compress $q(\cdot|\svbx)$ to the order of $\varepsilon$ bits of communication as well as simulate it in a near-lossless fashion. A proof can be found in Appendix \ref{appendix:utility_mrc}.

\begin{restatable}[Utility of $\MRC$]{theorem}{mrcaccuracylocal}\label{theorem:mrc_accuracy_local}
Consider any input alphabet $\cX$, output alphabet $\cZ$, data $\svbx \in \cX$, and $\varepsilon$-LDP mechanism $q(\cdot|\svbx)$. Consider any reference distribution $p(\cdot)$ such that $| \ln (q(\svbz|\svbx) / p(\svbz)) | \leq \varepsilon$ $\forall ~ \svbx \in \cX, \svbz \in \cZ$.\footnote{Note that this condition holds for many reference distributions $p(\cdot)$. For example, one can simply choose $p(\cdot) = q(\cdot|\svbx^*)$ for some $\svbx^* \in \cX$.} Let the number of candidates be $N = 2^{(\log e + 4c) \varepsilon}$ for some constant $c \geq 0$. Then, for $\alpha \in [0,1/2]$, $\qMRC$ is such that 
\begin{align}\textstyle
     \Big| \Expectation_{\qMRC}\big[\|\svbz - \svbx\|^2\big] - \Expectation_{q}\big[\|\svbz - \svbx\|^2\big] \Big| \leq \frac{2\alpha \sqrt{\Expectation_{q}[\|\svbz - \svbx\|^4]} }{1-\alpha}  \label{eq:mrc_utility}
 \end{align}
 holds with probability at least $1 - 2\alpha$, with $c$ and $\alpha$ related by the following: $\alpha = \sqrt{2^{-c\varepsilon} + 2^{-c^2/\log e + 1}}$.
\end{restatable}
In general, $\Expectation_{q}\lb \|\svbz - \svbx\|^4\rb$ in \eqref{eq:mrc_utility} can be well-controlled. See Remark~\ref{rmk:error_forth_moment} in Appendix \ref{appendix:utility_mrc} for more details.

In the next Theorem, we show that $\piMRC$ is $2\varepsilon$-LDP. Hence, the compressed mechanism $\qMRC$ is 2$\varepsilon$-LDP. 
\begin{restatable}[Pure DP guarantee of $\MRC$]{theorem}{mrcpureprivacy}\label{theorem:mrc_pure_privacy}
Consider any input alphabet $\cX$, output alphabet $\cZ$, and data $\svbx \in \cX$.
Consider any $\varepsilon$-LDP mechanism $q(\cdot|\svbx)$, reference distribution $p(\cdot)$, and number of candidates $N\geq 1$. Then, $\piMRC(\cdot)$ obtained from Algorithm \ref{alg:mrc} is a  2$\varepsilon$-LDP mechanism.
\end{restatable}
A proof is provided in Appendix \ref{appendix:pure_privacy_mrc} and it relies on fact that the following ratio can be bounded by $e^{2\varepsilon}$:
 \begin{align}
    \frac{\piMRC_{\svbx}(k)}{\piMRC_{\svbx'}(k)} = \frac{q(\svbz_k | \svbx)}{q(\svbz_k | \svbx')} \cdot \frac{\sum_{k'} q(\svbz_{k'} | \svbx')/p(\svbz_{k'})}{\sum_{k'} q(\svbz_{k'} | \svbx)/p(\svbz_{k'})}.
\end{align}

In the following Theorem, we show that $\piMRC$ is $(\varepsilon + \varepsilon_0, \delta)$-LDP implying that the compressed mechanism $\qMRC$ is $(\varepsilon + \varepsilon_0, \delta)$-LDP where $\varepsilon_0 > 0$ and $\delta \leq 1$ are free parameters. This Theorem can be viewed complementary to Theorem \ref{theorem:mrc_pure_privacy} where a stronger privacy parameter can be achieved (i.e., $\varepsilon + \varepsilon_0$ which can get arbitrarily close to $\varepsilon$ as opposed to $2\varepsilon$) albeit at the cost of trading pure privacy for approximate privacy. A proof is provided in Appendix \ref{appendix:approx_privacy_mrc}.

\begin{restatable}[Approximate DP guarantee of $\MRC$]{theorem}{mrcapproxprivacy}\label{theorem:mrc_approximate_privacy}
Consider any input alphabet $\cX$, output alphabet $\cZ$, data $\svbx \in \cX$, and $\varepsilon$-LDP mechanism $q(\cdot|\svbx)$. Consider any reference distribution $p(\cdot)$ such that $ | \ln (q(\svbz|\svbx) / p(\svbz)) | \leq \varepsilon$  $\forall ~ \svbx \in \cX, \svbz \in \cZ$.\textBlue{\footnotemark[5]}  Let $c_0 \geq 0$ be some constant and let the number of candidates $N = \exp(2\varepsilon + 2c_0)$. Then, for any $\delta \leq 1$, $\piMRC(\cdot)$ obtained from Algorithm \ref{alg:mrc} is $(\varepsilon + \varepsilon_0, \delta)$-LDP mechanism where
\begin{equation}
 \textstyle
    \varepsilon_0 \coloneqq \ln\dfrac{1+a_0}{1-a_0} \qquad \text{and} \qquad a_0 \coloneqq \exp(-c_0)\sqrt{\frac{1}{2}\ln\frac{2}{\delta}}.
\end{equation}
\end{restatable}


\subsection{Modified Minimal Random Coding (\texorpdfstring{$\MMRC$}{MMRC})}
\label{subsec:mmrc}
While the results regarding $\MRC$ in Section \ref{subsec:mrc} are general and offer order optimal privacy-accuracy tradeoffs with about $\varepsilon$ bits of communication, the resulting compressed mechanism is not exactly $\varepsilon$-LDP. More specifically, Theorem \ref{theorem:mrc_pure_privacy} introduces an additional factor of $2$ in the LDP guarantee and Theorem \ref{theorem:mrc_approximate_privacy} provides an approximate privacy guarantee instead of a pure privacy guarantee. To address these limitations, we focus on a class of $\varepsilon$-LDP mechanisms which we call \emph{cap-based} mechanisms and propose a modification to $\MRC$ such that the resulting compressed mechanism is $\varepsilon$-LDP. Further, like $\MRC$, $\MMRC$ can simulate the underlying $\varepsilon$-LDP mechanism in a near-lossless fashion while using only on the order of $\varepsilon$ bits.


We start with the definition of cap-based mechanism which is inspired from the structure of $\PrivUnit$ and $\SubsetSelection$.
\begin{definition}[Cap-based Mechanisms]\label{def:cap}
An $\varepsilon$-LDP mechanism $q(\svbz | \svbx)$ with  input alphabet $\cX$ and output alphabet $\cZ$ is a cap-based mechanism if it can be written in the following way:
\begin{align}
\textstyle
    q(\svbz | \svbx) = 
    \begin{cases}
      c_1(\varepsilon, d) & \text{if}\ \svbz \in \msf{Cap}_{\svbx}
      \\[10pt]
      c_2(\varepsilon, d) & \text{if}\ \svbz \notin \msf{Cap}_{\svbx}
    \end{cases} \label{eq:cap_mechanism}
\end{align}
where (a) $c_1(\varepsilon, d)$ and $c_2(\varepsilon, d)$ are constants with respect to $\svbx$ and $\svbz$ such that $c_1(\varepsilon, d) \geq c_2(\varepsilon, d)$, and (b) $\msf{Cap}_{\svbx} \subseteq \cZ$ such that $\Probability_{\svbz \sim \Unif(\cZ)}\lp \svbz \in \msf{Cap}_{\svbx} \rp$ is independent of $\svbx$ and is at least $c_2(\varepsilon, d) / 2 c_1(\varepsilon, d)$.
\end{definition}

In words, a cap-based $\varepsilon$-LDP mechanism samples uniformly either from $\msf{Cap}_{\svbx}$ or from $\cZ \setminus \msf{Cap}_{\svbx}$ where $\msf{Cap}_{\svbx} \subseteq \cZ$ is such that if $\svbz$ is sampled uniformly from $\cZ$, it will belong to $\msf{Cap}_{\svbx}$ with probability at least $c_2(\varepsilon, d) / 2 c_1(\varepsilon, d)$. 
It is easy to see that $\qSS$ defined in \eqref{eq:q_ss} is a cap-based mechanism with $\msf{Cap}_{\svbx} = \cZ_{\svbx}$, $c_1(\varepsilon, d) = \frac{e^{\varepsilon}}{{ \binom{d-1}{s-1} }e^{\varepsilon}+{ \binom{d-1}{s} }}$, and $c_2(\varepsilon, d) = \frac{1}{{ \binom{d-1}{s-1} }e^{\varepsilon}+{ \binom{d-1}{s} }}$. See Appendix \ref{appendix:ss} where we evaluate $\Probability_{\svbz \sim \Unif(\cZ)}\lp \svbz \in \cZ_{\svbx} \rp$ and show that it is at least $ 1/2e^{\varepsilon}$. In Appendix \ref{appendix:privunit}, we show that $\qPU$ is a cap-based mechanism.

For a cap-based $\varepsilon$-LDP mechanism $q(\svbz | \svbx)$ and a uniform reference distribution $p(\cdot)$, the distribution $\piMRC$ obtained from Algorithm \ref{alg:mrc} takes a special form:
\begin{align}
   \piMRC(k) =
  \begin{cases}
     \frac{1}{N} \times \frac{c_1(\varepsilon, d)}{\theta c_1(\varepsilon, d) + (1 - \theta) c_2(\varepsilon, d)} \ \text{if}\ \svbz_k \in \msf{Cap}_{\svbx}
      \\[10pt]
     \frac{1}{N} \times \frac{c_2(\varepsilon, d)}{\theta c_1(\varepsilon, d) + (1 - \theta) c_2(\varepsilon, d)} \ \text{if}\ \svbz_k \notin \msf{Cap}_{\svbx}
    \end{cases}\label{eq:special_pi_mrc}
\end{align}
where $\theta$ is the fraction of candidates inside the $\msf{Cap}_{\svbx}$, i.e.,  $\theta = \frac{1}{N} \sum_{k} \Indicator(\svbz_k \in \msf{Cap}_{\svbx})$. As is, $\piMRC$ in \eqref{eq:special_pi_mrc} is not necessarily $\varepsilon$-LDP because $\theta$ can be different for $\svbx$ and $\svbx'$. However, as $N \to \infty$, $\theta \to  \Expectation[\theta] = \Probability_{\svbz \sim \Unif(\cZ)}\lp \svbz \in \msf{Cap}_{\svbx} \rp$, which is not a function of $\svbx$, implying that $\piMRC_{\svbx}(k) / \piMRC_{\svbx'}(k) \leq c_1(\varepsilon,d)/c_2(\varepsilon, d) \leq \exp(\varepsilon)$\footnote{This follows from \eqref{eq:ldp} and \eqref{eq:cap_mechanism} because $q(\cdot | \svbx)$ is $\varepsilon$-LDP.}. This shows that $\piMRC$ is $\varepsilon$-LDP when $N \to \infty$. This motivates us to modify $\piMRC$ to $\piMMRC$ such that $\piMMRC$ is $\varepsilon$-LDP irrespective of $N$. Further, when $N$ is large enough, the modification is not by much, i.e., a sample from $\piMRC$ is similar to a sample from $\piMMRC$.

To that end, define an upper threshold $t_u = \frac{1}{N} \times \frac{c_1(\varepsilon, d)}{\Expectation[\theta] c_1(\varepsilon, d) + (1 - \Expectation[\theta]) c_2(\varepsilon, d)}$ and a lower threshold  $t_l = \frac{1}{N} \times \frac{c_2(\varepsilon, d)}{\Expectation[\theta] c_1(\varepsilon, d) + (1 - \Expectation[\theta]) c_2(\varepsilon, d)}$, and initialize $\piMMRC$ to be equal to $\piMRC$. We want to modify $\piMMRC$ so as to ensure:
\begin{align}
    t_l \leq \piMMRC(k) \leq t_u ~ \forall k \in [N], \label{eq:modification}
\end{align}
which, as argued above, guarantees $\varepsilon$-LDP irrespective of the choice of $N$.
First, it is easy to see that $\theta c_1(\varepsilon, d) + (1 - \theta) c_2(\varepsilon, d)$ is an increasing function of $\theta$. Next, we will look at 3 cases depending on the relationship between $\theta$ and $\E[\theta]$: (A) If $\theta = \E[\theta]$, then $\piMMRC$ already satisfies \eqref{eq:modification}; (B) If $\theta < \E[\theta]$, then only the upper threshold is violated and we set $\piMMRC(k) = t_u ~ \forall k : \svbz_k \in \msf{Cap}_{\svbx}$ and re-normalize the remaining $\piMMRC(k)$; (C) If $\theta > \E[\theta]$, then only the lower threshold is violated, we set $\piMMRC(k) = t_l ~ \forall k : \svbz_k \notin \msf{Cap}_{\svbx}$ and re-normalize the remaining $\piMMRC(k)$. The re-normalization step does not violate \eqref{eq:modification}. We provide pseudo-code to calculate $\piMMRC$ in Algorithm \ref{alg:mmrc}. 

\begin{algorithm}
\KwInput{$\varepsilon$-LDP cap-based mechanism $q(\cdot|\svbx)$, the associated $\msf{Cap}_{\svbx}$, reference distribution $p(\cdot)$, number of candidates $N$, lower threshold $t_l$, upper threshold $t_u$}
$\piMRC(\cdot), \{\svbz_1 , \cdots, \svbz_N \} \leftarrow$ $\MRC(p(\cdot), q(\cdot|\svbx), N)$\\
$\theta \leftarrow \frac{1}{N} \sum_{k} \Indicator(\svbz_k \in \msf{Cap}_{\svbx})$ \tcp*[h]{Compute the fraction of candidates inside the cap}\\
\KwInitialization{$\piMMRC(\cdot) \leftarrow \piMRC(\cdot)$}
\If{$\max_{k} \piMMRC(k) > t_u$} 
  {\tcp*[h]{Upper threshold is violated}\\
  $\piMMRC(k) \leftarrow t_u$, $\forall ~ k : \svbz_k \in \msf{Cap}_{\svbx}$
  $\piMMRC(k) \leftarrow \frac{1-N \theta t_u }{N(1 - \theta)}$, $\forall ~ k : \svbz_k \notin \msf{Cap}_{\svbx}$ 
  }
\ElseIf{{$\min_{k} \piMMRC(k) < t_l$}}
    {\tcp*[h]{Lower threshold is violated}\\
  $\piMMRC(k) \leftarrow t_l$, $\forall ~ k : \svbz_k \notin \msf{Cap}_{\svbx}$
  $\piMMRC(k) \leftarrow \frac{1-N (1-\theta) t_l }{N \theta }$, $\forall ~ k : \svbz_k \in \msf{Cap}_{\svbx}$
  }
 \KwOutput{$\piMMRC(\cdot), \{\svbz_1 , \cdots, \svbz_N \}$}
\caption{$\MMRC$}
\label{alg:mmrc}
\end{algorithm}

Let $\qMMRC$ denote the distribution of $\svbz_K$ where $K \sim \piMMRC(\cdot)$.
In the following Theorem, we show that $\piMMRC$ is $\varepsilon$-LDP implying that the compressed mechanism $\qMMRC$ is $\varepsilon$-LDP. The proof follows from \eqref{eq:modification} and can be found in Appendix \ref{appendix:privacy_mmrc}.
\begin{restatable}[DP guarantee of $\MMRC$]{theorem}{mmrcprivacy}\label{theorem:mmrc_privacy}
Consider any input alphabet $\cX$, output alphabet $\cZ$, data $\svbx \in \cX$, and $\varepsilon$-LDP cap-based mechanism $q(\cdot|\svbx)$.
Let the reference distribution $p(\cdot)$ be the uniform distribution on $\cZ$. Consider any number of candidates $N \geq 1$. Then, $\piMMRC(\cdot)$ obtained from Algorithm \ref{alg:mmrc} is an $\varepsilon$-LDP mechanism.
\end{restatable}
The following Theorem shows that, with number of candidates exponential in $\varepsilon$, samples drawn from $\qMMRC$ will be similar to the samples drawn from $\qMRC$ in terms of $\ell_2$ error. A proof can be found in Appendix \ref{appendix:utility_mmrc}.

\begin{restatable}[Utility of $\MMRC$]{theorem}{mmrcaccuracylocalwrtmrc}\label{theorem:mmrc_accuracy_local_wrt_mrc}
Consider any input alphabet $\cX$, output alphabet $\cZ$, data $\svbx \in \cX$, and $\varepsilon$-LDP cap-based mechanism $q(\cdot|\svbx)$.
Let the reference distribution $p(\cdot)$ be the uniform distribution on $\cZ$.
Let $N$ denote the number of candidates. Then, $\qMMRC$ is such that 
\begin{align}
    \Expectation_{\qMMRC} \big[ \lV  \svbz - \svbx \rV^2_2  \big]  \leq \Expectation_{\qMRC} \big[ \lV  \svbz - \svbx \rV^2_2  \big] + \sqrt{\frac{\rho (1+\varepsilon)}{2}}  \max_{\svbx, \svbz} \lV  \svbz - \svbx \rV^2_2  \label{eq:mmrc_utility}
\end{align}
 where $\rho \in (0,1)$ is such that $N = \frac{2\lp\exp(\varepsilon)-1\rp^2}{\rho^2}  \ln\frac{2}{\rho}.$
\end{restatable}
For bounded mechanisms, $\max_{\svbx, \svbz} \lV  \svbz - \svbx \rV^2_2$ in \eqref{eq:mmrc_utility} can be well-controlled. See Remark~\ref{rmk:error_max} in Appendix \ref{appendix:utility_mmrc} for a discussion.

In conjunction with Theorem \ref{theorem:mrc_accuracy_local}, Theorem \ref{theorem:mmrc_accuracy_local_wrt_mrc} implies that $\qMMRC$ can compress $q(\cdot|\svbx)$ to the order of $\varepsilon$ bits of communication and simulate it in a near-lossless fashion. This is stated formally and proved in Appendix \ref{appendix:utility_mmrc}.
\section{Mean Estimation}
\label{sec:mean_estimation}
In this section, we focus on the mean estimation problem, which is a canonical statistical task in distributed estimation with applications in distributed stochastic gradient descent, federated learning, etc. 
Let the input space $\mcal{X}$ be the $d$-dimensional unit $\ell_2$ sphere, i.e.,  $\mcal{X} = \sphere^{d-1}$.
Consider $n$ users where user $i$  
has some data $\svbx_i \in \mcal{X}$. For every $i \in [n]$, let $\svbx_i$ be privatized using an $\varepsilon$-LDP mechanism $q(\cdot|\svbx_i)$ and potentially post-processed to obtain an estimate $\hat{\svbx}_i$ of $\svbx_i$. We are interested in estimating the \emph{empirical mean} $\svbmu \eqDef \frac{1}{n}\sum_i \svbx_i$ using $\hat{\svbx}_1,\cdots,\hat{\svbx}_n$ such that the mean estimation error defined below is minimized
\begin{equation}\label{eq:rme_def}
     r_{\msf{ME}} \lp \hat{\svbmu}, q \rp \eqDef \max_{\svbx^n \in \mcal{X}^n} \E\lb \left\| \hat{\svbmu}\lp \hat{\svbx}_1,\cdots ,\hat{\svbx}_n \rp - \svbmu \right\|^2_2 \rb,
\end{equation}
where $\hat{\svbmu}$ is an estimate of $\svbmu$ and
the expectation is with respect to $q(\cdot|\svbx_i) $ as well as all (possibly shared) randomness used by $q(\cdot|\svbx_i)~\forall i \in [n]$.

\cite{BDFKR2018} show that $\PrivUnit$ achieves the order-optimal privacy-accuracy trade-off for mean estimation, i.e., $r_{\msf{ME}} \lp \hat{\svbmu}^\texttt{pu}, \qPU \rp = \Theta\lp\frac{d}{\min\lp\varepsilon, (e^\varepsilon-1)^2, d\rp}\rp $ where $\hat{\svbmu}^\texttt{pu} \coloneqq \frac{1}{n}\sum_i \hat{\svbx}^\texttt{pu}_i$.
Moreover, compared to other (order-optimal) $\varepsilon$-LDP mean estimation mechanisms, $\PrivUnit$ admits the best constants and gives the smallest $\ell_2$ error in practice (see \cite{FT21}). However, $\PrivUnit$ requires each user to send a $d$-dimensional real vector, so without any compression, the communication needed is $\Theta(d)$ bits, which can be an issue in many practical scenarios. 
 
To compress and simulate $\PrivUnit$, one can directly apply the generic $\MMRC$ mechanism defined in Section~\ref{subsec:mmrc}. However, for a fixed number of candidates $N$, $\MMRC$ yields a biased estimate of $\svbx$ and hence cannot get the correct (optimal) order of estimation error in \eqref{eq:rme_def}, i.e., the error would not decay with $n$\footnote{We note that this does not undermine the significance of Theorem \ref{theorem:mrc_accuracy_local} and Theorem \ref{theorem:mmrc_accuracy_local_wrt_mrc}. These are useful in single-user settings (i.e., $n = 1$) and are generic as they can compress (near-losslessly) any $\varepsilon$-LDP and $\varepsilon$-LDP cap-based mechanism, respectively.}. Fortunately, we show (in Section~\ref{subsec:mmrc_privunit}) that the bias can be corrected by appropriately scaling the privatized version of $\svbx$, i.e., by using an estimator which is
slightly different compared to the original estimator of $\PrivUnit$. Further, we also show (in Section \ref{subsec:sim_pu1}) that the resulting unbiased estimator for mean estimation ($\hat{\svbmu}^{\texttt{mmrc}}$) can simulate $\PrivUnit$ closely while only using on the order of $\varepsilon$ bits of communication.

\subsection{Debiasing \texorpdfstring{$\MMRC$}{MMRC} to simulate \texorpdfstring{$\PrivUnit$}{PrivUnit}}
\label{subsec:mmrc_privunit}

Let us focus on a single user and consider some data $\svbx \in \cX$. Recall the $\PrivUnit$ $\varepsilon$-LDP mechanism $\qPU$ described in Section \ref{sec:preliminaries} with parameters $p_0$ and $\gamma$. $\PrivUnit$ is a cap-based mechanism with $\msf{Cap}_{\svbx} = \{\svbz \in \sphere^{d-1} \mid
\lan \svbz, \svbx  \ran \ge \gamma\}$ (see Appendix \ref{appendix:privunit} for details).
Let $\piMMRC$ be the distribution and $\svbz_1, \svbz_2,...,\svbz_N$ be the candidates obtained from Algorithm \ref{alg:mmrc} when the reference distribution is $\Unif(\sphere^{d-1})$. Let  $K \sim \piMMRC(\cdot)$. Therefore, $\svbz_K$ is the privatized version of $\svbx$ using $\MMRC$.

Define $p_{\texttt{mmrc}} \coloneqq \Probability(\svbz_K \in \msf{Cap}_{\svbx})$
to be the probability with which the sampled candidate $\svbz_K$ belongs to the spherical cap associated with $\PrivUnit$.
Define $m_{\texttt{mmrc}}$ as the scaling factor in \eqref{eq:m} when $p_0$ in \eqref{eq:m} is replaced by $p_{\texttt{mmrc}}$. Define $\hat{\svbx}^\texttt{mmrc} \coloneqq \svbz_K / m_\texttt{mmrc}$ as the estimator of the $\MMRC$ mechanism simulating $\PrivUnit$.  The following Lemma shows that $\hat{\svbx}^\texttt{mmrc}$ is an unbiased estimator. See Appendix \ref{appendix:mmrc_pu_debias} for a proof.

\begin{restatable}{lemma}{mmrcprivunitbias}\label{lemma:mmrc_privunit_bias}
Let $\hat{\svbx}^\texttt{mmrc}$ be the estimator of the \emph{$\MMRC$} mechanism simulating \emph{$\PrivUnit$} as defined above. Then, $\Expectation_{\qMMRC}[\hat{\svbx}^\texttt{mmrc}] = \svbx$.
\end{restatable}

\subsection{Simulating \texorpdfstring{$\PrivUnit$}{PrivUnit} using \texorpdfstring{$\MMRC$}{MMRC}}
\label{subsec:sim_pu1}

Finally, we consider estimating the empirical mean $\svbmu$ defined earlier using the $\MMRC$ scheme simulating $\PrivUnit$. To that end, consider $n$ users and let $\hat{\svbx}^{\texttt{mmrc}}_i$ be the unbiased estimator of $\svbx_i$ at the $i^{th}$ user. Let the (unbiased) estimate of $\svbmu$ be $\hat{\svbmu}^{\texttt{mmrc}} \coloneqq \frac{1}{n}\sum_i \hat{\svbx}^{\texttt{mmrc}}_i$. 

The following Theorem shows that, for mean estimation, $\MMRC$ can simulate $\PrivUnit$ in a near-lossless manner (when $n$ is large and $\lambda$ is small) while only using on the order of $\varepsilon$ bits of communication.
A proof can be found in Appendix \ref{appendix:mmrc_pu_utility}. The key idea in the proof is to show that when the number of candidates $N$ is exponential in $\varepsilon$, the scaling factor $m_{\texttt{mmrc}}$ is close to the scaling parameter associated with $\PrivUnit$ (i.e., $m_{\texttt{pu}}$ defined in \eqref{eq:m}).

\begin{restatable}{theorem}{mmrcpu}
\label{thm:me_mmrc_pu}
Let $r_{\msf{ME}} \lp \hat{\svbmu}^\texttt{pu}, \qPU \rp$ and $r_{\msf{ME}} \lp \hat{\svbmu}^\texttt{mmrc}, \qMMRC \rp$ be the empirical mean estimation error for \emph{$\PrivUnit$} with parameter $p_0$ and \emph{$\MMRC$} simulating \emph{$\PrivUnit$} with $N$ candidates respectively. Consider any $\lambda > 0$. Then,
\begin{align}
    r_{\msf{ME}} \lp \hat{\svbmu}^\texttt{mmrc}, \qMMRC \rp \leq \lp 1+\lambda \rp^2 r_{\msf{ME}} \lp \hat{\svbmu}^\texttt{pu}, \qPU \rp + 2(1+\lambda)(2+\lambda) \sqrt{\frac{r_{\msf{ME}} \lp \hat{\svbmu}^\texttt{pu}, \qPU \rp}{n}} + \frac{(2+\lambda)^2}{n}.
\end{align}
as long as 
\begin{align}\label{eq:N_bdd_mmrc_pu}
    N \geq   \frac{e^{2\varepsilon}}{2}\lp\frac{2 (1+\lambda)}{\lambda \lp p_0 -1/2 \rp}\rp^2 \ln\lp \frac{4(1+\lambda)}{\lambda \lp p_0 -1/2 \rp} \rp.
\end{align}
\end{restatable}
We note that while a specific value of $\lambda$ can be chosen (say $0.1$ or smaller) in \eqref{eq:N_bdd_mmrc_pu}, in practice, the number of bits could be fixed (see Section \ref{subsec:mmrc_privunit_empirical}), determining the value of $\lambda$.

\subsection{Empirical Comparisons}
\label{subsec:mmrc_privunit_empirical}
Next, we empirically demonstrate the privacy-accuracy-communication tradeoffs of $\MMRC$ simulating $\PrivUnit$. Along with $\PrivUnit$, we compare against the SQKR algorithm of \cite{CKO20} which offers order-optimal privacy-accuracy tradeoffs while requiring only $\varepsilon$ bits. Following \cite{CKO20}, we generate data independently but non-identically to capture the distribution-free setting as well as ensure that the data non-central, i.e. $\svbmu \neq 0$. More specifically, we set $\svbx_1,\cdots,\svbx_{n/2} \diid N(1,1)^{\otimes d}$ and $\svbx_{n/2+1},\cdots,\svbx_{n} \diid N(10,1)^{\otimes d}$. Further, to ensure that each data lies on $\sphere^{d-1}$, we normalize each $\svbx_i$ by setting $\svbx_i \leftarrow \svbx_i/\lV \svbx_i \rV_2$. We report the average $\ell_2$ estimation error over 10 runs. See more variations in Appendix \ref{appendix:mmrc_pu_emp}. 

\begin{figure}[h]
\centering
\includegraphics[width=0.45\linewidth]{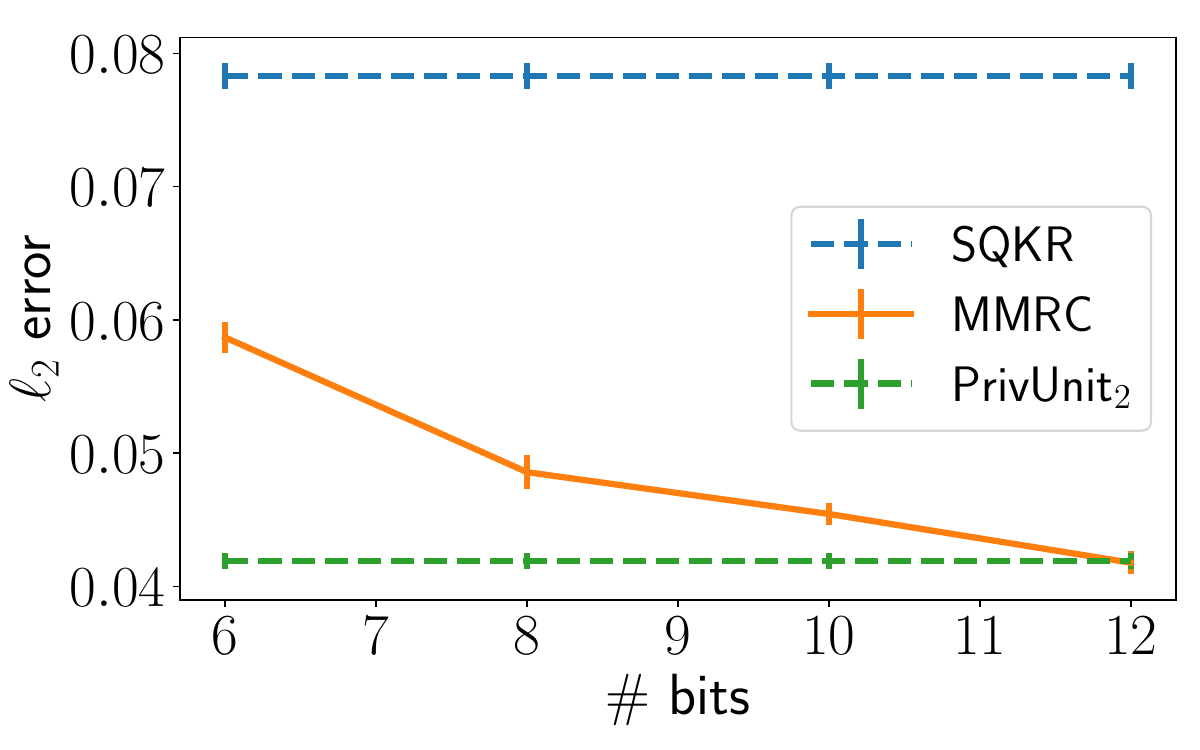} \qquad
\includegraphics[width=0.45\linewidth]{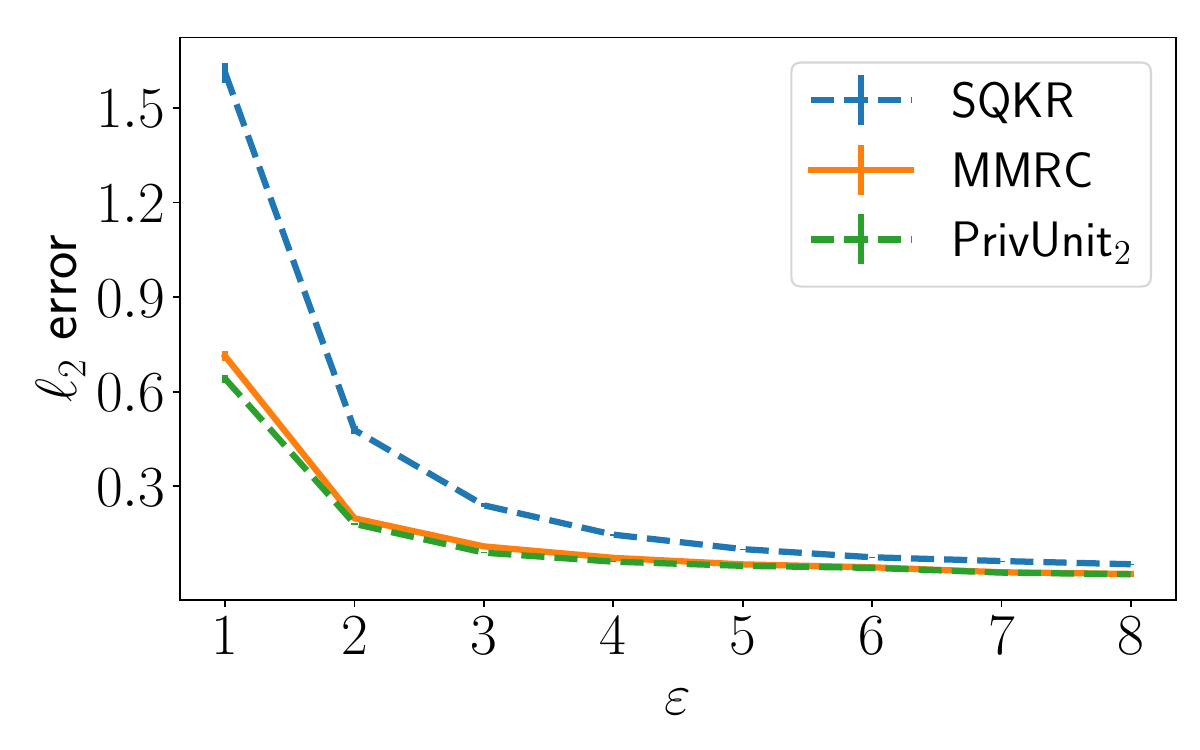}%
\caption{Comparing $\PrivUnit$, $\MMRC$ simulating $\PrivUnit$ and SQKR for mean estimation with $d = 500$ and $n = 5000$. \textbf{Left:} $\ell_2$ error vs $\#$bits for $\varepsilon = 6$. \textbf{Right:} $\ell_2$ error vs $\varepsilon$ for $\#$bits $= \max\{( \varepsilon/ \ln 2) + 2, 8\}$. SQKR uses $\#$-bits $= \varepsilon$ for both because it leads to a poor performance if $\#$-bits $ > \varepsilon$.}
\label{fig:mean}
\end{figure}

In Figure \ref{fig:mean} (Left), we show the communication-accuracy tradeoffs. We see that with correct order of bits, the accuracy of $\MMRC$ simulating $\PrivUnit$ converges to the accuracy of the uncompressed $\PrivUnit$. In Figure \ref{fig:mean} (Right), we show the privacy-accuracy tradeoffs. We see that $\MMRC$ simulating $\PrivUnit$ can attain accuracy of the uncompressed $\PrivUnit$ for the range of  $\varepsilon$'s typically considered by LDP mechanisms while only using $\max\{\lceil( \varepsilon/ \ln 2)\rceil + 2, 8\}$ bits.
\section{Frequency Estimation}
\label{sec:frequency_estimation}
In this section, we study the frequency estimation problem, which is another canonical statistical task in distributed distribution estimation, with application to federated analytics \citep{RM20}. 

Let $\mcal{X}$ be a set of $d$ distinct symbols and without loss of generality $\mcal{X} \coloneqq \lbp e_1, e_2,...,e_d \rbp$, where $e_j \in \{0, 1\}^d$ is the $j^{th}$ standard unit vector i.e., $e_j$ is the one-hot encoding of $j$. Consider $n$ users where user $i$  
has some data $\svbx_i \in \mcal{X}$. For every $i \in [n]$, let $\svbx_i$ be privatized using an $\varepsilon$-LDP mechanism $q(\cdot|\svbx_i)$ and potentially post-processed to obtain an estimate $\hat{\svbx}_i$ of $\svbx_i$.
We are interested in estimating the \emph{empirical distribution} of $\svbx_1, \cdots, \svbx_n$ defined as $\Pi \eqDef \frac{1}{n}\sum_i \svbx_i$ using $\hat{\svbx}_1,\cdots,\hat{\svbx}_n$ such that the estimation error defined below is minimized:
\begin{equation}\label{eq:rfe_def}
r_{\texttt{FE}} \lp \hat{\Pi}, q, \ell \rp \eqDef \max_{\svbx^n \in \mcal{X}^n} \E\lb \ell \lp  \hat{\Pi}(\hat{\svbx}_1,...,\hat{\svbx}_n), \Pi \rp \rb,
\end{equation}
where $\ell =  \lVert \cdot \rVert_1$ or $\lVert \cdot \rVert^2_2$, $\hat{\Pi}$ is an estimate of $\Pi$ and
the expectation is with respect to $q(\cdot|\svbx_i) $ as well as all (possibly shared) randomness used by $q(\cdot|\svbx_i)~\forall i \in [n]$. For simplicity, we only focus on $\ell_2$ error i.e., $\ell = \lVert \cdot \rVert^2_2$.

\cite{YB18} show that the $\SubsetSelection$ achieves the
order-optimal privacy-accuracy trade-off for frequency estimation i.e., $r_{\msf{FE}} \bigl( \hat{\Pi}^\texttt{ss}, \qSS \bigr) = \Theta\Bigl(
  \frac{d}{\min( e^\varepsilon, (e^\varepsilon-1)^2, d )}\Bigr)$ (where $\hat{\Pi}^\texttt{ss} \coloneqq \frac{1}{n}\sum_i \hat{\svbx}^{\texttt{ss}}_i$). Like $\PrivUnit$, compared to other (order-optimal) $\varepsilon$-LDP frequency estimation mechanisms, $\SubsetSelection$ admits the best constants and gives the smallest $\ell_2$ error in practice (see \cite{CKO20}).
However,  the communication cost associated with $\SubsetSelection$ is $O\bigl(\frac{d}{e^\varepsilon+1}\bigr)$ bits per user, which which can be
an issue for small and moderate $\varepsilon$. 
 
Similar to $\PrivUnit$, one could apply the generic $\MMRC$ scheme defined in Section~\ref{sec:main_results} to compress and simulate $\SubsetSelection$. However, for a fixed number of candidates $N$, it yields a biased estimate of $\svbx$ and hence cannot get the correct (optimal) order of estimation error in \eqref{eq:rfe_def} i.e., the error would not decay with $n$. Fortunately, similar to $\PrivUnit$, we show (in Section~\ref{subsec:mmrc_ss}) that the bias can be corrected by appropriately translating and scaling the privatized version of $\svbx$ i.e., by using an estimator which is slightly different compared to the original estimator of $\SubsetSelection$.
Further, we also show (in Section \ref{subsec:sim_ss1}) that the resulting unbiased estimator for frequency estimation ($\hat{\Pi}^{\texttt{mmrc}}$) can simulate $\SubsetSelection$ closely while only using on the order of $\varepsilon$-bits communication.  

\subsection{Debiasing \texorpdfstring{$\MMRC$}{MMRC} to simulate \texorpdfstring{$\SubsetSelection$}{Subset Selection}}
\label{subsec:mmrc_ss}
Let us focus on a single user and consider some data $\svbx \in \cX$.
Recall the $\SubsetSelection$ $\varepsilon$-LDP mechanism $\qSS$ described in Section~\ref{sec:preliminaries} with $s \coloneqq \lceil \frac{d}{1+e^\varepsilon}\rceil$. $\SubsetSelection$ is cap-based mechanism as discussed in Section \ref{sec:main_results} and Appendix \ref{appendix:ss} with $\msf{Cap}_{\svbx} = \cZ_{\svbx}$ and $\Probability_{\svbz \sim \Unif(\cZ)}\lp \svbz \in \msf{Cap}_{\svbx} \rp = s/d$.
Similar to Section \ref{subsec:mmrc_privunit}, let $\svbz_K$ be the privatized version of $\svbx$ using $\MMRC$. We define $\hat{\svbx}^\texttt{mmrc} \coloneqq (\svbz_K- b_\texttt{mmrc} )/ m_\texttt{mmrc}$ as the estimator of the $\MMRC$ mechanism simulating $\SubsetSelection$ where $m_\texttt{mmrc}$ and $b_\texttt{mmrc}$ (defined in Appendix \ref{appendix:mmrc_ss_debias}) are translation and scaling factor analogous to $m_\texttt{ss}$ and $b_\texttt{ss}$ in \eqref{eq:m_ss}. The following Lemma shows that $\hat{\svbx}^\texttt{mmrc}$ is an unbiased estimator. See Appendix \ref{appendix:mmrc_ss_debias} for a proof.

\begin{restatable}{lemma}{mmrcssbias}\label{lemma:mmrc_ss_bias}
Let $\hat{\svbx}^\texttt{mmrc}$ be the estimator of the \emph{$\MMRC$} mechanism simulating \emph{$\SubsetSelection$} as defined above. Then, $\Expectation[\hat{\svbx}^\texttt{mmrc}] = \svbx$.
\end{restatable}

\subsection{Simulating \texorpdfstring{$\SubsetSelection$}{Subset Selection} using \texorpdfstring{$\MMRC$}{MMRC}}
\label{subsec:sim_ss1}

Finally, we consider estimating the empirical frequency $\Pi$ defined earlier using the $\MMRC$ scheme simulating $\SubsetSelection$. To that end, consider $n$ users and let $\hat{\svbx}^{\texttt{mmrc}}_i$ be the unbiased estimator of $\svbx_i$ at the $i^{th}$ user. Let the (unbiased) estimate of $\Pi$ be $\hat{\Pi}^{ \texttt{mmrc}} \coloneqq \frac{1}{n}\sum_i \hat{\svbx}^\texttt{mmrc}_i$. 
The following Theorem shows that, for frequency estimation, $\MMRC$ can simulate $\SubsetSelection$ in a near-lossless manner (when $\lambda$ is small) while only using on the order of $\varepsilon$ bits of communication. A proof can be found in Appendix \ref{appendix:mmrc_ss_ut}. Similar to $\PrivUnit$, the key idea in the proof is to show that when the number of candidates $N$ is exponential in $\varepsilon$, the scaling factor $m_{\texttt{mmrc}}$ is close to the scaling parameter associated with $\qSS$ (i.e., $m_{\texttt{ss}}$ defined in \eqref{eq:m_ss}).


\begin{restatable}{theorem}{mmrcss}
\label{thm:fe_mmrc_ss}
 Let $r_{\msf{FE}} \lp \hat{\Pi}^\texttt{ss}, \qSS \rp$ and $r_{\msf{FE}} \lp \hat{\Pi}^\texttt{mmrc}, \qMMRC \rp$ be the empirical mean estimation error for \emph{$\SubsetSelection$} and \emph{$\MMRC$} simulating \emph{$\SubsetSelection$} with $N$ candidates respectively. Consider any $\lambda > 0$. Then
 \begin{equation}
     r_{\msf{FE}} \lp \hat{\Pi}^{ \texttt{mmrc}}, \qMMRC \rp \leq  
     \lp 1+4\lambda+5\lambda^2+2\lambda^3 \rp r_{\msf{FE}} \lp \hat{\Pi}^{ \texttt{ss}}, \qSS \rp,
 \end{equation}
 as long as 
  \begin{align}\label{eq:N_bdd_mmrc_ss}
      N\geq \frac{2(e^{\varepsilon}+1)^2(1+\lambda)^2}{0.24^2\lambda^2}\ln\lp \frac{8(1+\lambda)}{0.24\lambda}\rp.
  \end{align}
 \end{restatable}

Similar to mean estimation, while a specific value of $\lambda$ can be chosen in \eqref{eq:N_bdd_mmrc_ss}, in practice, the number of bits could be fixed (see Section \ref{subsec:mmrc_ss_empirical}), determining the value of $\lambda$.

\subsection{Empirical Comparisons.}
\label{subsec:mmrc_ss_empirical}
Next, we empirically demonstrate the privacy-accuracy-communication tradeoffs of $\MMRC$ simulating $\SubsetSelection$. Along with $\SubsetSelection$, we compare against the RHR algorithm of \cite{CKO20} which offers order-optimal privacy-accuracy tradeoffs while requiring only $\varepsilon$ bits. Following \cite{acharya2019hadamard}, we generate $\svbx_1, \cdots, \svbx_n$ from the Zipf distribution with degree 1. We report the average $\ell_2$ estimation error over 10 runs. See more variations in 
Appendix \ref{appendix:mmrc_ss_emp}.

\begin{figure}[h]
\centering
\includegraphics[width=0.45\linewidth]{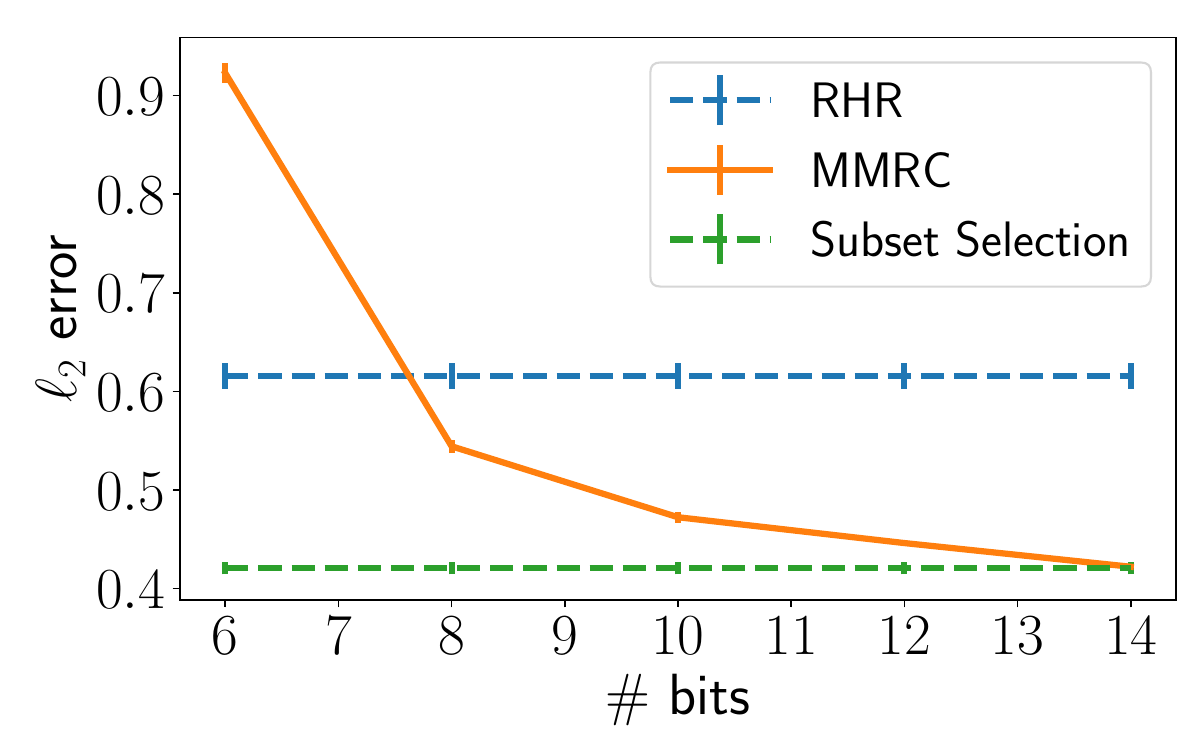} \qquad
\includegraphics[width=0.45\linewidth]{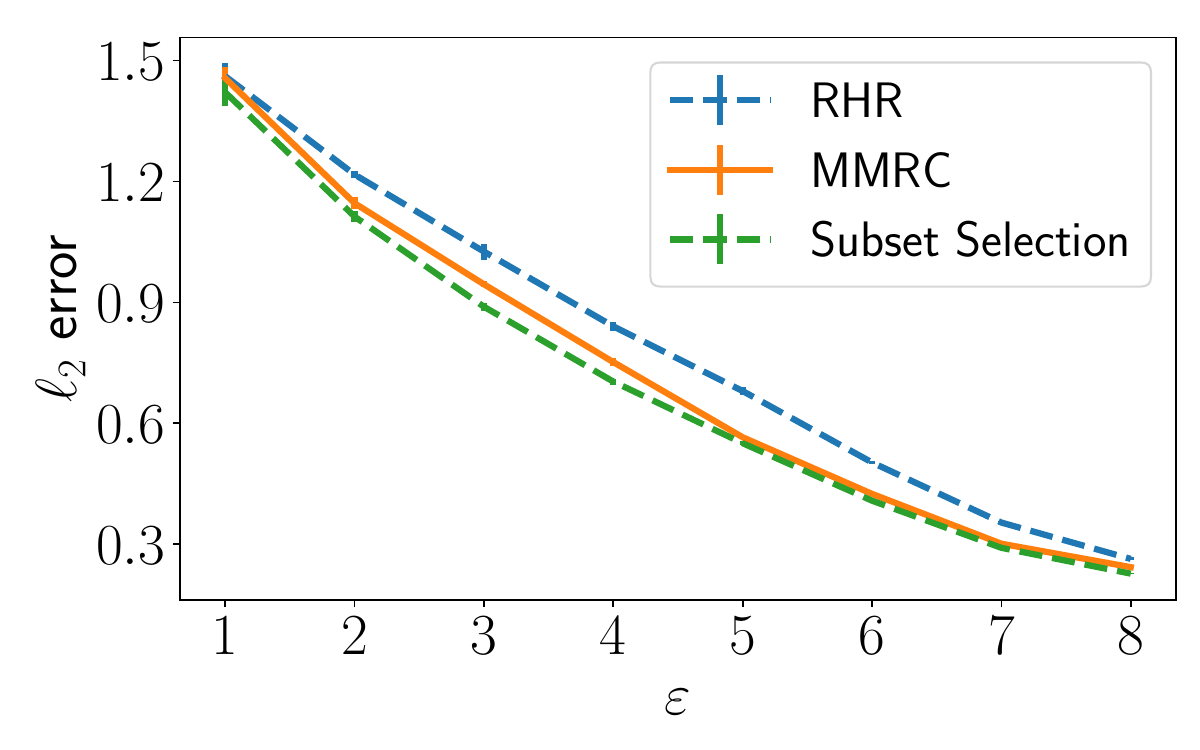}
\caption{Comparing $\SubsetSelection$, $\MMRC$ simulating $\SubsetSelection$ and RHR for frequency estimation with $d = 500$ and $n = 5000$. \textbf{Left:} $\ell_2$ error vs $\#$bits for $\varepsilon = 6$. \textbf{Right:} $\ell_2$ error vs $\varepsilon$ for $\#$bits $= \max\{\lceil( \varepsilon/ \ln 2)\rceil + 3, 8\}$. RHR uses $\#$-bits $= \varepsilon$ for both as it leads to a poor performance if $\#$-bits $ > \varepsilon$.}
\label{fig:freq}
\end{figure}

In Figure \ref{fig:freq} (Left), we show the communication-accuracy tradeoffs. We see that with correct order of bits, the accuracy of $\MMRC$ simulating $\SubsetSelection$ converges to the accuracy of the uncompressed $\SubsetSelection$. In Figure \ref{fig:freq} (Right), we show the privacy-accuracy tradeoffs. More specifically, $\MMRC$ simulating $\SubsetSelection$ can attain the accuracy of the uncompressed $\SubsetSelection$ for the range of  $\varepsilon$'s typically considered by LDP mechanisms while only using $\max\{\lceil( \varepsilon/ \ln 2)\rceil + 3, 8\}$ bits. 
\section{Conclusion and Future Work}
\label{sec:conclusion}
We demonstrated how Minimal Random Coding can be used to simulate a class of $\varepsilon$-LDP mechanisms in a manner which is communication efficient while preserving accuracy and differential privacy guarantees. Further, for mean and frequency estimation, we proposed unbiased versions of our schemes (relying only on translation and scaling) that attain the privacy-accuracy tradeoffs of the best known schemes i.e., $\PrivUnit$ and $\SubsetSelection$, while requiring on the order of $\varepsilon$ bits of communication.

We now briefly discuss a few non-trivial and interesting open questions. \\

\noindent\textbf{Computational Cost.} The computational cost of our approach, similar to \cite{FT21} grows linearly in $d$ and exponentially in $\varepsilon$ (as we need $N = \exp(O(\varepsilon))$ candidates to properly simulate the optimal mechanisms). An important question for future research is therefore how to increase the computational efficiency of $\MRC$ and $\MMRC$ with respect to $\varepsilon$.\\

\noindent\textbf{Privacy Amplification via Shuffling.} As mentioned in Section \ref{subsec:related_works}, privacy amplification via shuffling techniques ensure a central $\varepsilon \approx 1$ even when the local $\varepsilon >1$. While our method could be combined with these amplification techniques in principle, we leave the analysis of the privacy, utility, and communication guarantees of the resulting scheme as a question for future research.\\

\noindent\textbf{Other schemes to simulate noisy channels.} $\MRC$ is only one of several channel simulation schemes studied in information theory which could be considered for compression of $\varepsilon$-LDP mechanisms. 
Similar to $\MRC$, other channel simulation schemes, e.g., rejection sampling \citep{HJMR07} or schemes based on the Poisson functional representation \citep{LE18}, can also compress noisy signals to a number of bits which is close to the information contained in the signal (which decreases as noise increases). 
Analyzing these schemes for their effect on differential privacy guarantees is an interesting open question.\\

\noindent\textbf{Shared Randomness.} Finally, here we assumed the existence of a shared source of randomness. We further assumed that each user is using a different source of shared randomness. While shared randomnesss only adds to the cost of downlink and not uplink communication (which is usually the bottleneck in settings like federated learning), a question left for future research is how much communication is required to establish and select these sources of randomness.

\section*{Acknolwedgements}

We thank the anonymous reviewers of AISTATS 2022 for their comments and suggestions.
We sincerely thank Jakub Kone\v{c}n\'{y} and Wennan Zhu for helpful discussions. We also thank Zachary Charles for support with the software infrastructure.
\bibliographystyle{abbrvnat}
\bibliography{references}

\clearpage
\appendix
\section*{Appendix}
\textbf{Organization.}
The Appendix is organized as follows. In Appendix \ref{appendix:societal_impacts}, we discuss the societal impact associated with our work. In Appendix \ref{appendix:mrc}, we focus on $\MRC$ and provide the proofs of Theorem \ref{theorem:mrc_accuracy_local}, Theorem \ref{theorem:mrc_pure_privacy}, and Theorem \ref{theorem:mrc_approximate_privacy}. In Appendix \ref{appendix:mmrc}, we focus on $\MMRC$ and provide the proofs of Theorem \ref{theorem:mmrc_privacy} and Theorem \ref{theorem:mmrc_accuracy_local_wrt_mrc}. Further, we also provide Theorem \ref{theorem:mmrc_accuracy_local} where we show that $\MMRC$ can simulate any $\varepsilon$-LDP cap-based mechanism in a nearly lossless fashion with about $\varepsilon$ bits of communication. In Appendix \ref{appendix:privunit}, we provide additional preliminary on $\PrivUnit$ and also show that $\PrivUnit$ is a cap-based mechanism (Definition \ref{def:cap}). In Appendix \ref{appendix:mrc_pu}, we show how $\PrivUnit$ can be simulated using $\MRC$ analogous to how we simulated $\PrivUnit$ using $\MMRC$ in Section \ref{sec:mean_estimation}. Along with the theoretical guarantees, we also provide some empirical comparisons between $\MRC$ simulating $\PrivUnit$ and $\PrivUnit$. In Appendix \ref{appendix:mmrc_pu}, we provide the proofs of Lemma \ref{lemma:mmrc_privunit_bias} and Theorem \ref{thm:me_mmrc_pu} as well as some additional empirical comparisons between $\MMRC$ simulating $\PrivUnit$ and $\PrivUnit$.
In Appendix \ref{appendix:ss}, we provide additional preliminary on $\SubsetSelection$ and also show that $\SubsetSelection$ is a cap-based mechanism (Definition \ref{def:cap}). 
In Appendix \ref{appendix:mrc_ss}, we show how $\SubsetSelection$ can be simulated using $\MRC$ analogous to how we simulated $\SubsetSelection$ using $\MMRC$ in Section \ref{sec:frequency_estimation}. Along with the theoretical guarantees, we also provide some empirical comparisons between $\MRC$ simulating $\SubsetSelection$ and $\SubsetSelection$. 
In Appendix \ref{sec:mmrc_ss_app}, we provide the proofs of Lemma \ref{lemma:mmrc_ss_bias} and Theorem \ref{thm:fe_mmrc_ss} as well as some additional empirical comparisons between $\MMRC$ simulating $\SubsetSelection$ and $\SubsetSelection$.

\section{Societal impact}
\label{appendix:societal_impacts}

Collecting large datasets allows us to build better machine learning models which can facilitate our lives in many different ways. However, harnessing data from devices can expose their users to privacy risks. Research into differential privacy can help to minimize these risks. At present, our work is mostly theoretical in nature as there are a few unsolved questions. In particular, for large $\varepsilon$ the computational complexity of our approach may be too expensive to be practical.

\section{Minimal Random Coding}
\label{appendix:mrc}

Let $q(\svbz|\svbx)$ be an $\varepsilon$-LDP mechanism for all $\svbx \in \cX$ and $\svbz \in \cZ$. Let $p(\svbz)$ be the fixed reference distribution over $\cZ$ and let $\{\svbz_k\}_{k=1}^{N}$ be $N$ candidates drawn from $p(\svbz)$. From Algorithm \ref{alg:mrc}, the distribution over the indices $k \in [N]$ under minimal random coding $(\MRC)$ is as follows:
\begin{align}
  \piMRC_{\svbx}(k) \coloneqq \dfrac{q(\svbz_k|\svbx)/p(\svbz_k)}{\sum_{k'} q(\svbz_{k'}|\svbx)/p(\svbz_{k'})}  \label{eq:generic_pi} 
\end{align}
$\piMRC$ can be viewed as a function that maps $\svbx$ and $(\svbz_1,...,\svbz_N)$ to a distribution in $[N]$. However for notational convenience, when the context is clear, we will omit the dependence on $\svbx$ and $(\svbz_1,...,\svbz_N)$.

Let $\qMRC$ denote the distribution of $\svbz_K$ where $K \sim \piMRC(\cdot)$ i.e., with $\delta(\cdot)$ denoting the Dirac delta function:
\begin{align}
    \qMRC(\svbz | \svbx) \coloneqq \sum_{k} \piMRC(k) \delta(\svbz - \svbz_k). \label{eq:qmrc}
\end{align}

\subsection{Utility of \texorpdfstring{$\MRC$}{MRC}}\label{appendix:utility_mrc}
In this section, we prove Theorem \ref{theorem:mrc_accuracy_local} i.e., we show that $\MRC$ can simulate any $\varepsilon$-LDP mechanism in a nearly lossless fashion with about $\varepsilon$ bits of communication
\mrcaccuracylocal*
\begin{proof}
In order to prove this theorem, we invoke Theorem~3.2 of \citet{HPHJ19}.

Recall Theorem~3.2 \citep{HPHJ19}: Let $q'$ and $p$ be distributions over $\cZ$. Let $t \geq 0$ be some constant and let $N' = 2^{\lp \KLD{q'(\svbz)}{p(\svbz)} + t \rp}$.
Let $\tilde{q}$ be a discrete distribution constructed by drawing $N'$ samples $\{\svbz_k\}_{k=1}^{N'}$ from $p$ and defining \begin{align}
    \tilde{q}(\svbz) \coloneqq \sum_{k = 1}^{N'} \dfrac{q(\svbz_k)/p(\svbz_k)}{\sum_{k'}  q(\svbz_{k'})/p(\svbz_{k'})}\delta(\svbz - \svbz_k).
\end{align}
Furthermore, let $f$ be a measurable function and $\|f\|_{q'} = \sqrt{\Expectation_{q'(\svbz)}[f^2(\svbz)]}$ be its 2-norm under $q'$. Then it holds that 
\begin{align}
    \Probability\bigg(\Big|\Expectation_{\tilde{q}(\svbz)}[f(\svbz)] - \Expectation_{q'(\svbz)}[f(\svbz)] \Big| \geq \dfrac{2\|f\|_{q'} \alpha'}{1-\alpha'}\bigg) \leq 2\alpha'
\end{align}
where 
\begin{align}
\alpha' = \sqrt{2^{-t/4} + 2\sqrt{\Probability( \log(q'(\svbz)/p(\svbz)) > \KLD{q'(\svbz)}{p(\svbz)} + t/2)}}.
\end{align}

We apply Theorem~3.2 \citep{HPHJ19} to $q'(\svbz) \coloneqq q(\svbz|\svbx)$ and $f(\svbz) \coloneqq \|\svbz-\svbx\|^2$. We identify $\tilde{q}(\svbz) = \qMRC(\svbz|\svbx)$ and $N' = N$. To prove Theorem~\ref{theorem:mrc_accuracy_local}, it suffices to show that $\alpha \geq \alpha'$. Note that
\begin{align}
\KLD{q(\svbz|\svbx)}{p(\svbz)} \stackrel{(a)}{=} \Expectation_{q(\svbz|\svbx)}\bigg[\log\bigg(\dfrac{q(\svbz|\svbx)}{p(\svbz)}\bigg)\bigg] \stackrel{(b)}{\leq} \varepsilon \log e,
\end{align}
where $(a)$ follows the definition of KL-divergence and $(b)$ follows since $ |\log (q(\svbz|\svbx) / p(\svbz)) | \leq \varepsilon \log e$ by the assumption on $p$. We therefore have
\begin{align}
    t
    = (\log e + 4c)\varepsilon - \KLD{q(\svbz|\svbx)}{p(\svbz)}
    \geq 4c\varepsilon.
\end{align}
 It follows that
\begin{align}
    & \Probability \big(\log (q(\svbz|\svbx)\|p(\svbz)) > \KLD{q(\svbz|\svbx)}{p(\svbz)} + t/2\big)  \\
    & \qquad \leq
    \Probability \big(\log (q(\svbz|\svbx)\|p(\svbz)) > \Expectation\big[\log (q(\svbz|\svbx)\|p(\svbz))\big] + 2c\varepsilon\big)\\
    & \qquad \stackrel{(b)}{\leq} \exp(-2c^2/(\log e)^2)\\
    & \qquad = 2^{-2c^2/\log e}.
\end{align}
where $(b)$ follows from Hoeffding's inequality since $ |\log (q(\svbz|\svbx) / p(\svbz)) | \leq \varepsilon \log e$ by the assumption on $p$. Therefore,
\begin{align}
    \alpha'
    \leq \sqrt{2^{-c\varepsilon} + 2\sqrt{2^{-2c^2/\log e}}}
    = \sqrt{2^{-c\varepsilon} + 2^{-c^2/\log e + 1}} = \alpha.
\end{align}
\end{proof}

\begin{remark}\label{rmk:error_forth_moment}
For most $\varepsilon$-LDP mechanisms $q(\cdot|\svbx)$, the term $\Expectation_{q}\lb \|\svbz - \svbx\|^4\rb$ in \eqref{eq:mrc_utility} can be well-controlled. For instance, for \emph{$\SubsetSelection$} and \emph{$\PrivUnit$}, the output spaces are bounded, and therefore, $\sqrt{\Expectation_{q}[\|\svbz - \svbx\|^4]}$ is of the same order as $\Expectation_{q}\big[\|\svbz - \svbx\|^2\big]$. Therefore, by making $\alpha$ small enough (in Theorem~\ref{theorem:mrc_accuracy_local}) i.e. by increasing $c$, the estimation error of \emph{$\MRC$} can be arbitrarily close to the estimation error of the underlying scheme it is simulating.
\end{remark}

\subsection{Privacy of \texorpdfstring{$\MRC$}{MRC}}
\label{appendix:privacy_mrc}
\subsubsection{Pure Privacy of \texorpdfstring{$\MRC$}{MRC}}\label{appendix:pure_privacy_mrc}
In this section, we prove Theorem \ref{theorem:mrc_pure_privacy} i.e., we show that $\piMRC$ is a $2\varepsilon$-LDP mechanism.
\mrcpureprivacy*
\begin{proof}
For any $\svbx, \svbx' \in \cX, \svbz \in \cZ$, using the definition of an $\varepsilon$-LDP mechanism, we have
\begin{align}
    q(\svbz|\svbx) \leq \exp(\varepsilon) q(\svbz|\svbx'). \label{eq:ldp_definition}
\end{align}
For any $\svbx, \svbx' \in \cX, \{\svbz_k\}_{k=1}^{N} \in \cZ^N$ and $k \in [N]$, we have
\begin{align}
    \frac{\piMRC_{\svbx}(k)}{\piMRC_{\svbx'}(k)} & \stackrel{(a)}{=} \frac{q(\svbz_k | \svbx)}{q(\svbz_k | \svbx')} \times \frac{\sum_{k'} q(\svbz_{k'} | \svbx')/p(\svbz_{k'})}{\sum_{k'} q(\svbz_{k'} | \svbx)/p(\svbz_{k'})} \\
    & \stackrel{(b)}{\leq} \exp(\varepsilon) \times \frac{\sum_{k'} \exp(\varepsilon) q(\svbz_{k'} | \svbx)/p(\svbz_{k'})}{\sum_{k'} q(\svbz_{k'} | \svbx)/p(\svbz_{k'})} \\
    & = \exp(\varepsilon) \times \frac{\exp(\varepsilon) \sum_{k'}  q(\svbz_{k'} | \svbx)/p(\svbz_{k'})}{\sum_{k'} q(\svbz_{k'} | \svbx)/p(\svbz_{k'})} \\
    & = \exp(2\varepsilon).
\end{align}
where $(a)$ follows from the definition of $\piMRC$ and $(b)$ follows from \eqref{eq:ldp_definition}.
\end{proof}

\subsubsection{Approximate Privacy of \texorpdfstring{$\MRC$}{MRC}}\label{appendix:approx_privacy_mrc}
In this section, we prove Theorem \ref{theorem:mrc_approximate_privacy} i.e., we provide the approximate DP guarantee of $\piMRC$.
\mrcapproxprivacy*
\begin{proof}
Fix any $\svbx \in \cX$. Let us define the following random variable:
\begin{align}
w(\svbz|\svbx) = q(\svbz|\svbx) / p(\svbz). \label{eq:importance_weights}
\end{align}
Assuming $\svbz \sim p(\cdot)$, the expected value of the random variable $w(\svbz|\svbx)$ is 
\begin{align}
\Expectation_{p}[w(\svbz|\svbx)] = \Expectation_{p}[q(\svbz|\svbx) / p(\svbz)] = \int_{\svbz \in \cZ} q(\svbz|\svbx) = 1.
\end{align}
Further, the random variable $w(\svbz|\svbx)$ can be bounded as follows:
\begin{align}
| w(\svbz|\svbx) | = | q(\svbz|\svbx) / p(\svbz) | \stackrel{(a)}{\leq} \exp(\varepsilon).
\end{align}
where $(a)$ follows from the assumption on $p(\cdot)$. Therefore, we have
\begin{align}
    \Probability\bigg( \bigg| \frac{1}{N}\sum_{k=1}^{N} w(\svbz_k|\svbx) - 1 \bigg| \geq a_0 \bigg) & \stackrel{(a)}{\leq}  2 \exp\bigg(\frac{-2Na_0^2}{(\exp(\varepsilon)-\exp(-\varepsilon))^2}\bigg) \\
    & \leq 2 \exp\bigg(\frac{-2Na_0^2}{\exp(2\varepsilon)}\bigg) \stackrel{(b)}{=} \delta \label{eq:hoeffdings_inequality}
\end{align}
where $(a)$ follows from Hoeffding's inequality and $(b)$ follows from the definition of $a_0$ and $N$. Now, for any $\svbx, \svbx' \in \cX, \{\svbz_k\}_{k=1}^{N} \in \cZ^N$ and $k \in [N]$, we have
\begin{align}
    \frac{\piMRC_{\svbx}(k)}{\piMRC_{\svbx'}(k)} & 
    \stackrel{(a)}{=} \frac{q(\svbz_k | \svbx)}{q(\svbz_k | \svbx')} \times \frac{\sum_{k'} q(\svbz_{k'} | \svbx')/p(\svbz_{k'})}{\sum_{k'} q(\svbz_{k'} | \svbx)/p(\svbz_{k'})} \\
    & \stackrel{(b)}{=} \frac{q(\svbz_k | \svbx)}{q(\svbz_k | \svbx')} \times \frac{\sum_{k'} w(\svbz_{k'}|\svbx')}{\sum_{k'} w(\svbz_{k'}|\svbx)} \\
    & \stackrel{(c)}{\leq} \exp(\varepsilon) \times \frac{\frac{1}{N}\sum_{k'} w(\svbz_{k'}|\svbx')}{\frac{1}{N}\sum_{k'} w(\svbz_{k'}|\svbx)} \label{eq:approx_dp_intermediate}
\end{align}
where $(a)$ follows from the definition of $\piMRC$, $(b)$ follows from \eqref{eq:importance_weights}
and $(c)$ follows from \eqref{eq:ldp_definition}. Now, using \eqref{eq:hoeffdings_inequality} in \eqref{eq:approx_dp_intermediate}, we have with probability at least $1-\delta$:
\begin{align}
    \frac{\piMRC_{\svbx}(k)}{\piMRC_{\svbx'}(k)} \leq \exp(\varepsilon) \times \frac{1+a_0}{1-a_0} \stackrel{(a)}{=} \exp(\varepsilon + \varepsilon_0)
\end{align}
 where $(a)$ follows from the definition of $\varepsilon_0$.
\end{proof}

\section{Modified Minimal Random Coding}
\label{appendix:mmrc}

Let $q(\svbz|\svbx)$ be an $\varepsilon$-LDP cap-based mechanism (see definition \ref{def:cap}) for all $\svbx \in \cX$ and $\svbz \in \cZ$. Let $p(\svbz)$ be the uniform distribution over $\cZ$ and let $\{\svbz_k\}_{k=1}^{N}$ be $N$ candidates drawn from $p(\svbz)$. Let $\theta$ denote the fraction of candidates inside the $\msf{Cap}_{\svbx}$ associated with $q(\svbz|\svbx)$. Let $\piMMRC$ be the distribution over the indices $k \in [N]$ under modified minimal random coding $(\MMRC)$ obtained from Algorithm \ref{alg:mmrc}.
Recall that $\piMMRC(k)$ is bounded by an upper threshold $t_u$ and a lower threshold $t_l$ (Section~\ref{subsec:mmrc}),
\begin{align}
    t_u &= \frac{1}{N} \times \frac{c_1(\varepsilon, d)}{\Expectation[\theta] c_1(\varepsilon, d) + (1 - \Expectation[\theta]) c_2(\varepsilon, d)}, & 
    t_l &= \frac{1}{N} \times \frac{c_2(\varepsilon, d)}{\Expectation[\theta] c_1(\varepsilon, d) + (1 - \Expectation[\theta]) c_2(\varepsilon, d)}.
\end{align}

Similar to $\piMRC$, $\piMMRC$ can be be viewed as a function that maps $\svbx$ and $(\svbz_1,...,\svbz_N)$ to a distribution in $[N]$. However, to reduce clutter, we will generally omit the dependence on $\svbx$ and $(\svbz_1,...,\svbz_N)$. Further, since  $\piMMRC$ depends on $(\svbz_1,...,\svbz_N)$ only through $\theta$, we will sometimes show this dependence as $\piMMRC_{\svbx, \theta}$. 

Finally, let $\qMMRC$ denote the distribution of $\svbz_K$ where $K \sim \piMMRC$.
That is, with $\delta$ denoting the Dirac delta function:
\begin{align}
    \qMMRC(\svbz | \svbx) \coloneqq \sum_{k} \piMMRC(k) \delta(\svbz - \svbz_k). \label{eq:qmmrc}
\end{align}

\subsection{Privacy of \texorpdfstring{$\MMRC$}{MMRC}}\label{appendix:privacy_mmrc}
In this section, we prove Theorem \ref{theorem:mmrc_privacy} i.e., we show that $\piMMRC$ is a $\varepsilon$-LDP mechanism.

\mmrcprivacy*
\begin{proof}
For any $\varepsilon$-LDP cap-based $q(\cdot | \svbx)$, we have the following from \eqref{eq:ldp} and \eqref{eq:cap_mechanism}:
\begin{align}
 \frac{c_1(\varepsilon,d)}{c_2(\varepsilon,d)} \leq \exp(\varepsilon). \label{eq:cap-based-bound}
\end{align}
Further, the modification of $\piMRC$ to $\piMMRC$ ensures that \eqref{eq:modification} is true, that is,
\begin{align}
    t_l \leq \piMMRC(k) \leq t_u ~ \forall k \in [N].
\end{align}
Therefore, for any $\svbx, \svbx' \in \cX$ and $k \in [N]$, we have
\begin{align}
    \frac{\piMMRC_{\svbx}(k)}{\piMMRC_{\svbx'}(k)} & \leq \frac{t_u}{t_l} \stackrel{(a)}{=} \frac{c_1(\varepsilon,d)}{c_2(\varepsilon,d)} \stackrel{(b)}{\leq} \exp(\varepsilon),
\end{align}
where $(a)$ follows from the definitions of $t_u$ and $t_l$ and $(b)$ follows from \eqref{eq:cap-based-bound}.
\end{proof}

\subsection{Supporting Lemmas to prove the utility of \texorpdfstring{$\MMRC$}{MMRC}}
\label{appendix:supporting_lemma_mmrc_utility}

To prove Theorem~\ref{theorem:mmrc_accuracy_local_wrt_mrc} (Section \ref{appendix:utility_mmrc}), we prove that the expected KL divergence between $\piMRC$ and $\piMMRC$ can be controlled arbitrarily when the number of candidates is of the right order (Lemma \ref{lemma:expected_kl}).
To prove Lemma~\ref{lemma:expected_kl}, we first show that the KL divergence between $\piMRC$ and $\piMMRC$, for a given fraction of candidates inside the $\msf{Cap}_{\svbx}$, can be bounded in terms of $\varepsilon$ (Lemma \ref{lemma:kl}).

\subsubsection{The KL divergence between \texorpdfstring{$\piMRC$}{} and \texorpdfstring{$\piMMRC$}{} is small}
\begin{lemma}\label{lemma:kl}
Let $q(\svbz|\svbx)$ be an $\varepsilon$-LDP cap-based mechanism. Let $p(\svbz)$ be the uniform distribution over $\cZ$ and let $\{\svbz_k\}_{k=1}^{N}$ be $N$ candidates drawn from $p(\svbz)$. Let $\theta$ denote the fraction of candidates inside the $\msf{Cap}_{\svbx}$ associated with $q(\svbz|\svbx)$. Let $\piMRC$ be the distribution over the indices $k \in [N]$ under $\MRC$ obtained from Algorithm \ref{alg:mrc} and $\piMMRC$ be the distribution over the indices $k \in [N]$ under $\MMRC$ obtained from Algorithm \ref{alg:mmrc}. Then,
\begin{align}
    \KLD{\piMRC_{\svbx, \theta}(\cdot)}{\piMMRC_{\svbx, \theta}(\cdot)} \leq \varepsilon \log e
\end{align}
\end{lemma}
\begin{proof}
We consider three different cases depending on whether $\theta = \Expectation[\theta]$, $\theta < \Expectation[\theta]$ or $\theta > \Expectation[\theta]$.
\begin{enumerate}
    \item[1.]  For $\theta = \Expectation[\theta]$, we have $\piMRC_{\svbx, \theta}(\cdot) = \piMMRC_{\svbx, \theta}(\cdot)$. Therefore,
    \begin{align}
    \KLD{\piMRC_{\svbx, \theta}(\cdot)}{\piMMRC_{\svbx, \theta}(\cdot)} = \KLD{\piMRC_{\svbx, \theta}(\cdot)}{\piMRC_{\svbx, \theta}(\cdot)} = 0 \leq \varepsilon \log e. \label{eq:KL_0}
    \end{align}
    \item[2.]
If $\theta < \Expectation[\theta]$, then $\piMRC$ violates the upper threshold $t_u$ so that $\piMMRC(k) = t_u$ for all $k \in \msf{Cap}_{\svbx}$ and we have
\begin{align}
    & \KLD{\piMRC_{\svbx, \theta}(\cdot)}{\piMMRC_{\svbx, \theta}(\cdot)} = \sum_{k} \piMRC_{\svbx, \theta}(k) \log \dfrac{\piMRC_{\svbx, \theta}(k)}{\piMMRC_{\svbx, \theta}(k)} \\
    & \stackrel{(a)}{=} \sum_{k \in \msf{Cap}_{\svbx}} \frac{1}{N} \times \frac{ c_1(\varepsilon,d)}{c_2(\varepsilon,d) +  \theta (c_1(\varepsilon,d) - c_2(\varepsilon,d))} \log \dfrac{c_2(\varepsilon,d) + \Expectation[\theta] \times (c_1(\varepsilon,d) - c_2(\varepsilon,d))}{c_2(\varepsilon,d) +  \theta \times (c_1(\varepsilon,d) - c_2(\varepsilon,d))}\\ 
    &  \qquad  + \sum_{k \notin \msf{Cap}_{\svbx}} \frac{1}{N} \times \dfrac{c_2(\varepsilon,d)}{c_2(\varepsilon,d) +  \theta (c_1(\varepsilon,d) - c_2(\varepsilon,d))} \bigg[ \log \dfrac{c_2(\varepsilon,d)+ \Expectation[\theta] \times (c_1(\varepsilon,d) - c_2(\varepsilon,d))}{c_2(\varepsilon,d) +  \theta \times (c_1(\varepsilon,d) - c_2(\varepsilon,d))} \\
    &  \qquad \qquad + \log \dfrac{(1-\theta) \times c_2(\varepsilon,d)}{ (1-\Expectation[\theta]) \times c_2(\varepsilon,d) + (\Expectation[\theta] - \theta) c_1(\varepsilon,d)} \bigg] \\
    & \stackrel{(b)}{=} \dfrac{\theta c_1(\varepsilon,d)}{c_2(\varepsilon,d) +  \theta (c_1(\varepsilon,d) - c_2(\varepsilon,d))} \log \dfrac{c_2(\varepsilon,d) + \Expectation[\theta] \times (c_1(\varepsilon,d) - c_2(\varepsilon,d))}{c_2(\varepsilon,d) +  \theta \times (c_1(\varepsilon,d) - c_2(\varepsilon,d))} \\ 
    & \qquad + \dfrac{(1-\theta)c_2(\varepsilon,d)}{c_2(\varepsilon,d) +  \theta (c_1(\varepsilon,d) - c_2(\varepsilon,d))} \bigg[ \log \dfrac{c_2(\varepsilon,d)+ \Expectation[\theta] \times (c_1(\varepsilon,d) - c_2(\varepsilon,d))}{c_2(\varepsilon,d) +  \theta \times (c_1(\varepsilon,d) - c_2(\varepsilon,d))} \\
    &  \qquad \qquad + \log \dfrac{(1-\theta) \times c_2(\varepsilon,d)}{ (1-\Expectation[\theta]) \times c_2(\varepsilon,d) + (\Expectation[\theta] - \theta) c_1(\varepsilon,d)} \bigg]  \\
    & = \log \dfrac{c_2(\varepsilon,d) + \Expectation[\theta] \times (c_1(\varepsilon,d) - c_2(\varepsilon,d))}{c_2(\varepsilon,d) +  \theta \times (c_1(\varepsilon,d) - c_2(\varepsilon,d))} \\
    &  \qquad + \dfrac{(1-\theta)c_2(\varepsilon,d)}{c_2(\varepsilon,d) +  \theta (c_1(\varepsilon,d) - c_2(\varepsilon,d))} \bigg[ \log \dfrac{(1-\theta) \times c_2(\varepsilon,d)}{ (1-\Expectation[\theta]) \times c_2(\varepsilon,d) + (\Expectation[\theta] - \theta) c_1(\varepsilon,d)} \bigg]  \\
    & \stackrel{(c)}{\leq} \log \dfrac{c_2(\varepsilon,d) + \Expectation[\theta] \times (c_1(\varepsilon,d) - c_2(\varepsilon,d))}{c_2(\varepsilon,d) +  \theta \times (c_1(\varepsilon,d) - c_2(\varepsilon,d))} \\ & \stackrel{(d)}{\leq}  \log \bigg( \frac{c_2(\varepsilon,d) + \Expectation[\theta] \times (c_1(\varepsilon,d) - c_2(\varepsilon,d))}{c_2(\varepsilon,d)} \bigg) \label{eq:KL_1} \\
    & \stackrel{(e)}{\leq} \log \frac{c_1(\varepsilon,d)}{c_2(\varepsilon,d)}  \stackrel{(f)}{\leq} \varepsilon \log e
\end{align}
where $(a)$ follows from the definition of $\piMRC_{\svbx, \theta}(k)$ and $\piMMRC_{\svbx, \theta}(k)$, $(b)$ follows because $|\{k : k \in \msf{Cap}_{\svbx}\}| = \theta N$ and $|\{k : k \notin \msf{Cap}_{\svbx}\}| = (1-\theta) N$, $(c)$ follows because $\log \frac{(1-\theta) \times c_2(\varepsilon,d)}{ (1-\Expectation[\theta]) \times c_2(\varepsilon,d) + (\Expectation[\theta] - \theta) c_1(\varepsilon,d)} \leq 0$, $(d)$ follows because $\theta \geq 0$, $(e)$ follows because $\Expectation[\theta] \leq 1$, and $(f)$ follows because $c_1(\varepsilon,d) / c_2(\varepsilon,d) \leq \exp{(\varepsilon)}$.

\item[3.]
For $\theta > \Expectation[\theta]$, we have
\begin{align}
  & \KLD{\piMRC_{\svbx, \theta}(\cdot)}{\piMMRC_{\svbx, \theta}(\cdot)} = \sum_{\svbz_i} \piMRC_{\svbx, \theta}(k) \log \dfrac{\piMRC_{\svbx, \theta}(k)}{\piMMRC_{\svbx, \theta}(k)} \\
  & \stackrel{(a)}{=} \sum_{k \notin \msf{Cap}_{\svbx}} \frac{1}{N} \times  \dfrac{c_2(\varepsilon,d)}{c_2(\varepsilon,d) +  \theta (c_1(\varepsilon,d) - c_2(\varepsilon,d))} \log \dfrac{c_2(\varepsilon,d) + \Expectation[\theta] \times (c_1(\varepsilon,d) - c_2(\varepsilon,d))}{c_2(\varepsilon,d) +  \theta \times (c_1(\varepsilon,d) - c_2(\varepsilon,d))} \\ 
    & \qquad +  \sum_{k \in \msf{Cap}_{\svbx}} \frac{1}{N} \times  \dfrac{c_1(\varepsilon,d)}{c_2(\varepsilon,d) +  \theta (c_1(\varepsilon,d) - c_2(\varepsilon,d))} \bigg[ \log \dfrac{c_2(\varepsilon,d)+ \Expectation[\theta] \times (c_1(\varepsilon,d) - c_2(\varepsilon,d))}{c_2(\varepsilon,d) +  \theta \times (c_1(\varepsilon,d) - c_2(\varepsilon,d))} \\
    &  \qquad \qquad + \log \dfrac{\theta c_1(\varepsilon,d)}{\Expectation[\theta] c_1(\varepsilon,d) + (\theta - \Expectation[\theta]) \times c_2(\varepsilon,d)} \bigg] \\
    & \stackrel{(b)}{=}   \dfrac{(1-\theta) \times c_2(\varepsilon,d)}{c_2(\varepsilon,d) +  \theta (c_1(\varepsilon,d) - c_2(\varepsilon,d))} \log \dfrac{c_2(\varepsilon,d) + \Expectation[\theta] \times (c_1(\varepsilon,d) - c_2(\varepsilon,d))}{c_2(\varepsilon,d) +  \theta \times (c_1(\varepsilon,d) - c_2(\varepsilon,d))} \\ 
    & \qquad + \dfrac{\theta c_1(\varepsilon,d)}{c_2(\varepsilon,d) +  \theta (c_1(\varepsilon,d) - c_2(\varepsilon,d))} \bigg[ \log \dfrac{c_2(\varepsilon,d)+ \Expectation[\theta] \times (c_1(\varepsilon,d) - c_2(\varepsilon,d))}{c_2(\varepsilon,d) +  \theta \times (c_1(\varepsilon,d) - c_2(\varepsilon,d))}  \\
    &  \qquad \qquad + \log \dfrac{\theta c_1(\varepsilon,d)}{\Expectation[\theta] c_1(\varepsilon,d) + (\theta - \Expectation[\theta]) \times c_2(\varepsilon,d)} \bigg] \\
    & \stackrel{(c)}{\leq}  \dfrac{\theta c_1(\varepsilon,d)}{c_2(\varepsilon,d) +  \theta (c_1(\varepsilon,d) - c_2(\varepsilon,d))} \log \bigg( \dfrac{ \theta c_1(\varepsilon,d)}{\Expectation[\theta] c_1(\varepsilon,d) + (\theta - \Expectation[\theta]) \times c_2(\varepsilon,d)} \bigg) \label{eq:KL_2}\\
    & \stackrel{(d)}{\leq} \log \bigg( \dfrac{ c_1(\varepsilon,d)}{\Expectation[\theta] c_1(\varepsilon,d) + (1 - \Expectation[\theta]) \times c_2(\varepsilon,d)} \bigg)  \stackrel{(e)}{\leq} \log \frac{c_1(\varepsilon,d)}{c_2(\varepsilon,d)} \stackrel{(f)}{\leq} \varepsilon \log e  
\end{align}
where $(a)$ follows from the definition of $\piMRC_{\svbx, \theta}(k)$ and $\piMMRC_{\svbx, \theta}(k)$, $(b)$ follows because $|\{k : k \in \msf{Cap}_{\svbx}\}| = \theta N$ and $|\{k : k \notin \msf{Cap}_{\svbx}\}| = (1-\theta) N$, $(c)$ follows because $\log \dfrac{c_2(\varepsilon,d) + \Expectation[\theta] \times (c_1(\varepsilon,d) - c_2(\varepsilon,d))}{c_2(\varepsilon,d) +  \theta \times (c_1(\varepsilon,d) - c_2(\varepsilon,d))} \leq 0$, $(d)$ follows because $\theta \leq 1$, $(e)$ follows because $\Expectation[\theta] \geq 0$, and $(f)$ follows because $c_1(\varepsilon,d) / c_2(\varepsilon,d) \leq \exp{(\varepsilon)}$.
\end{enumerate}
\end{proof}

\subsubsection{The expected KL divergence between the distribution of indices in \texorpdfstring{$\MRC$}{MRC} and \texorpdfstring{$\MMRC$}{MMRC} can be controlled arbitrarily when \texorpdfstring{$N$}{N} is in the right order} 
\begin{lemma}\label{lemma:expected_kl}
Let $q(\svbz|\svbx)$ be an $\varepsilon$-LDP cap-based mechanism. Let $p(\svbz)$ be the uniform distribution over $\cZ$ and let $\{\svbz_k\}_{k=1}^{N}$ be $N$ candidates drawn from $p(\svbz)$. Let $\theta$ denote the fraction of candidates inside the $\msf{Cap}_{\svbx}$ associated with $q(\svbz|\svbx)$. Let $\piMRC$ be the distribution over the indices $k \in [N]$ under $\MRC$ obtained from Algorithm \ref{alg:mrc} and $\piMMRC$ be the distribution over the indices $k \in [N]$ under $\MMRC$ obtained from Algorithm \ref{alg:mmrc}. Then,
\begin{align}
    \Expectation_{\theta}\lb\KLD{\piMRC_{\svbx, \theta}(\cdot)}{\piMMRC_{\svbx, \theta}(\cdot)}\rb \leq \rho \times \log e \times ( 1  +  \varepsilon)
\end{align}
where $\rho \in (0,1)$ is a free variable that is related to $N$ as follows:
\begin{align}
    N = \frac{2\lp\exp(\varepsilon)-1\rp^2}{\rho^2}  \ln\frac{2}{\rho}.  
\end{align}
\end{lemma}
\begin{proof}
Let $\theta$ denote the fraction of candidates inside the cap, i.e.,
\begin{align}
\theta = \dfrac{1}{N} \sum_{k =1}^{N} \Indicator(\svbz_k \in \msf{Cap}_{\svbx}).
\end{align}
Therefore, we have
\begin{align}
\Expectation[\theta] = \Probability_{\svbz_k \sim \Unif(\cZ)}\lp \svbz_k \in \msf{Cap}_{\svbx} \rp = \Probability_{\svbz \sim \Unif(\cZ)}\lp \svbz \in \msf{Cap}_{\svbx} \rp. \label{eq:expec_prob_rel}
\end{align}
Now, using the Hoeffding's inequality, we have $\Probability\lbp \lba \theta- \E[\theta] \rba \geq \sqrt{\frac{\ln\lp 2/\rho \rp}{2N}} \rbp \leq \rho$. Letting $\hat{\rho} = \sqrt{\frac{\ln\lp 2/\rho \rp}{2N}}$, we have
\begin{align}
    & \Expectation_{\theta}\lb\KLD{\piMRC_{\svbx, \theta}(\cdot)}{\piMMRC_{\svbx, \theta}(\cdot)}\rb \\
    & \qquad = \sum_{\theta} \Probability(\theta) \times \KLD{\piMRC_{\svbx, \theta}(\cdot)}{\piMMRC_{\svbx, \theta}(\cdot)}\\
    & \qquad = \sum_{\theta : \lba \theta- \E[\theta] \rba \leq \hat{\rho}} \Probability(\theta) \KLD{\piMRC_{\svbx, \theta}(\cdot)}{\piMMRC_{\svbx, \theta}(\cdot)} +  \hspace{-5mm} \sum_{\theta : \lba \theta- \E[\theta] \rba > \hat{\rho}} \Probability(\theta) \KLD{\piMRC_{\svbx, \theta}(\cdot)}{\piMMRC_{\svbx, \theta}(\cdot)}
    \label{eq:expected_kl_1}
\end{align}
Now, we will upper bound $\KLD{\piMRC_{\svbx, \theta}(\cdot)}{\piMMRC_{\svbx, \theta}(\cdot)}$ whenever $\theta$ is such that $\lba \theta- \E[\theta] \rba \leq \hat{\rho}$. As in the proof of Lemma \ref{lemma:kl}, we have three different cases depending on whether $\theta = \Expectation[\theta]$, $\theta < \Expectation[\theta]$ or $\theta > \Expectation[\theta]$. 
\begin{enumerate}
\item[1.]  For $\theta = \Expectation[\theta]$, using \eqref{eq:KL_0}, we have $\KLD{\piMRC_{\svbx, \theta}(\cdot)}{\piMMRC_{\svbx, \theta}(\cdot)} = 0$. 
\item[2.]
For $\theta < \Expectation[\theta]$, using \eqref{eq:KL_1}, we have
\begin{align}
\KLD{\piMRC_{\svbx, \theta}(\cdot)}{\piMMRC_{\svbx, \theta}(\cdot)} & \leq  \log \dfrac{c_2(\varepsilon,d) + \Expectation[\theta] \times (c_1(\varepsilon,d) - c_2(\varepsilon,d))}{c_2(\varepsilon,d) +  \theta \times (c_1(\varepsilon,d) - c_2(\varepsilon,d))} \\
& \stackrel{(a)}{=} \log \dfrac{c_2(\varepsilon,d) + \Expectation[\theta] \times (c_1(\varepsilon,d) - c_2(\varepsilon,d))}{c_2(\varepsilon,d) +  \lp\Expectation[\theta] - t\rp \times (c_1(\varepsilon,d) - c_2(\varepsilon,d))} \\
& = \log \lp 1+ \dfrac{t \times (c_1(\varepsilon,d) - c_2(\varepsilon,d))}{c_2(\varepsilon,d) +  \lp\Expectation[\theta] - t\rp \times (c_1(\varepsilon,d) - c_2(\varepsilon,d))} \rp\\
& \stackrel{(b)}{\leq} \dfrac{\log e \times t \times (c_1(\varepsilon,d) - c_2(\varepsilon,d))}{c_2(\varepsilon,d) +  \lp\Expectation[\theta] - t\rp \times (c_1(\varepsilon,d) - c_2(\varepsilon,d))}\\
& \stackrel{(c)}{\leq} \log e \times t \times \lp \frac{c_1(\varepsilon,d) - c_2(\varepsilon,d)}{c_2(\varepsilon,d)} \rp\\ 
& \stackrel{(d)}{\leq} \log e \times \hat{\rho} \times \lp \frac{c_1(\varepsilon,d) - c_2(\varepsilon,d)}{c_2(\varepsilon,d)} \rp \label{eq:ekl_1}
\end{align}
where $(a)$ follows by letting $\theta = \Expectation[\theta] - t$ with $t > 0$, $(b)$ follows by using $\log (1+x) \leq x \log e$ for $x = \frac{t \times (c_1(\varepsilon,d) - c_2(\varepsilon,d))}{1+ \lp\Expectation[\theta] - t\rp \times (c_1(\varepsilon,d) - c_2(\varepsilon,d))} > 0$, $(c)$ follows because $\Expectation[\theta] - t = \theta \geq 0$, and $(d)$ follows because $t = \Expectation[\theta] - \theta \leq \hat{\rho}$.
\item[3.]
    For $\theta > \Expectation[\theta]$, using \eqref{eq:KL_2}, we have
\begin{align}
& \KLD{\piMRC_{\svbx, \theta}(\cdot)}{\piMMRC_{\svbx, \theta}(\cdot)} \\
& \qquad \leq \dfrac{\theta c_1(\varepsilon,d)}{c_2(\varepsilon,d) +  \theta (c_1(\varepsilon,d) - c_2(\varepsilon,d))} \log \bigg( \dfrac{ \theta c_1(\varepsilon,d)}{\Expectation[\theta] c_1(\varepsilon,d) + (\theta  - \Expectation[\theta] ) \times c_2(\varepsilon,d)} \bigg) \\
&  \qquad \stackrel{(a)}{\leq} \log \bigg( \dfrac{ \theta c_1(\varepsilon,d)}{\Expectation[\theta] c_1(\varepsilon,d) + (\theta  - \Expectation[\theta]) \times c_2(\varepsilon,d)} \bigg) \\
&  \qquad \stackrel{(b)}{=} \log \bigg( \dfrac{ (\Expectation[\theta] + t)c_1(\varepsilon,d)}{\Expectation[\theta] c_1(\varepsilon,d) + t c_2(\varepsilon,d)} \bigg) \\
& \qquad = \log \bigg(1+ \dfrac{ t\lp c_1(\varepsilon,d)-c_2(\varepsilon,d) \rp}{\Expectation[\theta] c_1(\varepsilon,d) + t c_2(\varepsilon,d)} \bigg) \\
& \qquad \stackrel{(c)}{\leq} \frac{\log e \times  t\lp c_1(\varepsilon,d)-c_2(\varepsilon,d) \rp}{\Expectation[\theta] c_1(\varepsilon,d) + t c_2(\varepsilon,d)} \\
& \qquad \stackrel{(d)}{\leq} \frac{\log e \times t\lp c_1(\varepsilon,d)-c_2(\varepsilon,d) \rp}{\Expectation[\theta] c_1(\varepsilon,d)} \stackrel{(e)}{\leq} \frac{\log e \times \hat{\rho}\lp c_1(\varepsilon,d)-c_2(\varepsilon,d) \rp}{\Expectation[\theta] c_1(\varepsilon,d)} \label{eq:ekl_2}
\end{align}
where $(a)$ follows because $\theta \leq 1$, $(b)$ follows by letting $\theta = \Expectation[\theta] + t$ with $t > 0$, $(c)$ follows by using $\log (1+x) \leq x \log e$ for $x = \frac{ t\lp c_1(\varepsilon,d)-c_2(\varepsilon,d) \rp}{\Expectation[\theta] c_1(\varepsilon,d) + t c_2(\varepsilon,d)} > 0$, $(d)$ follows because $t > 0$, and $(e)$ follows because $t = \theta - \Expectation[\theta]  \leq \hat{\rho}$.
\end{enumerate}
Therefore, for $\theta$ such that $\lba \theta- \E[\theta] \rba \leq \hat{\rho}$, we have the following from \eqref{eq:ekl_1} and \eqref{eq:ekl_2}:
\begin{align}
    \KLD{\piMRC_{\svbx, \theta}(\cdot)}{\piMMRC_{\svbx, \theta}(\cdot)} & \leq \dfrac{  \log e \times \hat{\rho}\lp c_1(\varepsilon,d)-c_2(\varepsilon,d) \rp}{\min \lbp c_2(\varepsilon,d), \Expectation[\theta] c_1(\varepsilon,d)\rbp} \\
    & \stackrel{(a)}{=} \dfrac{ \log e \times \hat{\rho}\lp c_1(\varepsilon,d)-c_2(\varepsilon,d) \rp}{\min \lbp c_2(\varepsilon,d), c_1(\varepsilon,d) \Probability_{\svbz \sim \Unif(\cZ)}\lp \svbz \in \msf{Cap}_{\svbx} \rp \rbp} \\
    & \stackrel{(b)}{\leq} \dfrac{ 2 \log e \times   \hat{\rho}\lp c_1(\varepsilon,d) - c_2(\varepsilon,d) \rp}{c_2(\varepsilon,d)} \\
    & \stackrel{(c)}{\leq} 2 \log e \times  \hat{\rho} \lp \exp(\varepsilon) - 1 \rp
    \label{eq:expected_kl_2}
\end{align}
where $(a)$ follows from \eqref{eq:expec_prob_rel}, $(b)$ follows because $\Probability_{\svbz \sim \Unif(\cZ)}\lp \svbz \in \msf{Cap}_{\svbx} \rp \geq c_2(\varepsilon,d)/2c_1(\varepsilon,d)$ from the definition of cap-based mechanisms, and $(c)$ follows because $c_1(\varepsilon,d) / c_2(\varepsilon,d) \leq \exp{(\varepsilon)}$.

Using \eqref{eq:expected_kl_2} and Lemma \ref{lemma:kl} in \eqref{eq:expected_kl_1}, we have
\begin{align}
 & \Expectation_{\theta}\lb\KLD{\piMRC_{\svbx, \theta}(\cdot)}{\piMMRC_{\svbx, \theta}(\cdot)}\rb \\
 & \qquad \leq \sum_{\theta : \lba \theta- \E[\theta] \rba \leq \hat{\rho}} \Probability(\theta) \times 2\log e \times \hat{\rho} \lp \exp(\varepsilon) - 1 \rp +  \sum_{\theta : \lba \theta- \E[\theta] \rba > \hat{\rho}} \Probability(\theta) \times \varepsilon \log e \\
 & \qquad \stackrel{(a)}{\leq} 2 \log e \times \hat{\rho} \lp \exp(\varepsilon) - 1 \rp + \rho \varepsilon \log e
 \label{eq:expected_kl_3} \\
  & \qquad \stackrel{(b)}{\leq} 2 \log e \times \sqrt{\dfrac{\ln\lp 2/\rho \rp}{2N}} \lp \exp(\varepsilon) - 1 \rp + \rho \varepsilon \log e\\
  & \qquad \stackrel{(c)}{\leq} \log e \times \rho (1+\varepsilon)
\end{align}
where $(a)$ follows because $\Probability \lp\lba \theta- \E[\theta] \rba \leq \hat{\rho}\rp \leq 1$ and $\Probability\lp\lba \theta- \E[\theta] \rba \geq \hat{\rho}\rp\ \leq \rho$, $(b)$ follows by plugging in $\hat{\rho} = \sqrt{\frac{\ln\lp 2/\rho \rp}{2N}}$, and $(c)$ follows by plugging in $N$.
\end{proof}

\subsection{Utility of \texorpdfstring{$\MMRC$}{MMRC}}\label{appendix:utility_mmrc}

In this section, we first prove Theorem \ref{theorem:mmrc_accuracy_local_wrt_mrc} i.e., we show that, with number of candidates exponential in $\varepsilon$, samples drawn from $\qMMRC$ will be similar to the samples drawn from $\qMRC$ in terms of $\ell_2$ error.

Then, in Theorem \ref{theorem:mmrc_accuracy_local}, we show that $\MMRC$ can simulate any $\varepsilon$-LDP cap-based mechanism in a nearly lossless fashion with about $\varepsilon$ bits of communication.


\subsubsection{Utility of \texorpdfstring{$\MMRC$}{MMRC} with respect to \texorpdfstring{$\qMRC$}{}}\label{appendix:utility_mmrc_mrc}
\mmrcaccuracylocalwrtmrc*
\begin{proof}
We will first upper bound the difference between $\Expectation_{\qMMRC} \big[ \lV  \svbz - \svbx \rV^2_2  \big]$ and $\Expectation_{\qMRC} \big[ \lV  \svbz - \svbx \rV^2_2  \big]$ in terms of the total variation distance between $\qMRC$ and $\qMMRC$. Due to a property of the total variation distance \citep[e.g.,][]{SCV16}, we have
\begin{align}
  \Expectation_{\qMMRC} \big[ \lV  \svbz - \svbx \rV^2_2  \big]   - \Expectation_{\qMRC}\big[ \lV  \svbz - \svbx \rV^2_2  \big] \leq \max_{\svbx, \svbz} \lV  \svbz- \svbx \rV^2_2  \times   \TV{\qMRC(\svbz | \svbx)}{\qMMRC(\svbz | \svbx)}. \label{eq:mmrc_generic_1}
\end{align}
Next, we will upper bound the total variation distance between $\qMRC$ and $\qMMRC$ using Pinsker's inequality as follows:
\begin{align}
  \TV{\qMRC(\svbz | \svbx)}{\qMMRC(\svbz | \svbx)} \leq \sqrt{\frac{1}{2 \log e} \KLD{\qMRC(\svbz | \svbx)}{\qMMRC(\svbz | \svbx)}}. \label{eq:mmrc_generic_2}
\end{align}
Next, we will upper bound the KL divergence between $\qMRC(\svbz | \svbx)$ and $\qMMRC(\svbz | \svbx)$. To that end, for every $\svbx \in \cX$, let $\pMRC(\svbz_1, \cdots, \svbz_N, K, \svbz_K | \svbx)$ denote the joint distribution of the candidates $\svbz_1, \cdots, \svbz_N$ drawn from $p(\rvbz)$,  the transmitted index $K$ under $\MRC$, and the sample $\svbz_K$ corresponding to $K$. We have
\begin{align}
& \pMRC(\svbz_1, \cdots, \svbz_N, K, \svbz_K | \svbx)\\
& \qquad =  p(\svbz_1, \cdots, \svbz_N | \svbx)  \times \pMRC(K | \svbz_1, \cdots, \svbz_N, \svbx) \times \pMRC(\svbz_K | \svbz_1, \cdots, \svbz_N, K, \svbx)\\
& \qquad \stackrel{(a)}{=} p(\svbz_1, \cdots, \svbz_N) \times \pMRC(K | \svbz_1, \cdots, \svbz_N, \svbx) \times \pMRC(\svbz_K | \svbz_1, \cdots, \svbz_N, K, \svbx)\\
& \qquad \stackrel{(b)}{=} p(\svbz_1, \cdots, \svbz_N) \times \piMRC_{\svbx, \theta}(k) \times \pMRC(\svbz_K | \svbz_1, \cdots, \svbz_N, K, \svbx)\\
& \qquad \stackrel{(c)}{=} p(\svbz_1, \cdots, \svbz_N) \times \piMRC_{\svbx, \theta}(k) \label{eq:mmrc_generic_3}
\end{align}
where $(a)$ follows because $\svbz_1, \cdots, \svbz_N$ are independent of $\svbx$, $(b)$ follows because $\pMRC(K | \svbz_1, \cdots, \svbz_N, \svbx) = \piMRC_{\svbx, \theta}(k)$, and $(c)$ follows because $\pMRC(\svbz_K | \svbz_1, \cdots, \svbz_N, K, \svbx) = 1$ (note that $\svbz_K$ can be viewed as a function of $(\svbz_1,...,\svbz_N, K)$). 

Similarly, for every $\svbx \in \cX$, let $\pMMRC(\svbz_1, \cdots, \svbz_N, K, \svbz_K | \svbx)$ denote the joint distribution of the candidates $\svbz_1, \cdots, \svbz_N$ drawn from $p(\rvbz)$,  the transmitted index $K$ under $\MMRC$, and the sample $\svbz_K$ corresponding to $K$. We have
\begin{align}
& \pMMRC(\svbz_1, \cdots, \svbz_N, K, \svbz_K | \svbx) \\
& \qquad =  p(\svbz_1, \cdots, \svbz_N | \svbx)  \times \pMMRC(K | \svbz_1, \cdots, \svbz_N, \svbx) \times \pMMRC(\svbz_K | \svbz_1, \cdots, \svbz_N, K, \svbx)\\
& \qquad \stackrel{(a)}{=} p(\svbz_1, \cdots, \svbz_N) \times \pMMRC(K | \svbz_1, \cdots, \svbz_N, \svbx) \times \pMMRC(\svbz_K | \svbz_1, \cdots, \svbz_N, K, \svbx)\\
& \qquad \stackrel{(b)}{=} p(\svbz_1, \cdots, \svbz_N) \times \piMMRC_{\svbx, \theta}(k) \times \pMMRC(\svbz_K | \svbz_1, \cdots, \svbz_N, K, \svbx)\\
& \qquad \stackrel{(c)}{=} p(\svbz_1, \cdots, \svbz_N) \times \piMMRC_{\svbx, \theta}(k) \label{eq:mmrc_generic_4}
\end{align}
where $(a)$ follows because $\svbz_1, \cdots, \svbz_N$ are independent of $\svbx$, $(b)$ follows because $\pMMRC(K | \svbz_1, \cdots, \svbz_N, \svbx) = \piMMRC_{\svbx, \theta}(k)$, and $(c)$ follows because $\pMMRC(\svbz_K | \svbz_1, \cdots, \svbz_N, K, \svbx) = 1$. 

We are now in a position to upper bound the KL divergence between $\qMRC(\svbz_K | \svbx)$ and $\qMMRC(\svbz_K | \svbx)$:
\begin{align}
 \KLD{\qMRC(\svbz | \svbx)}{\qMMRC(\svbz | \svbx)} &  \stackrel{(a)}{\leq}  \KLD{\pMRC(\svbz_1, \cdots, \svbz_N, K, \svbz_K | \svbx)}{\pMMRC(\svbz_1, \cdots, \svbz_N, K, \svbz_K | \svbx)}\\
 &  \stackrel{(b)}{=}   \KLD{p(\svbz_1, \cdots, \svbz_N) \times \piMRC_{\svbx, \theta}(k)}{p(\svbz_1, \cdots, \svbz_N) \times \piMMRC_{\svbx, \theta}(k)}\\
 &  \stackrel{(c)}{=} \Expectation_{\svbz_1, \cdots, \svbz_N} \lb\KLD{\piMRC_{\svbx, \theta}(k)}{\piMMRC_{\svbx, \theta}(k)}\rb \\
 &  \stackrel{(d)}{=}\Expectation_{\theta}\lb\KLD{\piMRC_{\svbx, \theta}(\cdot)}{\piMMRC_{\svbx, \theta}(\cdot)}\rb  \stackrel{(e)}{\leq} \rho \times \log e \times (1+\varepsilon) \label{eq:mmrc_generic_5}
\end{align}
where $(a)$ follows because by the chain rule for KL-divergence, $(b)$ follows from \eqref{eq:mmrc_generic_3} and \eqref{eq:mmrc_generic_4}, $(c)$ follows by the definition of KL-divergence, $(d)$ follows because $\piMRC$ and $\piMMRC$ depend on $\svbz_1, \cdots, \svbz_N$ only via $\theta$ for cap-based mechanisms, and $(e)$ follows from Lemma \ref{lemma:expected_kl} because $N = \frac{2\lp\exp(\varepsilon)-1\rp^2}{\rho^2}  \ln\frac{2}{\rho}  $. Combining \eqref{eq:mmrc_generic_1}, \eqref{eq:mmrc_generic_2},and \eqref{eq:mmrc_generic_5}, we have 
\begin{align}
    \Expectation_{\qMMRC} \big[ \lV  \svbz - \svbx \rV^2_2  \big]  \leq \Expectation_{\qMRC} \big[ \lV  \svbz - \svbx \rV^2_2  \big]  + \sqrt{\frac{\rho (1+\varepsilon)}{2}} \times \max_{\svbx, \svbz} \lV  \svbz - \svbx \rV^2_2. \label{eq:mrc_mmrc_relation}
\end{align}
\end{proof}

\begin{remark}\label{rmk:error_max}
For bounded $\varepsilon$-LDP mechanisms such as \emph{$\PrivUnit$} and \emph{$\SubsetSelection$}, the term $\max_{\svbx, \svbz} \lV  \svbz - \svbx \rV^2_2$ in \eqref{eq:mmrc_utility} is of the same order as $\Expectation_{q}\big[\|\svbz - \svbx\|^2\big]$. Therefore, by picking a large $N$ in Theorem \ref{theorem:mmrc_accuracy_local_wrt_mrc} (i.e. $\log N \geq C \varepsilon$ for a sufficiently large $C$), $\rho$ can be made arbitrarily small and the estimation error of \emph{$\MMRC$} can be arbitrarily close to the estimation error of \emph{$\MRC$}.
\end{remark}

\subsubsection{Utility of \texorpdfstring{$\MMRC$}{MMRC} with respect to \texorpdfstring{$q$}{q}}\label{appendix:utility_mmrc_q}

\begin{theorem}\label{theorem:mmrc_accuracy_local}
Consider any input alphabet $\cX$, output alphabet $\cZ$, data $\svbx \in \cX$, and $\varepsilon$-LDP cap-based mechanism $q(\cdot|\svbx)$.
Let the reference distribution $p(\cdot)$ be the uniform distribution on $\cZ$.
Let $N$ denote the number of candidates. Then, $\qMMRC$ is such that 
\begin{align}
    \Expectation_{\qMMRC} \big[ \lV  \svbz - \svbx \rV^2_2  \big]  \leq \Expectation_{q} \big[ \lV  \svbz - \svbx \rV^2_2  \big]  + \sqrt{\frac{\rho (1+\varepsilon)}{2}} \times \max_{\svbx, \svbz} \lV  \svbz - \svbx \rV^2_2  + \frac{2\alpha }{1-\alpha} \times \sqrt{\Expectation_{q}[\|\svbz - \svbx\|^4]}
\end{align}
holds with probability at least $1 - 2\alpha$ where
\begin{align}
\alpha = \sqrt{2^{-c\varepsilon} + 2^{-c^2/\log e + 1}}.
\end{align}
and $c$ and $\rho \in (0,1)$ are free variables such that
\begin{align}
    N = \max\bigg\{ 2^{(\log e + 4c) \varepsilon}, \frac{2\lp\exp(\varepsilon)-1\rp^2 }{\rho^2}  \ln\frac{2}{\rho} \bigg\}
\end{align}
\end{theorem}
\begin{proof}
The proof follows from Theorem \ref{theorem:mrc_accuracy_local} and Theorem \ref{theorem:mmrc_accuracy_local_wrt_mrc}.
\end{proof}

\section{Preliminary on \texorpdfstring{$\PrivUnit$}{PrivUnit}}
\label{appendix:privunit}

First, we briefly recap the $\PrivUnit$ mechanism ($\qPU$) proposed in \cite{BDFKR2018}. $\PrivUnit$ is a private sampling scheme when the input alphabet $\cX$ is the $d-$dimensional unit $\ell_2$ sphere $\sphere^{d-1}$. More formally, given a vector $\svbx \in \sphere^{d-1}$, $\PrivUnit$ (see Algorithm \ref{alg:privunit}) draws a vector $\svbz$ from a spherical cap $\{\svbz \in \sphere^{d-1} \mid
\lan \svbz, \svbx  \ran \ge \gamma\}$ with probability $p_0 \geq 1/2$ or from its complement $\{\svbz \in \sphere^{d-1} \mid \lan \svbz, \svbx \ran < \gamma\}$  with probability $1 - p_0$, where
$\gamma \in [0, 1]$ and $p_0$ are constants that trade accuracy and privacy. In other words, the conditional density $\qPU(\svbz | \svbx)$ is:
\begin{align}\label{eq:qpu}
    \qPU(\rvbz | \rvbx) = 
    \begin{cases}
      p_0 \times \dfrac{2}{A(1,d)I_{1-\gamma^2}(\frac{d-1}{2},\frac{1}{2})} & \text{if}\ \langle \rvbx, \rvbz \rangle \geq \gamma
      \\[10pt]
      (1-p_0) \times \dfrac{2}{2A(1,d)- A(1,d)I_{1-\gamma^2}(\frac{d-1}{2},\frac{1}{2})} & \text{otherwise}
    \end{cases}
\end{align}
where $A(1,d)$ denotes the area of $\sphere^{d-1}$ and $I_x(a,b)$ denotes the regularized incomplete beta function.

\begin{algorithm}[h]
\caption{Privatized Unit Vector: $\PrivUnit$}
\label{alg:privunit}
\begin{algorithmic}
\Require $\svbx \in \sphere^{d-1}$, $\gamma \in [0,1]$, $p_0 \ge 1/2$.
\State Draw random vector $\svbz$ according to the distribution
\State \begin{equation}
	  \label{eqn:w-flip-mechanism}
	  \svbz = \begin{cases}
	    \mbox{uniform on } \{\svbz \in \sphere^{d-1} \mid
	    \lan \svbz, \svbx 
	    \ran \ge \gamma\} & \mbox{with~probability~} p_0 \\
	    \mbox{uniform on } \{\svbz \in \sphere^{d-1} \mid
	    \lan \svbz, \svbx 
	    \ran < \gamma \} & \mbox{otherwise}.
	  \end{cases}
\end{equation}
\State Set $\alpha = \frac{d-1}{2}$, $\tau = \frac{1+\gamma}{2}$, and
\begin{equation}
  m_{\texttt{pu}} = \frac{(1 - \gamma^2)^\alpha}{2^{d-2} (d - 1)}
  \left[\frac{p_0}{B(\alpha,\alpha) - B(\tau; \alpha,\alpha)}
    - \frac{1 - p_0}{B(\tau; \alpha, \alpha)}\right]
  \label{eqn:norm-of-W}
\end{equation}
\Return $\hat{\svbx}^{\texttt{pu}} = \frac{\svbz}{m_{\texttt{pu}}}$
\end{algorithmic}
\end{algorithm}

Given its inputs $\svbx, \gamma,$ and $p_0$, Algorithm \ref{alg:privunit} returns an estimator $\hat{\svbx}^{\texttt{pu}} \coloneqq \svbz/m_{\texttt{pu}}$ which is differentially private and unbiased where $m_{\texttt{pu}}$ is a scaling factor. The choice of $\gamma$ described in Theorem \ref{thm:PrivUnitPriv} ensures differential privacy and the choice of the scaling factor $m$ described in \eqref{eqn:norm-of-W} ensures unbiasedness where
\begin{align}
  B(x;\alpha,\beta) \coloneqq \int_{0}^x
  t^{\alpha - 1} (1- t)^{\beta-1}dt
  ~~ \mbox{where} ~~
  B(\alpha,\beta) \coloneqq
  B(1;\alpha,\beta) = \frac{\Gamma(\alpha)
    \Gamma(\beta)}{\Gamma(\alpha+\beta)} \label{eq:incomplete_beta}
\end{align}
denotes the incomplete beta function. 

\subsection{\texorpdfstring{$\PrivUnit$}{PrivUnit} is a differentially private mechanism}
\label{subsubsec:privunit_dp}
The following theorem borrowed from \cite{BDFKR2018} describes the choice of $\gamma$ and provides the precise associated differential privacy guarantee of the $\PrivUnit$ mechanism.
\begin{theorem}[\cite{BDFKR2018}]
  \label{thm:PrivUnitPriv}
  Let $\gamma \in [0,1]$ and $p_0 = \frac{e^{\varepsilon_0}}{1 +
    e^{\varepsilon_0}}$. Then algorithm \emph{$\PrivUnit(\cdot, \gamma, p_0)$}
  is $\varepsilon = (\bar{\varepsilon}+\varepsilon_0)$-differentially private whenever $\gamma \ge 0$ is
  such that
  \begin{align}
  \begin{aligned}\label{eqn:sufficient-gamma}
      \bar{\varepsilon} & \ge
      \log\frac{ 1+\gamma \cdot \sqrt{ 2(d-1) / \pi } }{
        {\lp 1 - \gamma \cdot \sqrt{ 2(d-1) / \pi } \rp_+} },
      ~~~ \mbox{i.e.} ~~~
      \gamma \le \frac{e^{\bar{\varepsilon}} - 1}{e^{\bar{\varepsilon}} + 1} \sqrt{\frac{\pi}{2(d-1)}},\\
     or\\
       \bar{\varepsilon} & \ge 1/2 \log(d) + \log 6 - \frac{d - 1}{2}
       \log (1 - \gamma^2) + \log \gamma
       ~~ \mbox{and} ~~
       \gamma \ge \sqrt{\frac{2}{d}}.
  \end{aligned}
  \end{align}
\end{theorem}

Here, $\varepsilon$ can be viewed as the total privacy budget. Typically, $\mu$ fraction of this budget is allocated for the spherical cap threshold $\gamma$ and $1-\mu$ fraction is allocated to the probability parameter $p_0$ with which a particular spherical cap is chosen i.e., $\bar{\varepsilon} = \mu \varepsilon$ and $\varepsilon_0 = (1-\mu) \varepsilon$ for some $\mu \in [0,1]$. While the parameter $\mu$ can be optimized over as described in \cite{FT21}, we will view it as a constant for convenience. Our results on $\MRC$ and $\MMRC$ simulating $\PrivUnit$ can be easily extended to the setup where $\mu$ needs to be optimized over.

\subsection{\texorpdfstring{$\PrivUnit$}{PrivUnit} is unbiased and order-optimal}
The following lemma borrowed from \cite{BDFKR2018} shows that the output of the $\PrivUnit$ mechanism  (a) is unbiased, (b) has a bounded norm, and (c) has order-optimal utility.

\begin{proposition}[\cite{BDFKR2018}]
 Let $\hat{\svbx}^{\texttt{pu}}$ = \emph{$\PrivUnit(\svbx,\gamma, p_0)$} for some $\svbx \in \sphere^{d-1}$, $\gamma
  \in [0,1]$, and $p_0 \in [1/2, 1]$. Then, $\E[\hat{\svbx}^{\texttt{pu}}] = \svbx$. Further, assume that $0 \le \varepsilon \le d$. Then, there exists a numerical constant $c < \infty$ such that if
  $\gamma$ saturates either of the two inequalities
  \eqref{eqn:sufficient-gamma}, then $\gamma \gtrsim \min\{\varepsilon / \sqrt{d},
  \sqrt{\varepsilon / d}\}$, and 
  \begin{equation*}
    \ltwo{\hat{\svbx}^{\texttt{pu}}}
    \le c \cdot \sqrt{\frac{d}{\varepsilon} \vee
      \frac{d}{(e^\varepsilon - 1)^2}}.
  \end{equation*}
  Additionally, $\E[\ltwo{\hat{\svbx}^{\texttt{pu}} - \svbx}^2] \lesssim
  \frac{d}{\varepsilon} \vee \frac{d}{(e^\varepsilon - 1)^2}$.
 \label{proposition:ltwo-utility}
\end{proposition}

\subsection{\texorpdfstring{$\PrivUnit$}{PrivUnit} is a cap-based mechanism}
The randomness in the estimator $\hat{\svbx}^{\texttt{pu}}$ obtained from the $\PrivUnit(\svbx,\gamma, p_0)$ mechanism comes from $\svbz$. Therefore, we obtain a convenient expression for the conditional distribution of $\svbz$ conditioned on $\svbx$ i.e., $\qPU(\svbz | \svbx)$. Define $\msf{Cap}_{\svbx} \coloneqq \{\svbz | \langle \svbx, \svbz \rangle \geq \gamma\}$. Recall from \eqref{eqn:sufficient-gamma} that $\gamma$ is a function of $\varepsilon$ and $d$. 
Further, as described in Section \ref{subsubsec:privunit_dp}, when the budget split parameter $\mu$ is known, $p_0$ can viewed as a function of $\varepsilon$. Then, the conditional distribution $\qPU(\svbz | \svbx)$ in \eqref{eq:qpu} can be written as follows:
\begin{align}
    \qPU(\svbz | \svbx) = 
    \begin{cases}
      c_1(\varepsilon, d) & \text{if}\ \svbz \in \msf{Cap}_{\svbx}
      \\[10pt]
      c_2(\varepsilon, d) & \text{if}\ \svbz \notin \msf{Cap}_{\svbx}
    \end{cases} \label{eq:privunit_density}
\end{align}
where $c_1(\varepsilon, d) = p_0 \times \dfrac{2}{A(1,d)I_{1-\gamma^2}(\frac{d-1}{2},\frac{1}{2})}$ and $c_2(\varepsilon, d) = (1-p_0) \times \dfrac{2}{2A(1,d)- A(1,d)I_{1-\gamma^2}(\frac{d-1}{2},\frac{1}{2})}$ are functions of $\varepsilon$ and $d$. 

Further, $\Probability_{\svbz \sim \Unif(\cZ)}\lp \svbz \in \cZ_{\svbx} \rp = \frac{I_{1-\gamma^2}(\frac{d-1}{2},\frac{1}{2})}{2}$. Therefore,
\begin{align}
    \frac{c_1(\varepsilon, d)}{c_2(\varepsilon, d)} \times \Probability_{\svbz \sim \Unif(\cZ)}\lp \svbz \in \cZ_{\svbx} \rp = \frac{p_0 \times (2 - I_{1-\gamma^2}(\frac{d-1}{2},\frac{1}{2}))}{2(1-p_0)} \stackrel{(a)}{\geq} \frac{p_0}{2(1-p_0)} \stackrel{(b)}{\geq} \frac{1}{2}
\end{align}
where $(a)$ follows because $I_{1-\gamma^2}(\frac{d-1}{2},\frac{1}{2}) \leq 1$ and $(b)$ follows because $p_0 \geq 1/2$.

\section{Simulating \texorpdfstring{$\PrivUnit$}{PrivUnit} using Minimal Random Coding}\label{appendix:mrc_pu}

In this section, we simulate $\PrivUnit$ using $\MRC$ analogous to how we simulate $\PrivUnit$ using $\MMRC$ in Section \ref{sec:mean_estimation}. First, in Appendix \ref{appendix:debias_mrc}, we provide an unbiased estimator for $\MRC$ simulating $\PrivUnit$.
Next, in Appendix \ref{appendix:mrc_pu_ut} we provide the utility guarantee associated with $\MRC$ simulating $\PrivUnit$. To do that, first, in Appendix \ref{appendix:scaling_mrc_pu}, we show that when the number of candidates $N$ is exponential in $\varepsilon$, the scaling factor $m_{\texttt{mrc}}$ is close to the scaling parameter associated with $\PrivUnit$ (i.e., $m_{\texttt{pu}}$). Next, in Appendix \ref{appendix:mrc_pu_scaling_mse}, we provide the relationship between the mean squared error associated with $\MRC$ simulating $\PrivUnit$ and the mean squared error associated with $\PrivUnit$. In Appendix \ref{appendix:mrc_pu_utility}, we combine everything and show that, for mean estimation, $\MRC$ can simulate $\PrivUnit$ in a near-lossless manner while only using on the order of $\varepsilon$-bits of communication. Finally, in Appendix \ref{appendix:mrc_pu_emp}, we provide some empirical comparisons.

\subsection{Unbiased Minimal Random Coding simulating \texorpdfstring{$\PrivUnit$}{PrivUnit}}\label{appendix:debias_mrc}
Consider the $\PrivUnit$ $\varepsilon$-LDP mechanism $\qPU$ described in Section \ref{sec:preliminaries} with parameters $p_0$ and $\gamma$. $\PrivUnit$ is a cap-based mechanism with $\msf{Cap}_{\svbx} = \{\svbz \in \sphere^{d-1} \mid
\lan \svbz, \svbx  \ran \ge \gamma\}$ as discussed in Appendix \ref{appendix:privunit}. Let $\piMRC$ be the distribution and $\svbz_1, \svbz_2,...,\svbz_N$ be the candidates obtained from Algorithm \ref{alg:mrc} when the reference distribution is $\Unif(\sphere^{d-1})$. Let  $K \sim \piMRC(\cdot)$. Define $p_{\texttt{mrc}} \coloneqq \Probability(\svbz_K \in \msf{Cap}_{\svbx})$
to be the probability with which the sampled candidate $\svbz_K$ belongs to the spherical cap associated with $\PrivUnit$.
Define $m_{\texttt{mrc}}$ as the scaling factor in \eqref{eq:m} when $p_0$ in \eqref{eq:m} is replaced by $p_{\texttt{mrc}}$. Define $\hat{\svbx}^\texttt{mrc} \coloneqq \svbz_K / m_\texttt{mrc}$ as the estimator of the $\MRC$ mechanism simulating $\PrivUnit$.  The following Lemma shows that $\hat{\svbx}^\texttt{mrc}$ is an unbiased estimator.


\begin{lemma}\label{theorem:mrc_privunit_bias}
Let $\hat{\svbx}^\texttt{mrc}$ be the estimator of the \emph{$\MRC$} mechanism simulating \emph{$\PrivUnit$} as defined above. Then, $\Expectation_{\qMRC}[\hat{\svbx}^\texttt{mrc}] = \svbx$.
\end{lemma}
\begin{proof}
For $k \in [N]$, let $A_k \coloneqq \Indicator(\svbz_k \in \msf{Cap}_{\svbx})$. Then, $p_{\texttt{mrc}} = \Probability(A_K = 1)$. Using the definition of $\hat{\svbx}^\texttt{mrc}$, we have
\begin{align}
\Expectation_{\qMRC} [\hat{\svbx}^\texttt{mrc}] & = \frac{1}{m_\texttt{mrc}} \Expectation_{\qMRC} [\svbz_K ].
\end{align}
Let us evaluate $\Expectation_{\qMRC} [\svbz_K ]$. We have
\begin{align}
\Expectation_{\qMRC} [\svbz_K ] & \stackrel{(a)}{=} \Expectation_{K, \svbz_1, \cdots, \svbz_N} [\svbz_K ] \\
& \stackrel{(b)}{=}  \Expectation_{\svbz_1, \cdots, \svbz_N} \Big[\sum_{k = 1}^{N}  \piMRC_{\svbx, \svbz_1,...,\svbz_N}(k) \times \svbz_k \Big] \\
& \stackrel{(c)}{=} \Expectation_{A_1, \cdots,A_N}\bigg[ \Expectation_{\svbz_1, \cdots, \svbz_N} \Big[\sum_{k = 1}^{N} \piMRC_{\svbx, \svbz_1,...,\svbz_N}(k) \times \svbz_k \Big| A_1, \cdots,A_N\Big]\bigg] \\
& \stackrel{(d)}{=} \sum_{k = 1}^{N} \Expectation_{A_1, \cdots,A_N}\bigg[   \Expectation_{\svbz_1, \cdots, \svbz_N} \Big[ \piMRC_{\svbx, \svbz_1,...,\svbz_N}(k) \times \svbz_k  \Big| A_1, \cdots,A_N\Big]\bigg] \\
& \stackrel{(e)}{=} \sum_{k = 1}^{N} \Expectation_{A_1, \cdots,A_N}\bigg[  \piMRC_{\svbx, A_1,...,A_N}(k) \Expectation_{\svbz_1, \cdots, \svbz_N} \Big[ \svbz_k  \Big| A_1, \cdots,A_N\Big]\bigg] \\
& \stackrel{(f)}{=} \sum_{k = 1}^{N} \Expectation_{A_1, \cdots,A_N}\bigg[\piMRC_{\svbx, A_1,...,A_N}(k) \Expectation_{\svbz_1, \cdots, \svbz_N} \Big[ \svbz_k  \Big| A_k \Big]\bigg] \\
& \stackrel{(g)}{=} \sum_{k = 1}^{N} \Expectation_{A_1, \cdots,A_N}\bigg[ \piMRC_{\svbx, A_1,...,A_N}(k) \Expectation_{\svbz_k} \Big[ \svbz_k  \Big| A_k\Big]\bigg] \\
& \stackrel{(h)}{=} \sum_{k = 1}^{N} \Expectation_{A_k} \bigg[\Expectation_{A_1, \cdots,A_N}\Big[ \piMRC_{\svbx, A_1,...,A_N}(k) \Expectation_{\svbz_k} \big[ \svbz_k  \big| A_k\big] \Big| A_k \Big] \bigg]\\
& \stackrel{(i)}{=} \sum_{k = 1}^{N} \Probability(A_k = 1) \bigg[\Expectation_{A_1, \cdots,A_N}\Big[ \piMRC_{\svbx, A_1,...,A_N}(k) \Expectation_{\svbz_k} \big[ \svbz_k  \big| A_k\big] \Big| A_k = 1\Big] \bigg] \\
& + \sum_{k = 1}^{N} \Probability(A_k = 0) \bigg[\Expectation_{A_1, \cdots,A_N}\Big[ \piMRC_{\svbx, A_1,...,A_N}(k) \Expectation_{\svbz_k} \big[ \svbz_k  \big| A_k\big] \Big| A_k = 0\Big] \bigg] \\
& = \sum_{k = 1}^{N} \Probability(A_k = 1) \bigg[\Expectation_{A_1, \cdots,A_N}\Big[\piMRC_{\svbx, A_1, \cdots,A_k = 1, \cdots, A_N}(k) \Expectation_{\svbz_k} \big[ \svbz_k  \big| A_k = 1\big] \Big] \bigg] \\
&  + \sum_{k = 1}^{N} \Probability(A_k = 0) \bigg[\Expectation_{A_1, \cdots,A_N}\Big[ \piMRC_{\svbx, A_1, \cdots,A_k = 0, \cdots, A_N}(k) \Expectation_{\svbz_k} \big[ \svbz_k  \big| A_k = 0\big] \Big] \bigg]\\
& \stackrel{(j)}{=} \Expectation_{\svbz}\big[ \svbz  \big| A = 1\big]  \sum_{k = 1}^{N} \Probability(A_k = 1) \bigg[\Expectation_{A_1, \cdots,A_N}\Big[ \piMRC_{\svbx, A_1, \cdots,A_k = 1, \cdots, A_N}(k)  \Big] \bigg] \\
&  + \Expectation_{\svbz}\big[ \svbz  \big| A = 0\big] \sum_{k = 1}^{N} \Probability(A_k = 0) \bigg[\Expectation_{A_1, \cdots,A_N}\Big[ \piMRC_{\svbx, A_1, \cdots,A_k = 0, \cdots, A_N}(k) \Big] \bigg]\\
& \stackrel{(k)}{=} \Expectation_{\svbz}\big[ \svbz  \big| A = 1\big]  \sum_{k = 1}^{N} \Probability(A_k = 1) \piMRC_{\svbx, A_k = 1}(k) \\
&+  \Expectation_{\svbz}\big[ \svbz  \big| A = 0\big]  \sum_{k = 1}^{N} \Probability(A_k = 0) \piMRC_{\svbx, A_k = 0}(k) \\
& \stackrel{(l)}{=} \Expectation_{\svbz}\big[ \svbz  \big| A = 1\big] \Probability(A_K = 1) + \Expectation_{\svbz}\big[ \svbz  \big| A = 0\big] \Probability(A_K = 0)\\
& \stackrel{(m)}{=} m_\texttt{mrc} \svbx \label{eq:miracle_zero_bias}
\end{align}
where $(a)$ follows because the randomness in $\qMRC$ comes from the randomness in $K, \svbz_1, \cdots, \svbz_N$, $(b)$ follows by calculating the expectation over $K$ and showing the dependence of $\piMRC$ on $\svbz_1, \cdots, \svbz_N$ explicitly, $(c)$ follows by the tower property of expectation, $(d)$ follows by linearity of expectation, $(e)$ follows because $\piMRC_{\svbx, \svbz_1,...,\svbz_N}(k) = \piMRC_{\svbx, A_1,...,A_N}(k)$ since $\piMRC$ depends on $\svbz_1,...,\svbz_N$ via $A_1, \cdots,A_N$, $(f)$ follows because $\svbz_k$ is independent of $A_1, \cdots,A_{k-1}, A_{k+1}, \cdots,A_N$ given $A_k$, $(g)$ follows by marginalizing $\svbz_1, \cdots, \svbz_{k-1}, \svbz_{k+1}, \cdots, \svbz_N$, $(h)$ follows by the tower property of expectation, $(i)$ follows by evaluating the expectation over $A_k$, $(j)$ follows because $\Expectation_{\svbz} \big[ \svbz \big| A = 1\big] \coloneqq \Expectation_{\svbz_k} \big[ \svbz_k  \big| A_k = 1\big]$ and $\Expectation_{\svbz} \big[ \svbz \big| A = 0\big] \coloneqq \Expectation_{\svbz_k} \big[ \svbz_k  \big| A_k = 0\big]$ are constants for every $k \in [N]$, $(k)$ follows by marginalizing $A_1, \cdots, A_N$, $(l)$ follows from the definitions of $\Probability(A_K = 1)$ and $\Probability(A_K = 0)$, and $(m)$ follows from rotational symmetry (see the proof of Lemma 4.1 in \cite{BDFKR2018} for details). Therefore, we can write
\begin{align}
\Expectation_{\qMRC} [\hat{\svbx}^\texttt{mrc}] & = \frac{1}{m_\texttt{mrc}} \Expectation_{\qMRC} [\svbz_K ] \stackrel{(a)}{=} \svbx
\end{align}
where $(a)$ follows from \eqref{eq:miracle_zero_bias}.
\end{proof}

\subsection{Utility of Minimal Random Coding simulating \texorpdfstring{$\PrivUnit$}{PrivUnit}}\label{appendix:mrc_pu_ut}

\subsubsection{The scaling factors of \texorpdfstring{$\PrivUnit$}{PrivUnit} and \texorpdfstring{$\MRC$}{MRC} are close when \texorpdfstring{$N$}{N} is of the right order}\label{appendix:scaling_mrc_pu}

In the following Lemma, we show that when the number of candidates $N$ is exponential in $\varepsilon$, then the scaling parameters associated with $\PrivUnit$ and the $\MRC$ scheme simulating $\PrivUnit$ are close.
\begin{lemma}\label{lemma:mrc_privunit_approximation_error}
Let $N$ denote the number of candidates used in the \emph{$\MRC$} scheme. Let $K \sim \piMRC$ where $\piMRC$ is the distribution over the indices $[N]$ associated the \emph{$\MRC$} scheme simulating \emph{$\PrivUnit(\svbx,\gamma, p_0)$}. Consider any $\lambda > 0$.
Then, the scaling factor $m_{\texttt{pu}}$ associated with \emph{$\PrivUnit$} and the scaling factor $m_\texttt{mrc}$ associated with the \emph{$\MRC$} scheme simulating \emph{$\PrivUnit$} are such that
\begin{align}
    m_{\texttt{pu}} - m_\texttt{mrc} \leq \lambda\cdot m_\texttt{mrc}
\end{align}
as long as
\begin{align}
    N \geq   2e^{2\varepsilon}\lp\frac{2 (1+\lambda)}{\lambda \lp p_0 -1/2 \rp}\rp^2 \ln\lp \frac{4(1+\lambda)}{\lambda \lp p_0 -1/2 \rp} \rp.
\end{align}
\end{lemma}
\begin{proof}
Following the proofs of Lemma 4.1 and Proposition 4 in \cite{BDFKR2018}, we can write $m_{\texttt{pu}} = \gamma_+ p_0 + \gamma_-(1-p_0)$ and $m_\texttt{mrc} = \gamma_+ p_{\texttt{mrc}} + \gamma_-(1-p_{\texttt{mrc}})$ where
\begin{align}
    \gamma_+ \eqDef \frac{\lp 1-\gamma^2\rp^\alpha}{2^{d-2}(d-1)\lp B(\alpha, \alpha)-B(\tau; \alpha, \alpha) \rp}, \qquad \mathrm{and} \qquad 
    \gamma_- \eqDef \frac{\lp 1-\gamma^2\rp^\alpha}{2^{d-2}(d-1)\lp B(\tau; \alpha, \alpha) \rp}.
\end{align}
Therefore, we have
\begin{align}
  \frac{1}{m_\texttt{mrc}} - \frac{1}{m_{\texttt{pu}}} & =
  \frac{m_{\texttt{pu}}-m_\texttt{mrc}}{m_{\texttt{pu}} \cdot m_\texttt{mrc}} = \frac{1}{m_{\texttt{pu}}} \lp \frac{(\gamma_+-\gamma_-) \cdot \lp p_0 - p_{\texttt{mrc}} \rp}{\lp (\gamma_+ - \gamma_-) p_{\texttt{mrc}}+\gamma_-\rp} \rp \\
  & = \frac{1}{m_{\texttt{pu}}}\lp\frac{ p_0 - p_{\texttt{mrc}} }{ p_{\texttt{mrc}} +\dfrac{\gamma_-}{\gamma_+-\gamma_-}}\rp \label{eq:mrc_privunit_approximation_error_0}
\end{align}
From \cite{BDFKR2018}, we have $\gamma_- \leq 0 \leq \gamma_+$ and $\lba \gamma_+ \rba \geq \lba \gamma_- \rba$. These inequalities imply $\frac{\gamma_-}{\gamma_+-\gamma_-} \geq -\frac{1}{2}$. Plugging this in \eqref{eq:mrc_privunit_approximation_error_0}, we have
\begin{align}
    \frac{1}{m_\texttt{mrc}} - \frac{1}{m_{\texttt{pu}}} \leq \frac{1}{m_{\texttt{pu}}}\lp\frac{ p_0 - p_{\texttt{mrc}} }{ p_{\texttt{mrc}} - 1/2}\rp = \frac{1}{m_{\texttt{pu}}}\lp\frac{ 1 }{ \dfrac{p_0 - 1/2}{p_0 - p_{\texttt{mrc}}}-1}\rp
    \label{eq:mrc_privunit_approximation_error_1}
\end{align}
We will now upper bound $\dfrac{p_0 - p_{\texttt{mrc}}}{p_0 - 1/2}$. We start by obtaining convenient expressions for $p_{\texttt{mrc}}$ and $p_0$. To compute $p_{\texttt{mrc}} = \Probability(\svbz_K \in \msf{Cap}_{\svbx})$, recall that $\theta$ denotes the fraction of candidates that belong inside the $\msf{Cap}_{\svbx}$. Let $c_1(\varepsilon, d)$ and $c_2(\varepsilon, d)$ be as defined in \eqref{eq:privunit_density}. Let $\bar{c}_1(\varepsilon, d) = c_1(\varepsilon, d) \times A(1,d)$ and $\bar{c}_2(\varepsilon, d) = c_2(\varepsilon, d) \times A(1,d)$. It is easy to see from Algorithm \ref{alg:privunit} and \eqref{eq:privunit_density} that $\Probability(\svbz_k \in \msf{Cap}_{\svbx}) =  \bar{c}_1(\varepsilon, d) / p_0$. 
Further, since $\svbz_k$ are generated uniformly at random,
$$ \theta \sim \frac{1}{N} \msf{Binom}\lp N, \frac{ \bar{c}_1(\varepsilon, d)}{p_0}\rp,$$ so we have
\begin{align}
    p_{\texttt{mrc}} = \Pr\lbp \svbz_K \in \msf{Cap}_{\svbx} \rbp  & = \E\lb \Pr\lbp \svbz_K \in \msf{Cap}_{\svbx}| \theta \rbp \rb  \\
    & \stackrel{(a)}{=} \E\lb \frac{ \bar{c}_1(\varepsilon, d) \theta}{ \bar{c}_1(\varepsilon, d) \theta +  \bar{c}_2(\varepsilon, d)(1-\theta)} \rb \\
    & = \frac{ \bar{c}_1(\varepsilon, d)}{ \bar{c}_1(\varepsilon, d)- \bar{c}_2(\varepsilon, d)}\E\lb \frac{( \bar{c}_1(\varepsilon, d) -  \bar{c}_2(\varepsilon, d))\theta}{( \bar{c}_1(\varepsilon, d) -  \bar{c}_2(\varepsilon, d))\theta +  \bar{c}_2(\varepsilon, d)} \rb \\
    & \stackrel{(b)}{=} \frac{ \bar{c}_1(\varepsilon, d) \bar{c}_2(\varepsilon, d)}{( \bar{c}_1(\varepsilon, d)- \bar{c}_2(\varepsilon, d))^2}\E\lb \frac{ \bar{c}_1(\varepsilon, d)- \bar{c}_2(\varepsilon, d)}{ \bar{c}_2(\varepsilon, d)}-\frac{1}{\theta+\dfrac{ \bar{c}_2(\varepsilon, d)}{ \bar{c}_1(\varepsilon, d)- \bar{c}_2(\varepsilon, d)}} \rb
\label{eq:mrc_privunit_approximation_error_2}
\end{align}
where $(a)$ follows from \eqref{eq:special_pi_mrc} because $\qPU$ is a cap-based mechanism and $(b)$ follows by simple manipulations. 

To compute $p_0$, observe that we have the following relationship between $ \bar{c}_1(\varepsilon, d)$, $ \bar{c}_2(\varepsilon, d)$, and $p_0$ from \eqref{eq:privunit_density}:
\begin{align}
    \frac{ p_0}{\bar{c}_1(\varepsilon, d)} + \frac{ 1-p_0}{\bar{c}_2(\varepsilon, d)} = 1 \label{eq:c1c2p}
\end{align}
Using this and with some simple manipulations, we have 
\begin{align}
p_0
& = \frac{ \bar{c}_1(\varepsilon, d) \bar{c}_2(\varepsilon, d)}{( \bar{c}_1(\varepsilon, d)- \bar{c}_2(\varepsilon, d))^2} \lp\frac{ \bar{c}_1(\varepsilon, d)- \bar{c}_2(\varepsilon, d)}{ \bar{c}_2(\varepsilon, d)}-\frac{1}{\E\lb \theta \rb+\dfrac{ \bar{c}_2(\varepsilon, d)}{ \bar{c}_1(\varepsilon, d)- \bar{c}_2(\varepsilon, d)}}\rp \label{eq:mrc_privunit_approximation_error_3}
\end{align}
From \eqref{eq:mrc_privunit_approximation_error_2} and \eqref{eq:mrc_privunit_approximation_error_3}, we have
\begin{align}
    p_0 - p_{\texttt{mrc}} 
    &= \frac{ \bar{c}_1(\varepsilon, d) \bar{c}_2(\varepsilon, d)}{( \bar{c}_1(\varepsilon, d)- \bar{c}_2(\varepsilon, d))^2} \lp \E\lb \frac{1}{\theta+\dfrac{ \bar{c}_2(\varepsilon, d)}{ \bar{c}_1(\varepsilon, d)- \bar{c}_2(\varepsilon, d)}} - \frac{1}{\E[\theta]+\dfrac{ \bar{c}_2(\varepsilon, d)}{ \bar{c}_1(\varepsilon, d)- \bar{c}_2(\varepsilon, d)}} \rb \rp\\
    & = \frac{ \bar{c}_1(\varepsilon, d) \bar{c}_2(\varepsilon, d)}{( \bar{c}_1(\varepsilon, d)- \bar{c}_2(\varepsilon, d))^2} \lp \E\lb \frac{\E[\theta]-\theta}{\lp \theta+\dfrac{ \bar{c}_2(\varepsilon, d)}{ \bar{c}_1(\varepsilon, d)- \bar{c}_2(\varepsilon, d)} \rp \lp \E[\theta]+\dfrac{ \bar{c}_2(\varepsilon, d)}{ \bar{c}_1(\varepsilon, d)- \bar{c}_2(\varepsilon, d)} \rp} \rb \rp.
\end{align} 
Now, using the Hoeffding's inequality, we have $\Probability\lbp \lba \theta - \E[\theta] \rba \geq \sqrt{\frac{\ln \lp 2/\beta \rp}{2N}} \rbp \leq \beta$. Conditioned on the event $\lbp \lba \theta - \E[\theta] \rba \leq \sqrt{\frac{\ln \lp 2/\beta \rp}{2N}} \rbp$ and using the fact that $\lba p_0 - p_{\texttt{mrc}}\rba\leq 1$, we have
\begin{align}
    p_0   -   p_{\texttt{mrc}}  &  \leq   \dfrac{ \bar{c}_1(\varepsilon, d) \bar{c}_2(\varepsilon, d)}{( \bar{c}_1(\varepsilon, d)   -   \bar{c}_2(\varepsilon, d))^2}  \times \\
    & \lp \frac{\sqrt{\frac{\ln \lp 2/\beta \rp}{2N}}}{\lp \dfrac{p_0}{ \bar{c}_1(\varepsilon, d)}   -   \sqrt{\dfrac{\ln \lp 2/\beta \rp}{2N}}   +   \dfrac{ \bar{c}_2(\varepsilon, d)}{ \bar{c}_1(\varepsilon, d)   -   \bar{c}_2(\varepsilon, d)} \rp \lp \dfrac{p_0}{ \bar{c}_1(\varepsilon, d)}   +  \dfrac{ \bar{c}_2(\varepsilon, d)}{ \bar{c}_1(\varepsilon, d)   -   \bar{c}_2(\varepsilon, d)} \rp} \rp \hspace{-1mm} + \hspace{-1mm} \beta \label{eq:mrc_privunit_approximation_error_4}
    \end{align}
    where we have also plugged in $\E[\theta] = \dfrac{p_0}{ \bar{c}_1(\varepsilon,d) }$. Now, we can lower bound $\lp \dfrac{p_0}{ \bar{c}_1(\varepsilon, d)}+\dfrac{ \bar{c}_2(\varepsilon, d)}{ \bar{c}_1(\varepsilon, d)- \bar{c}_2(\varepsilon, d)} \rp$ as follows:
    \begin{align}
        \lp \frac{p_0}{ \bar{c}_1(\varepsilon, d)}+\frac{ \bar{c}_2(\varepsilon, d)}{ \bar{c}_1(\varepsilon, d)- \bar{c}_2(\varepsilon, d)} \rp \stackrel{(a)}{\geq} \frac{ \bar{c}_2(\varepsilon, d)}{ \bar{c}_1(\varepsilon, d)- \bar{c}_2(\varepsilon, d)}
         \stackrel{(b)}{\geq} \frac{1}{\exp(\varepsilon)-1}
    \end{align}
    where $(a)$ follows by lower bounding $ p_0/\bar{c}_1(\varepsilon, d)$ by 0 and $(b)$ follows because we have $ \bar{c}_1(\varepsilon,d)/  \bar{c}_2(\varepsilon,d) \leq \exp(\varepsilon)$. Further, if we pick $N \geq 2 \ln \lp 2/\beta \rp \lp\exp(\varepsilon)-1\rp^2$, then
    \begin{align}
        \sqrt{\frac{\ln \lp 2/\beta \rp}{2N}} \leq \frac{1}{2} \times \frac{1}{\exp(\varepsilon)-1}\leq \frac{1}{2}   \lp \frac{p_0}{ \bar{c}_1(\varepsilon, d)}+\frac{ \bar{c}_2(\varepsilon, d)}{ \bar{c}_1(\varepsilon, d)- \bar{c}_2(\varepsilon, d)} \rp. \label{eq:mrc_privunit_approximation_error_5}
    \end{align}
    Using \eqref{eq:mrc_privunit_approximation_error_5} in \eqref{eq:mrc_privunit_approximation_error_4}, we have 
    \begin{align}
p_0 - p_{\texttt{mrc}} & \leq  \frac{ \bar{c}_1(\varepsilon, d) \bar{c}_2(\varepsilon, d)}{( \bar{c}_1(\varepsilon, d)- \bar{c}_2(\varepsilon, d))^2} \times \\
& \qquad \lp \frac{2\sqrt{\frac{\ln \lp 2/\beta \rp}{2N}}}{ \lp \dfrac{p_0}{ \bar{c}_1(\varepsilon, d)}+\dfrac{ \bar{c}_2(\varepsilon, d)}{ \bar{c}_1(\varepsilon, d)- \bar{c}_2(\varepsilon, d)} \rp  \lp \dfrac{p_0}{ \bar{c}_1(\varepsilon, d)}+\dfrac{ \bar{c}_2(\varepsilon, d)}{ \bar{c}_1(\varepsilon, d)- \bar{c}_2(\varepsilon, d)} \rp} \rp + \beta \\
& = \lp \frac{2 \bar{c}_1(\varepsilon, d) \bar{c}_2(\varepsilon, d) \sqrt{\frac{\ln \lp 2/\beta \rp}{2N}}}{\lp  p_0 \lp1-\dfrac{ \bar{c}_2(\varepsilon, d)}{ \bar{c}_1(\varepsilon, d)}\rp + \bar{c}_2(\varepsilon, d) \rp^2} \rp + \beta\\
& \stackrel{(a)}{\leq}\lp \frac{2 \bar{c}_1(\varepsilon, d)}{ \bar{c}_2(\varepsilon, d)}  \sqrt{\frac{\ln \lp 2/\beta \rp}{2N}} \rp+ \beta \stackrel{(b)}{\leq}\lp 2 \exp(\varepsilon) \sqrt{\frac{\ln \lp 2/\beta \rp}{2N}} \rp+ \beta \stackrel{(c)}{\leq} \frac{\lambda(p_0 - 1/2)}{1+\lambda}. \label{eq:mrc_privunit_approximation_error_6}
\end{align}
where $(a)$ follows because $p_0 \lp1-\frac{ \bar{c}_2(\varepsilon, d)}{ \bar{c}_1(\varepsilon, d)}\rp \geq 0$, $(b)$ follows because we have $ \bar{c}_1(\varepsilon,d)/  \bar{c}_2(\varepsilon,d) \leq \exp(\varepsilon)$ and $(c)$ follows if we pick 
\begin{align}
    \beta & \leq \frac{\lambda(p_0 - 1/2)}{2(1+\lambda)}, \\
    N & \geq\frac{2\exp(2\varepsilon)\ln \lp 2/\beta \rp}{\lp \dfrac{\lambda(p_0 - 1/2)}{1+\lambda} - \beta \rp^2}
    = 2\exp(2\varepsilon)\lp\frac{2 (1+\lambda)}{\lambda \lp p_0 -1/2 \rp}\rp^2 \ln\lp \frac{4(1+\lambda)}{\lambda \lp p_0 -1/2 \rp} \rp. \label{eq:mrc_privunit_approximation_error_7}
\end{align}
Further, it is easy to verify that \eqref{eq:mrc_privunit_approximation_error_5} holds since the choice of $N$ in \eqref{eq:mrc_privunit_approximation_error_7} is such that $N \geq \frac{1}{2} \ln \lp 2/\beta \rp \lp\exp(\varepsilon)-1\rp^2$. Now, rearranging \eqref{eq:mrc_privunit_approximation_error_6} gives us an upper bound on $\dfrac{p_0 - p_{\texttt{mrc}}}{p_0 -1/2}$, i.e.,
\begin{align}
    \frac{p_0 - p_{\texttt{mrc}}}{p_0 -1/2} \leq \frac{\lambda}{1+\lambda}. \label{eq:mrc_privunit_approximation_error_8}
\end{align}
Using \eqref{eq:mrc_privunit_approximation_error_8} in \eqref{eq:mrc_privunit_approximation_error_1}, we have
\begin{align}
    \frac{1}{m_\texttt{mrc}} - \frac{1}{m_{\texttt{pu}}} \leq \frac{\lambda}{m_{\texttt{pu}}}. \label{eq:mrc_privunit_approximation_error_9}
\end{align}
Rearranging \eqref{eq:mrc_privunit_approximation_error_9} completes the proof.
\end{proof}

\subsubsection{Relationship between mean squared errors associated with \texorpdfstring{$\PrivUnit$}{PrivUnit} and \texorpdfstring{$\MRC$}{MRC} simulating \texorpdfstring{$\PrivUnit$}{PrivUnit}}\label{appendix:mrc_pu_scaling_mse}
In the following Proposition,
we show that if the scaling factor $m_{\texttt{mrc}}$ is close to the scaling parameter $m_{\texttt{pu}}$, then the mean squared error associated with $\MRC$ simulating $\PrivUnit$ (i.e., $\Expectation_{\qMRC} \big[ \lV  \hat{\svbx}^\texttt{mrc} - \svbx \rV^2_2  \big]$) is close to the mean squared error associated with $\PrivUnit$ (i.e., $\Expectation_{\qPU} \big[ \lV  \hat{\svbx}^\texttt{pu} - \svbx \rV^2_2  \big]$).
\begin{proposition}\label{proposition:mrc_mse_wrt_privunit}
Let $\qPU(\svbz | \svbx)$ be the  $\varepsilon$-LDP \emph{$\PrivUnit$} mechanism with parameters $p_0$ and $\gamma$ and estimator $\hat{\svbx}^{\texttt{pu}}$. Let $\qMRC(\svbz|\svbx)$ denote the \emph{$\MRC$} privatization mechanism simulating \emph{$\PrivUnit$} with $N$ candidates and estimator $\hat{\svbx}^{\texttt{mrc}}$.
Let $m_{\texttt{pu}}$ denote the scaling factor  associated with \emph{$\PrivUnit$} and $m_\texttt{mrc}$ denote the scaling factor  associated with the \emph{$\MRC$} scheme simulating \emph{$\PrivUnit$}. Consider any $\lambda > 0$. If $m_{\texttt{pu}} - m_\texttt{mrc} \leq \lambda \cdot m_\texttt{mrc}$, then 
  \begin{align}
    \Expectation_{\qMRC} \big[ \lV  \hat{\svbx}^\texttt{mrc} - \svbx \rV^2_2  \big]  \leq  \lp 1+\lambda \rp^2  \Expectation_{\qPU}\big[\|\hat{\svbx}^{\texttt{pu}} - \svbx\|^2\big] + 2(1+\lambda)(2+\lambda) \sqrt{\Expectation_{\qPU}\big[\|\hat{\svbx}^{\texttt{pu}} - \svbx\|^2\big]} + (2+\lambda)^2.
\end{align}
\end{proposition}
\begin{proof}
We will start by upper bounding $1/m_{\texttt{pu}}$ in terms of $\Expectation_{\qPU}\big[\|\hat{\svbx}^{\texttt{pu}} - \svbx\|^2\big]$. First, observe that
\begin{align}
    \|\hat{\svbx}^{\texttt{pu}} - \svbx\|  \stackrel{(a)}{\geq} \lV \hat{\svbx}^{\texttt{pu}}\rV - \lV \svbx \rV  \stackrel{(b)}{\geq}
    \frac{1}{m_{\texttt{pu}}}  - 1  \label{eq:mrc_privunit_variance_1}
\end{align}
where $(a)$ follows from the triangle inequality and $(b)$ follows because $\|\hat{\svbx}^{\texttt{pu}}\| = 1/m_{\texttt{pu}}$ and $\|\svbx\| \leq 1$. Next, we have
\begin{align}
    \frac{1}{m_{\texttt{pu}}} = \frac{1}{m_{\texttt{pu}}} - 1 + 1 \stackrel{(a)}{\leq} \sqrt{\Expectation_{\qPU}\big[\|\hat{\svbx}^{\texttt{pu}} - \svbx\|^2\big]} + 1 \label{eq:mrc_privunit_variance_2}
\end{align}
where $(a)$ follows from \eqref{eq:mrc_privunit_variance_1}. We will now upper bound $\Expectation_{\qMRC}[\| \hat{\svbx}^\texttt{mrc} - \svbx\|^2]$. We have
\begin{align}
    \Expectation_{\qMRC}[\| \hat{\svbx}^\texttt{mrc} - \svbx\|^2] 
    & = \Expectation_{\qMRC} [\| \hat{\svbx}^\texttt{mrc}\|^2] + \lV \svbx \rV^2_2 - 2\lan \Expectation_{\qMRC} [\hat{\svbx}^\texttt{mrc}], \svbx\ran \\
    & \stackrel{(a)}{\leq} \Expectation_{\qMRC} [\| \hat{\svbx}^\texttt{mrc}\|^2] + \lV \svbx \rV^2_2 + 2\sqrt{ \Expectation_{\qMRC} [\lV \hat{\svbx}^\texttt{mrc} \rV^2] \cdot \lV \svbx \rV^2 } \\
    & \stackrel{(b)}{\leq} \lp\frac{1}{m_\texttt{mrc}}\rp^2  + 1 + \frac{2}{m_\texttt{mrc}}\\
    & \stackrel{(c)}{\leq} \lp\frac{1+\lambda}{m_{\texttt{pu}}}\rp^2 + 1 + \frac{2(1+\lambda)}{m_{\texttt{pu}}} \\
    & \stackrel{(d)}{\leq}
     \lp 1+\lambda \rp^2  \Expectation_{\qPU}\big[\|\hat{\svbx}^{\texttt{pu}} - \svbx\|^2\big] + 2(1+\lambda)(2+\lambda) \sqrt{\Expectation_{\qPU}\big[\|\hat{\svbx}^{\texttt{pu}} - \svbx\|^2\big]} + (2+\lambda)^2
\end{align}
where $(a)$ follows from Cauchy–Schwarz inequality, $(b)$ follows because $\|\hat{\svbx}^\texttt{mrc}\| = 1/m_\texttt{mrc}$ and $\|\svbx\| \leq 1$, (c) follows from Lemma~\ref{lemma:mrc_privunit_approximation_error} (which shows $m_{\texttt{pu}} - m_\texttt{mrc} \leq \lambda\cdot m_\texttt{mrc}$), and $(d)$ follows using \eqref{eq:mrc_privunit_variance_2} and some simple manipulations.
\end{proof}

In the following Lemma, we show that with on the order of $\varepsilon$-bits of communication, the mean squared error associated with $\MRC$ simulating $\PrivUnit$ (i.e., $\Expectation_{\qMRC} \big[ \lV  \hat{\svbx}^\texttt{mrc} - \svbx \rV^2_2  \big]$) is close to the mean squared error associated with $\PrivUnit$ (i.e., $\Expectation_{\qPU} \big[ \lV  \hat{\svbx}^\texttt{pu} - \svbx \rV^2_2  \big]$).

\begin{restatable}{lemma}{mrcprivunit}\label{theorem:mrc_privunit}
Let $\qPU(\svbz | \svbx)$ be the  $\varepsilon$-LDP \emph{$\PrivUnit$} mechanism with parameters $p_0$ and $\gamma$ and estimator $\hat{\svbx}^{\texttt{pu}}$. Let $\qMRC(\svbz|\svbx)$ denote the \emph{$\MRC$} privatization mechanism simulating \emph{$\PrivUnit$} with $N$ candidates and estimator $\hat{\svbx}^{\texttt{mrc}}$. 
Consider any $\lambda > 0$. Then,
\begin{align}
   \Expectation_{\qMRC} \big[ \lV  \hat{\svbx}^\texttt{mrc} - \svbx \rV^2_2  \big] 
     \leq  \lp 1+\lambda \rp^2  \Expectation_{\qPU}\big[\|\hat{\svbx}^{\texttt{pu}} -\svbx\|^2\big] + 2(1+\lambda)(2+\lambda) \sqrt{\Expectation_{\qPU}\big[\|\hat{\svbx}^{\texttt{pu}} - \svbx\|^2\big]} + (2+\lambda)^2
\end{align}
as long as 
\begin{align}
    N \geq   2e^{2\varepsilon}\lp\frac{2 (1+\lambda)}{\lambda \lp p_0 -1/2 \rp}\rp^2 \ln\lp \frac{4(1+\lambda)}{\lambda \lp p_0 -1/2 \rp} \rp.
\end{align}
\end{restatable}
\begin{proof}
The proof follows from Proposition \ref{proposition:mrc_mse_wrt_privunit} and Lemma \ref{lemma:mrc_privunit_approximation_error}.
\end{proof}

\subsubsection{Simulating \texorpdfstring{$\PrivUnit$}{PrivUnit} using Minimal Random Coding }\label{appendix:mrc_pu_utility}
The following Theorem shows that, for mean estimation, $\MRC$ can simulate $\PrivUnit$ in a near-lossless manner (when $n$ is large and $\lambda$ is small) while only using on the order of $\varepsilon$ bits of communication.

\begin{restatable}{theorem}{mrcpu}
\label{thm:me_mrc_pu}
Let $r_{\msf{ME}} \lp \hat{\svbmu}^\texttt{pu}, \qPU \rp$ and $r_{\msf{ME}} \lp \hat{\svbmu}^\texttt{mrc}, \qMRC \rp$ be the empirical mean estimation error for \emph{$\PrivUnit$} with parameter $p_0$ and \emph{$\MRC$} simulating \emph{$\PrivUnit$} with $N$ candidates respectively. Consider any $\lambda > 0$. Then,
\begin{align}
    r_{\msf{ME}} \lp \hat{\svbmu}^\texttt{mrc}, \qMRC \rp \leq  
    \lp 1+\lambda \rp^2 r_{\msf{ME}} \lp \hat{\svbmu}^\texttt{pu}, \qPU \rp + 2(1+\lambda)(2+\lambda) \sqrt{\frac{r_{\msf{ME}} \lp \hat{\svbmu}^\texttt{pu}, \qPU \rp}{n}} + \frac{(2+\lambda)^2}{n}.
\end{align}
as long as 
\begin{align}\label{eq:N_bdd_mrc_pu}
    N \geq   2e^{2\varepsilon}\lp\frac{2 (1+\lambda)}{\lambda \lp p_0 -1/2 \rp}\rp^2 \ln\lp \frac{4(1+\lambda)}{\lambda \lp p_0 -1/2 \rp} \rp.
\end{align}
\end{restatable}
\begin{proof}
The proof follows directly from Lemma~\ref{theorem:mrc_privunit} since for all $i \in [n]$, $\hat{\svbx}^{\texttt{mrc}}_i$ are independent of each other as well as unbiased.
\end{proof}

\subsection{Empirical Comparisons}
\label{appendix:mrc_pu_emp}
In this section, we compare $\MRC$ simulating $\PrivUnit$ (using its approximate DP guarantee) against $\PrivUnit$ and SQKR for mean estimation with $d = 500$ and $n = 5000$. We use the same data generation scheme described in Section \ref{subsec:mmrc_privunit_empirical} and set $\delta = 10^{-6}$. As before, SQKR uses $\#$-bits $= \varepsilon$ because it leads to a poor performance if $\#$-bits $ > \varepsilon$. We show the privacy-accuracy tradeoffs for these three methods in Figure \ref{fig:a_mean}. We see that $\MRC$ simulating $\PrivUnit$ can attain the accuracy of the uncompressed $\PrivUnit$ for the range of  $\varepsilon$'s typically considered by LDP mechanisms while only using $(3\varepsilon/ \ln 2) + 6$ bits. In comparison with the results from Section \ref{subsec:mmrc_privunit_empirical}, the results in this section come with an approximate guarantee ($\delta = 10^{-6}$) and with a higher number of bits of communication. In other words, along with the obvious gains of pure privacy instead of approximate privacy, $\MMRC$ results in a lower communication cost (and therefore a lower computation cost) compared to $\MRC$.
\begin{figure}[h]
\centering
\includegraphics[width=0.45\linewidth]{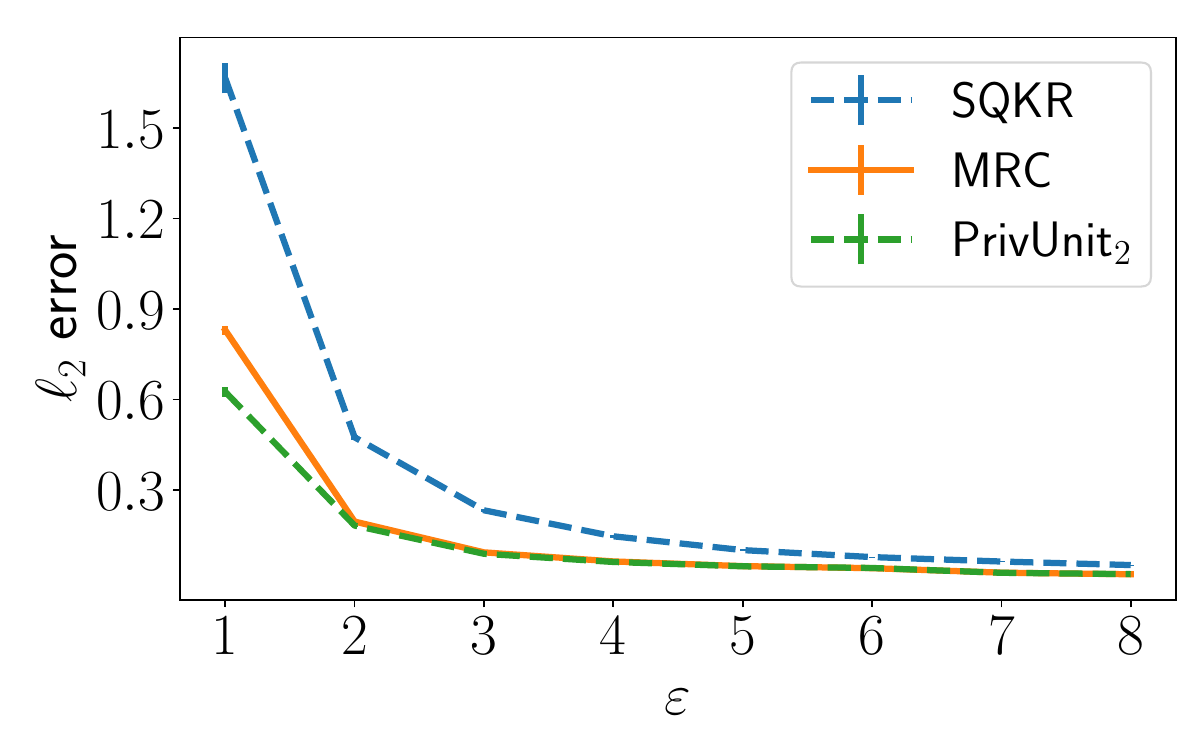}
\caption{Comparing $\PrivUnit$, $\MRC$ simulating $\PrivUnit$ and SQKR for mean estimation in terms of $\ell_2$ error vs $\varepsilon$ with $d = 500$, $n = 5000$, and $\#$bits $= (3\varepsilon/ \ln 2) + 6$.}
\label{fig:a_mean}
\end{figure}

\section{Modified Minimal Random Coding Simulating \texorpdfstring{$\PrivUnit$}{PrivUnit}}
\label{appendix:mmrc_pu}
In this section, we prove Lemma \ref{lemma:mmrc_privunit_bias} (in Appendix \ref{appendix:mmrc_pu_debias}) and Theorem \ref{thm:me_mmrc_pu} (in Appendix \ref{appendix:mmrc_pu_ut}). To prove Theorem \ref{thm:me_mmrc_pu}, first, in Appendix \ref{appendix:scaling_mmrc_pu}, we show that when the number of candidates $N$ is exponential in $\varepsilon$, the scaling factor $m_{\texttt{mmrc}}$ is close to the scaling parameter associated with $\PrivUnit$ (i.e., $m_{\texttt{pu}}$). Next, in Appendix \ref{appendix:mmrc_pu_scaling_mse}, we provide the relationship between the mean squared error associated with $\MMRC$ simulating $\PrivUnit$ and the mean squared error associated with $\PrivUnit$.
Finally, in Appendix \ref{appendix:mmrc_pu_emp}, we provide some empirical comparisons in addition to the ones in Section \ref{subsec:mmrc_privunit_empirical} between $\MMRC$ simulating $\PrivUnit$ and $\PrivUnit$.


\subsection{Unbiased Modified Minimal Random Coding simulating \texorpdfstring{$\PrivUnit$}{PrivUnit}}\label{appendix:mmrc_pu_debias}

Consider the $\PrivUnit$ $\varepsilon$-LDP mechanism $\qPU$ described in Section \ref{sec:preliminaries} with parameters $p_0$ and $\gamma$. $\PrivUnit$ is a cap-based mechanism with $\msf{Cap}_{\svbx} = \{\svbz \in \sphere^{d-1} \mid
\lan \svbz, \svbx  \ran \ge \gamma\}$ as discussed in Appendix \ref{appendix:privunit}. Let $\piMMRC$ be the distribution and $\svbz_1, \svbz_2,...,\svbz_N$ be the candidates obtained from Algorithm \ref{alg:mmrc} when the reference distribution is $\Unif(\sphere^{d-1})$. Let  $K \sim \piMMRC(\cdot)$. Define $p_{\texttt{mmrc}} \coloneqq \Probability(\svbz_K \in \msf{Cap}_{\svbx})$
to be the probability with which the sampled candidate $\svbz_K$ belongs to the spherical cap associated with $\PrivUnit$.
Define $m_{\texttt{mmrc}}$ as the scaling factor in \eqref{eq:m} when $p_0$ in \eqref{eq:m} is replaced by $p_{\texttt{mmrc}}$. Define $\hat{\svbx}^\texttt{mmrc} \coloneqq \svbz_K / m_\texttt{mmrc}$ as the estimator of the $\MMRC$ mechanism simulating $\PrivUnit$. 
\mmrcprivunitbias*
\begin{proof}
The proof is similar to the proof of Lemma \ref{theorem:mrc_privunit_bias}.
\end{proof}

\subsection{Utility of Modified Minimal Random Coding simulating \texorpdfstring{$\PrivUnit$}{PrivUnit}}\label{appendix:mmrc_pu_utility}

\subsubsection{The scaling factors of \texorpdfstring{$\PrivUnit$}{PrivUnit} and \texorpdfstring{$\MMRC$}{MMRC} are close when \texorpdfstring{$N$}{N} is of the right order}\label{appendix:scaling_mmrc_pu}
In the following Lemma, we show that when the number of candidates $N$ is exponential in $\varepsilon$, then the scaling parameters associated with $\PrivUnit$ and the $\MMRC$ scheme simulating $\PrivUnit$ are close.

\begin{lemma}\label{lemma:mmrc_privunit_approximation_error}
Let $N$ denote the number of candidates used in the \emph{$\MMRC$} scheme. Let $K \sim \piMMRC$ where $\piMMRC$ is the distribution over the indices $[N]$ associated the \emph{$\MMRC$} scheme simulating \emph{$\PrivUnit(\svbx,\gamma, p_0)$}. Consider any $\lambda > 0$.
Then, the scaling factor $m_{\texttt{pu}}$ associated with \emph{$\PrivUnit$} and the scaling factor $m_\texttt{mmrc}$ associated with the \emph{$\MMRC$} scheme simulating \emph{$\PrivUnit$} are such that
\begin{align}
    m_{\texttt{pu}} - m_\texttt{mmrc} \leq \lambda\cdot m_\texttt{mmrc}
\end{align}
as long as
\begin{align}
    N \geq   \frac{e^{2\varepsilon}}{2} \lp\frac{2 (1+\lambda)}{\lambda \lp p_0 -1/2 \rp}\rp^2 \ln\lp \frac{4(1+\lambda)}{\lambda \lp p_0 -1/2 \rp} \rp.
\end{align}
\end{lemma}
\begin{proof}
The proof follows a structure similar to the proof of Lemma \ref{lemma:mrc_privunit_approximation_error}. As in the proof of Lemma \ref{lemma:mrc_privunit_approximation_error}, we have
\begin{align}
    \frac{1}{m_\texttt{mmrc}} - \frac{1}{m_{\texttt{pu}}} \leq \frac{1}{m_{\texttt{pu}}}\lp\frac{ 1 }{ \dfrac{p_0 -1/2}{p_0 - p_\texttt{mmrc}}-1}\rp
    \label{eq:mmrc_privunit_approximation_error_1}
\end{align}
We will now upper bound $\dfrac{p_0 - p_\texttt{mmrc}}{p_0 -1/2}$. We start by obtaining expressions for $p_\texttt{mmrc}$ and $p_0$.

To compute $p_\texttt{mmrc} \coloneqq \Pr\lbp \svbz_K \in \msf{Cap}_{\svbx} \rbp$, recall that $\theta$ denotes the fraction of candidates that belong inside the $\msf{Cap}_{\svbx}$. Let $c_1(\varepsilon, d)$ and $c_2(\varepsilon, d)$ be as defined in \eqref{eq:privunit_density}. Let $\bar{c}_1(\varepsilon, d) = c_1(\varepsilon, d) \times A(1,d)$ and $\bar{c}_2(\varepsilon, d) = c_2(\varepsilon, d) \times A(1,d)$. It is easy to see from Algorithm \ref{alg:privunit} and \eqref{eq:privunit_density} that $\Probability(\svbz_k \in \msf{Cap}_{\svbx}) =  \bar{c}_1(\varepsilon, d) / p_0$. 
Further, since $\svbz_k$ are generated uniformly at random,
$$ \theta \sim \frac{1}{N} \msf{Binom}\lp N, \frac{ \bar{c}_1(\varepsilon, d)}{p_0}\rp,$$ so we have
\begin{align}
    p_\texttt{mmrc} = \Pr\lbp \svbz_K \in \msf{Cap}_{\svbx} \rbp  & = \E\lb \Pr\lbp \svbz_K \in \msf{Cap}_{\svbx}| \theta \rbp \rb  \\
    & \stackrel{(a)}{=} \E\bigg[ \frac{\theta  \bar{c}_1(\varepsilon,d) }{ \bar{c}_2(\varepsilon,d) + \E\lb \theta \rb \lp  \bar{c}_1(\varepsilon,d) -  \bar{c}_2(\varepsilon,d) \rp} \times \Indicator \lp \theta \leq \E\lb \theta \rb \rp  \\
    & \qquad + \frac{\E \lb \theta \rb   \bar{c}_1(\varepsilon,d) + (\theta - \E \lb \theta \rb)  \bar{c}_2(\varepsilon,d)}{ \bar{c}_2(\varepsilon,d) + \E\lb \theta \rb \lp  \bar{c}_1(\varepsilon,d) -  \bar{c}_2(\varepsilon,d) \rp}  \times \Indicator \lp \theta > \E\lb \theta \rb \rp \bigg]
    \label{eq:mmrc_privunit_approximation_error_2}
\end{align}
where $(a)$ follows from Algorithm \ref{alg:mmrc}.

Similarly, with some simple manipulations on the definition of $p_0$, we have 
\begin{align}
p_0 
& = \frac{\E\lb \theta \rb  \bar{c}_1(\varepsilon,d)}{ \bar{c}_2(\varepsilon,d) + \E\lb \theta \rb \lp  \bar{c}_1(\varepsilon,d) -  \bar{c}_2(\varepsilon,d) \rp} \label{eq:mmrc_privunit_approximation_error_3}
\end{align}
From \eqref{eq:mmrc_privunit_approximation_error_2} and \eqref{eq:mmrc_privunit_approximation_error_3}, we have
\begin{align}
    p_0 - p_\texttt{mmrc} 
    & =  \frac{\E\lb  \bar{c}_1(\varepsilon,d) (\E[\theta] - \theta) \times \Indicator \lp \theta \leq \E\lb \theta \rb \rp +  \bar{c}_2(\varepsilon,d) (\E[\theta] - \theta) \times \Indicator \lp \theta > \E\lb \theta \rb \rp \rb}{ \bar{c}_2(\varepsilon,d) + \E\lb \theta \rb \lp  \bar{c}_1(\varepsilon,d) -  \bar{c}_2(\varepsilon,d) \rp}
     \\
    & \stackrel{(a)}{\leq} \frac{\E\lb  \bar{c}_1(\varepsilon,d) (\E[\theta] - \theta) \times \Indicator \lp \theta \leq \E\lb \theta \rb \rp \rb}{ \bar{c}_2(\varepsilon,d) + \E\lb \theta \rb \lp  \bar{c}_1(\varepsilon,d) -  \bar{c}_2(\varepsilon,d) \rp}
\end{align}
where $(a)$ follows because $(\E[\theta] - \theta) \times \Indicator \lp \theta > \E\lb \theta \rb \rp \leq 0$. Now, using the Hoeffding's inequality, we have $\Probability\lbp \lba \theta - \E[\theta] \rba \geq \sqrt{\frac{\ln\lp 2/\beta \rp}{2N}} \rbp \leq \beta$. Conditioned on the event $\lbp \lba \theta - \E[\theta] \rba \leq \sqrt{\frac{\ln\lp 2/\beta \rp}{2N}} \rbp$ and using the fact that $\lba p_0 - p_\texttt{mmrc}  \rba\leq 1$, we have
\begin{align}
    p_0 - p_\texttt{mmrc} & \leq \frac{ \bar{c}_1(\varepsilon,d) \sqrt{\frac{\ln(2/\beta)}{2N}}}{ \bar{c}_2(\varepsilon,d) + \E\lb \theta \rb \lp  \bar{c}_1(\varepsilon,d) -  \bar{c}_2(\varepsilon,d) \rp}  + \beta \\
    & \stackrel{(a)}{\leq}  \frac{ \bar{c}_1(\varepsilon,d)}{ \bar{c}_2(\varepsilon,d)} \sqrt{\frac{\ln(2/\beta)}{2N}} + \beta
    \stackrel{(b)}{\leq} \exp(\varepsilon) \sqrt{\frac{\ln(2/\beta)}{2N}} + \beta \stackrel{(c)}{\leq} \frac{\lambda(p_0 - 1/2)}{1+\lambda}. \label{eq:mmrc_privunit_approximation_error_6}
\end{align}
where $(a)$ follows because $\E\lb \theta \rb \geq 0$, $(b)$ follows because $ \bar{c}_1(\varepsilon,d)/ \bar{c}_2(\varepsilon,d) \leq e^{\varepsilon}$, and $(c)$ follows if we pick 
\begin{align*}
    \beta & \leq \frac{\lambda(p_0 - 1/2)}{2(1+\lambda)},\\
    N & \geq\frac{\exp(2\varepsilon) \ln(2/\beta)}{2\lp \frac{\lambda(p_0 - 1/2)}{1+\lambda} - \beta \rp^2}
    = \frac{\exp(2\varepsilon)}{2}\lp\frac{2 (1+\lambda)}{\lambda \lp p_0 -1/2 \rp}\rp^2 \ln\lp \frac{4(1+\lambda)}{\lambda \lp p_0 -1/2 \rp} \rp.
\end{align*}
The rest of the proof is similar to the proof of Lemma \ref{lemma:mrc_privunit_approximation_error}.
\end{proof}

\subsubsection{Relationship between the mean squared errors associated with \texorpdfstring{$\PrivUnit$}{PrivUnit} and \texorpdfstring{$\MMRC$}{MMRC} simulating \texorpdfstring{$\PrivUnit$}{PrivUnit}}\label{appendix:mmrc_pu_scaling_mse}
In the following Proposition,
we show that if the scaling factor $m_{\texttt{mmrc}}$ is close to the scaling parameter $m_{\texttt{pu}}$, then the mean squared error associated with $\MMRC$ simulating $\PrivUnit$ (i.e., $\Expectation_{\qMMRC} \big[ \lV  \hat{\svbx}^\texttt{mmrc} - \svbx \rV^2_2  \big]$) is close to the mean squared error associated with $\PrivUnit$ (i.e., $\Expectation_{\qPU} \big[ \lV  \hat{\svbx}^\texttt{pu} - \svbx \rV^2_2  \big]$).
\begin{proposition}\label{proposition:mmrc_mse_wrt_privunit}
Let $\qPU(\svbz | \svbx)$ be the  $\varepsilon$-LDP \emph{$\PrivUnit$} mechanism with parameters $p_0$ and $\gamma$ and estimator $\hat{\svbx}^{\texttt{pu}}$. Let $\qMMRC(\svbz|\svbx)$ denote the \emph{$\MMRC$} privatization mechanism simulating \emph{$\PrivUnit$} with $N$ candidates and estimator $\hat{\svbx}^{\texttt{mmrc}}$.
Let $m_{\texttt{pu}}$ denote the scaling factor  associated with \emph{$\PrivUnit$} and $m_\texttt{mmrc}$ denote the scaling factor  associated with the \emph{$\MMRC$} scheme simulating \emph{$\PrivUnit$}. Consider any $\lambda > 0$. If $m_{\texttt{pu}} - m_\texttt{mmrc} \leq \lambda \cdot m_\texttt{mmrc}$, then 
  \begin{align}
    \Expectation_{\qMMRC} \big[ \lV  \hat{\svbx}^\texttt{mmrc} - \svbx \rV^2_2  \big]  \leq  \lp 1+\lambda \rp^2  \Expectation_{\qPU}\big[\|\hat{\svbx}^{\texttt{pu}} - \svbx\|^2\big] + 2(1+\lambda)(2+\lambda) \sqrt{\Expectation_{\qPU}\big[\|\hat{\svbx}^{\texttt{pu}} - \svbx\|^2\big]} + (2+\lambda)^2.
\end{align}
\end{proposition}
\begin{proof}
The proof is similar to the proof of Proposition \ref{proposition:mrc_mse_wrt_privunit}.
\end{proof}

In the following Lemma, we show that with on the order of $\varepsilon$-bits of communication, the mean squared error associated with $\MMRC$ simulating $\PrivUnit$ (i.e., $\Expectation_{\qMMRC} \big[ \lV  \hat{\svbx}^\texttt{mmrc} - \svbx \rV^2_2  \big]$) is close to the mean squared error associated with $\PrivUnit$ (i.e., $\Expectation_{\qPU} \big[ \lV  \hat{\svbx}^\texttt{pu} - \svbx \rV^2_2  \big]$).
\begin{restatable}{lemma}{mmrcprivunit}\label{theorem:mmrc_accuracy_privunit}
Let $\qPU(\svbz | \svbx)$ be the  $\varepsilon$-LDP \emph{$\PrivUnit$} mechanism with parameters $p_0$ and $\gamma$ and estimator $\hat{\svbx}^{\texttt{pu}}$. Let $\qMMRC(\svbz|\svbx)$ denote the \emph{$\MMRC$} privatization mechanism simulating \emph{$\PrivUnit$} with $N$ candidates and estimator $\hat{\svbx}^{\texttt{mmrc}}$ as defined above. 
Consider any $\lambda > 0$. Then,
\begin{align}
  \Expectation_{\qMMRC} \big[ \lV  \hat{\svbx}^\texttt{mmrc} - \svbx \rV^2_2  \big] 
     \leq  \lp 1+\lambda \rp^2  \Expectation_{\qPU}\big[\|\hat{\svbx}^{\texttt{pu}} - \svbx\|^2\big]  + 2(1+\lambda)(2+\lambda) \sqrt{\Expectation_{\qPU}\big[\|\hat{\svbx}^{\texttt{pu}} - \svbx\|^2\big]} + (2+\lambda)^2
\end{align}
as long as 
\begin{align}
N \geq   \frac{e^{2\varepsilon}}{2}\lp\frac{2 (1+\lambda)}{\lambda \lp p_0 -1/2 \rp}\rp^2 \ln\lp \frac{4(1+\lambda)}{\lambda \lp p_0 -1/2 \rp} \rp.
\end{align}
\end{restatable}
\begin{proof}
The proof follows from Proposition \ref{proposition:mmrc_mse_wrt_privunit} and Lemma \ref{lemma:mmrc_privunit_approximation_error}.
\end{proof}

\subsubsection{Simulating \texorpdfstring{$\PrivUnit$}{PrivUnit} using Modified Minimal Random Coding}\label{appendix:mmrc_pu_ut}
Now, we provide a proof of Theorem \ref{thm:me_mmrc_pu}.
\mmrcpu*
\begin{proof}
The proof follows directly from Lemma~\ref{theorem:mmrc_accuracy_privunit} since for all $i \in [n]$, $\hat{\svbx}^{\texttt{mmrc}}_i$ are independent of each other as well as unbiased.
\end{proof}

\subsection{Additional Empirical Comparisons}\label{appendix:mmrc_pu_emp}
In Section \ref{subsec:mmrc_privunit_empirical}, we empirically demonstrated the privacy-accuracy-communication tradeoffs of $\MMRC$ simulating $\PrivUnit$ against $\PrivUnit$ and SQKR in terms of $\ell_2$ error vs $\#$bits and $\ell_2$ error vs $\varepsilon$ (see Figure \ref{fig:mean}). In this section, we provide comparisons between these methods in terms of $\ell_2$ error vs $d$ (see Figure \ref{fig:mean_app} (left)) and $\ell_2$ error vs $n$ (see Figure \ref{fig:mean_app} (right)) for a fixed $\varepsilon$ (=6) and a fixed $\#$bits (=11). As before, SQKR uses $\#$bits $= \varepsilon$ for both because it leads to a poor performance if $\#$bits $ > \varepsilon$.
\begin{figure}[h]
\centering
\includegraphics[width=0.45\linewidth]{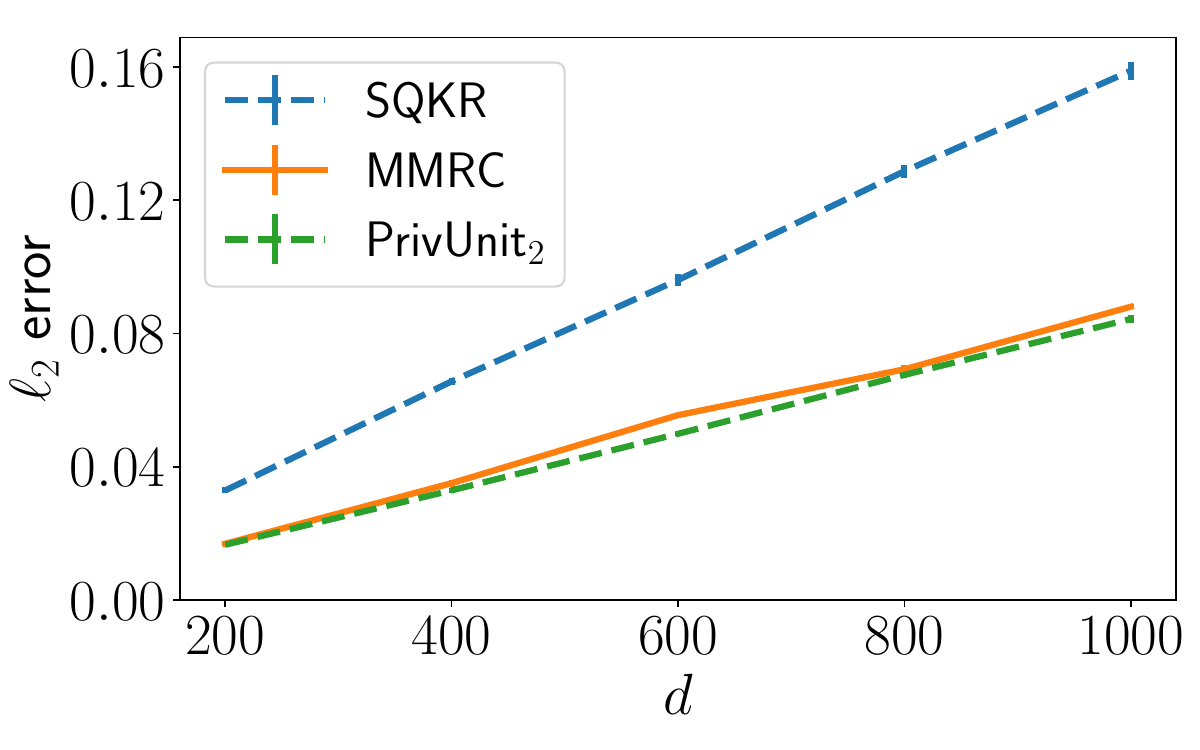} \qquad \includegraphics[width=0.45\linewidth]{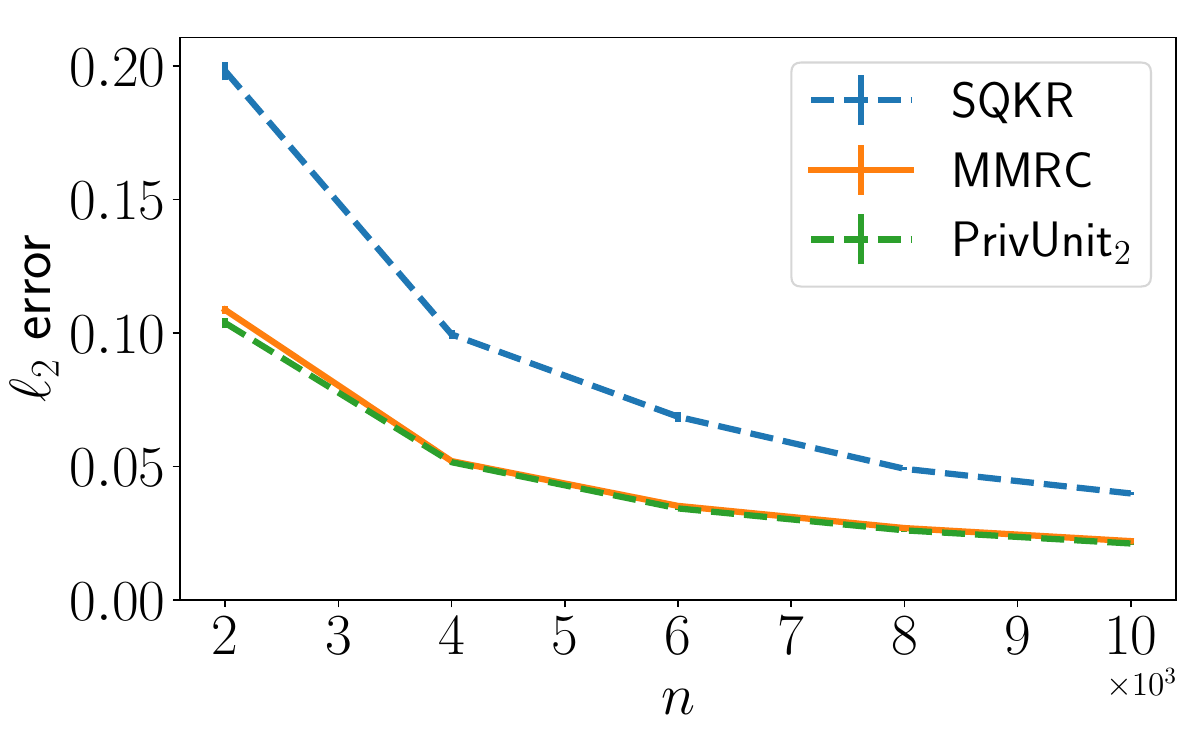}%
\caption{Comparing $\PrivUnit$, $\MMRC$ simulating $\PrivUnit$ and SQKR for mean estimation with $\varepsilon=6$ and $\#$bits $=11$. \textbf{Left:} $\ell_2$ error vs $d$ for $n = 5000$. \textbf{Right:} $\ell_2$ error vs $n$ for $d = 500$.}
\label{fig:mean_app}
\end{figure}
\section{Preliminary on \texorpdfstring{$\SubsetSelection$}{Subset Selection}}
\label{appendix:ss}

In this section, we briefly recap the $\SubsetSelection$ (SS) mechanism proposed in \cite{YB18}.
Let $\svbx = (x_1, x_2,...,x_d) \in \{0, 1\}^d$ be the one-hot representation of an input symbol in $\cX = [d] = \{1,\cdots, d\}$\footnote{With a slight abuse of notation, when context is clear, we sometime use $\svbx = i$ for some $ i \in [d]$ to indicate the one-hot representation of symbol $i$}. Let $\qSS(\svbz|\svbx)$ be the $\SubsetSelection$ mechanism defined in \cite{YB18} where the output alphabet is the set of all $d-$bit binary strings with Hamming weight $s \in [d]$, i.e.,
\begin{align}
\mcal{Z} = \lbp \svbz = (z_1, z_2,...,z_d) \in \{0, 1\}^d: \sum_{i=1}^d z_i = s \rbp. \label{eq:z_ss}    
\end{align}
Given $\svbx \in \cX$,  $\SubsetSelection$ maps it to $\svbz \in \cZ$ with the following conditional probability:
\begin{align}\label{eq:q_ss_app}
    \qSS(\svbz|\svbx = i) \coloneqq \begin{cases}
    \dfrac{e^{\varepsilon}}{{ \binom{d-1}{s-1} }e^{\varepsilon}+{ \binom{d-1}{s} }} & \text{if}\ \svbz \in \cZ_{i}
    \\[10pt]
    \dfrac{1}{{ \binom{d-1}{s-1} }e^{\varepsilon}+{ \binom{d-1}{s} }} & \text{if}\ \svbz \in \cZ \setminus \cZ_{i}
    \end{cases}
\end{align}
where $\cZ_{i} = \lbp \svbz = (z^{(1)}, \cdots ,z^{(d)}) \in \cZ : z^{(i)} = 1 \rbp$ is the set of elements in $\cZ$ with $1$ in the $i^{th}$ location.

\cite{YB18} show that the marginal distribution of $\svbz$ is a linear function of that of $\svbx$. In particular, if we define $p_i \coloneqq \Pr\lbp x_i = 1 \rbp$ for all $i \in [d]$ and let $\svbz \sim \qSS(\cdot |\svbx)$, then \eqref{eq:ss_eq5} is due to (5) in \cite{YB18},
\begin{align}
    \qSS_i \coloneqq \Pr\lbp z_i = 1 \rbp
    & = \frac{\binom{d-1}{s-1} e^\varepsilon p_i + \lp \binom{d-2}{s-2}e^\varepsilon + \binom{d-2}{s-1} \rp (1-p_i)}{\binom{d-1}{s-1} e^\varepsilon + \binom{d-1}{s}} \label{eq:ss_eq5}\\
    & = \frac{s(d  -  s)(e^{\varepsilon}-1)}{(d  -  1)(s(e^{\varepsilon}  -  1)+d)} p_i+  \frac{s((s  -  1)e^{\varepsilon} +(d  -  s))}{(d  -  1)(s(e^{\varepsilon}  -  1)+d)} \\
    & = m_{\texttt{ss}} \cdot p_i + b_{\texttt{ss}}, \label{eq:ss_marginal}
\end{align}
where 
\begin{align}
    m_{\texttt{ss}} \coloneqq \frac{s(d  -  s)(e^{\varepsilon}-1)}{(d  -  1)(s(e^{\varepsilon}  -  1)+d)}, \qquad  b_{\texttt{ss}} \coloneqq  \frac{s((s  -  1)e^{\varepsilon} +(d  -  s))}{(d  -  1)(s(e^{\varepsilon}  -  1)+d)}. \label{eq:scaling_factors_ss}
\end{align} 

The final estimator of $\svbx$ is denoted by $\hat{\svbx}^{\texttt{ss}}$ and is defined as $\frac{1}{m_{\texttt{ss}}} \cdot (\svbz-b_{\texttt{ss}}\cdot\mb{1}_d)$, where 
$\mb{1}_d \eqDef [1,\cdots,1]^\intercal \in \mbb{R}^d$. In other words, $m_{\texttt{ss}}$ and $b_{\texttt{ss}}$ are used de-bias the outcome $\svbz$. The scheme is summarized in Algorithm~\ref{alg:ss}.
\begin{algorithm}
\caption{Subset Selection}
\label{alg:ss}
\begin{algorithmic}
\Require $\svbx \in [d]$, $s \in [d]$.
\State Draw a $s$-hot random vector $\svbz$ according to the distribution $\qSS(\svbz|\svbx)$ in \eqref{eq:q_ss_app}.

\Return $\hat{\svbx}^{\texttt{ss}} = \frac{1}{m_{\texttt{ss}}} \cdot (\svbz-b_{\texttt{ss}}\cdot\mb{1}_d)$
\end{algorithmic}
\end{algorithm}

\subsection{\texorpdfstring{$\SubsetSelection$}{Subset Selection} is unbiased and order-optimal}
The following proposition borrowed from \cite{YB18} shows that the output of the $\SubsetSelection$ mechanism  (a) is unbiased and (b) has order-optimal utility.

\begin{proposition}
Let $\hat{\svbx}^{\texttt{ss}}$ = \emph{$\SubsetSelection(\svbx,s)$} for some $\svbx \in \cX$ and $s \in [d]$. Then, $\E[\hat{\svbx}^{\texttt{ss}}] = \svbx$. Further, the $\ell_2$ estimation error is 
\begin{equation*}
    \E\lb \lV \hat{\svbx}^{\texttt{ss}} -\svbx \rV^2_2\rb = \lp \frac{\lp s(d-2)+1 \rp e^{2\varepsilon}}{(d-s)\lp e^\varepsilon -1\rp^2}+\frac{2(d-2)}{\lp e^\varepsilon -1\rp^2}+ \frac{(d-2)(d-s)+1}{s\lp e^\varepsilon -1\rp^2} - \sum_i p_i^2\rp.
\end{equation*}
Moreover, if we pick $s \coloneqq \lceil \frac{d}{1+e^\varepsilon}\rceil$, then
$$\E\lb \lV \hat{\svbx}^{\texttt{ss}} -\svbx\rV^2_2\rb = \frac{d}{\min\lp e^\varepsilon, \lp e^\varepsilon-1 \rp^2, d \rp}, $$
which is order-optimal.
\end{proposition}

\subsection{\texorpdfstring{$\SubsetSelection$}{Subset Selection} is a cap-based mechanism}\label{appendix:ss_cap}
As discussed in Section \ref{sec:main_results}, $\qSS$ defined in \eqref{eq:q_ss} is a cap-based mechanism with $\msf{Cap}_{\svbx} = \cZ_{\svbx}$, $c_1(\varepsilon, d) = \dfrac{e^{\varepsilon}}{{ \binom{d-1}{s-1} }e^{\varepsilon}+{ \binom{d-1}{s} }}$, and $c_2(\varepsilon, d) = \dfrac{1}{{ \binom{d-1}{s-1} }e^{\varepsilon}+{ \binom{d-1}{s} }}$.

Further, $\Probability_{\svbz \sim \Unif(\cZ)}\lp \svbz \in \cZ_{\svbx} \rp = \dfrac{ \binom{d-1}{s-1}}{ \binom{d}{s}} = \frac{s}{d}$. Therefore,
\begin{align}
    \frac{c_1(\varepsilon, d)}{c_2(\varepsilon, d)} \times \Probability_{\svbz \sim \Unif(\cZ)}\lp \svbz \in \cZ_{\svbx} \rp = e^{\varepsilon} \times \frac{s}{d} \stackrel{(a)}{=}
    \frac{e^{\varepsilon}}{d} \times \lceil \frac{d}{1+e^\varepsilon}\rceil \geq \frac{e^{\varepsilon}}{d} \times \frac{d}{1+e^\varepsilon} \stackrel{(b)}{\geq} \frac{1}{2}
\end{align}
where $(a)$ follows by plugging in $s = \lceil \frac{d}{1+e^\varepsilon}\rceil$ and $(b)$ follows because $\varepsilon \geq 0$.
\section{Simulating \texorpdfstring{$\SubsetSelection$}{Subset Selection} using Minimal Random Coding}\label{appendix:mrc_ss}

In this section, we simulate $\SubsetSelection$ using $\MRC$ analogous to how we simulate $\SubsetSelection$ using $\MMRC$ in Section \ref{sec:frequency_estimation}. 
First, in Appendix \ref{appendix:debias_mrc_ss}, we provide an unbiased estimator for $\MRC$ simulating $\SubsetSelection$.
Next, in Appendix \ref{appendix:mrc_ss_ut} we provide the utility guarantee associated with $\MRC$ simulating $\SubsetSelection$. 
To do that, first, in Appendix \ref{appendix:scaling_mrc_ss}, we show that when the number of candidates $N$ is exponential in $\varepsilon$, the scaling factor $m_{\texttt{mrc}}$ is close to the scaling parameter associated with $\SubsetSelection$ (i.e., $m_{\texttt{ss}}$). 
Next, in Appendix \ref{appendix:mrc_ss_scaling_mse}, we provide the relationship between the mean squared error associated with $\MRC$ simulating $\SubsetSelection$ and the mean squared error associated with $\SubsetSelection$. In Appendix \ref{appendix:mrc_ss_utility}, we combine everything and show that, for frequency estimation, $\MRC$ can simulate $\SubsetSelection$ in a near-lossless manner while only using on the order of $\varepsilon$-bits of communication. Finally, in Appendix \ref{appendix:mrc_ss_emp}, we provide some empirical comparisons.

\subsection{Unbiased Minimal Random Coding simulating \texorpdfstring{$\SubsetSelection$}{Subset Selection}}\label{appendix:debias_mrc_ss}
Consider the $\SubsetSelection$ $\varepsilon$-LDP mechanism $\qSS$ with parameter $s$ as described in Section \ref{sec:preliminaries} and Appendix \ref{appendix:ss}. 
Let $\piMRC$ be the distribution and $\svbz_1, \svbz_2,...,\svbz_N$ be the candidates obtained from Algorithm \ref{alg:mrc} when the reference distribution is $\Unif(\cZ)$ where $\cZ$ is as defined in \eqref{eq:z_ss}.
Let $\theta$ denote the fraction of candidates inside $\msf{Cap}_{\svbx} = \cZ_{\svbx}$ where $\cZ_{\svbx}$ is the set of elements in $\cZ$ with $1$ in the same location as $\svbx$. 
It is easy to see that $\theta \sim \frac{1}{N}\msf{Binom}\lp N, \frac{s}{d} \rp$.
Let $\qMRC_i = \Probability(z_i = 1)$ where $\svbz \sim \qMRC(\cdot|\svbx)$ i.e., $\qMRC_i = \Probability\lbp(\svbz_K)_i = 1\rbp$ where $K \sim \piMRC(\cdot)$. 

The following lemma shows that the marginal distribution of $\qMRC_i$ can be written as a linear function of $p_i$ similar to $\qSS_i$ in \eqref{eq:ss_marginal}. This allows us to provide an unbiased estimator for $\MRC$ simulating $\SubsetSelection$.
\begin{restatable}{lemma}{mrcssbias}\label{thm:mrc_ss_bias}
Let $K \sim \piMRC(\cdot)$ and $\qMRC_i = \Probability\lbp(\svbz_K)_i = 1\rbp$ for $i \in [d]$. Then, 
\begin{align}
    \qMRC_i = p_i  m_\texttt{mrc} + b_\texttt{mrc}
\end{align}
where
\begin{align}\label{eq:q_check}
m_\texttt{mrc} & \coloneqq \E\lb \frac{\theta e^\varepsilon}{e^\varepsilon \theta+(1-\theta)} \rb - \frac{1}{d-1}\E\lb s-\frac{e^\varepsilon \theta}{e^\varepsilon \theta+(1-\theta)} \rb, \nonumber\\
b_\texttt{mrc} & \coloneqq \frac{1}{d-1}\E\lb s-\frac{e^\varepsilon \theta}{e^\varepsilon \theta + (1-\theta)} \rb.
\end{align}
Further, $\hat{\svbx}_\texttt{mrc} \coloneqq (\svbz_K - b_\texttt{mrc}\cdot\mb{1}_d)/m_\texttt{mrc}$ is an unbiased estimator of $\svbx$, i.e., $\E[\hat{\svbx}_\texttt{mrc}] = \svbx$.
\end{restatable}
\begin{proof}
We have
\begin{align}
    \Pr\lbp (\svbz_K)_i  = 1 \rbp 
    & = \sum_j p_j \Pr\lbp (\svbz_K)_i = 1 \mv \svbx = j\rbp  \\
    & \stackrel{(a)}{=} p_i\Pr\lbp (\svbz_K)_i = 1 \mv \svbx = i\rbp + (1-p_i) \Pr\lbp (\svbz_K)_i = 1 \mv \svbx = j\rbp. \label{eq:ss_step_1}
\end{align}
where $(a)$ follows by symmetry. Next, we compute $\Pr\lbp (\svbz_K)_i = 1 \mv \svbx = i\rbp$ and $\Pr\lbp (\svbz_K)_i = 1 \mv \svbx = j \rbp$ separately.

To compute $\Pr\lbp (\svbz_K)_i = 1 \mv \svbx = i\rbp$, recall that $\theta$ denotes the fraction of candidates that belong inside the $\msf{Cap}_{\svbx}$ i.e., have $1$ in the same location as $\svbx$. From Appendix \ref{appendix:ss_cap}, recall that $c_1(\varepsilon,d) \coloneqq \dfrac{e^\varepsilon}{\binom{d-1}{s-1}e^\varepsilon + \binom{d-1}{s}}$, $c_2(\varepsilon,d) \coloneqq \dfrac{1}{\binom{d-1}{s-1}e^\varepsilon + \binom{d-1}{s}}$. Further, since $\svbz_k$ are generated uniformly at random,
$$ \theta \sim \frac{1}{N} \msf{Binom}\lp N, \frac{\binom{d-1}{s-1}}{\binom{d}{s}} \rp = \frac{1}{N} \msf{Binom}\lp N, \frac{s}{d} \rp,$$ so we have
\begin{align}
    \Pr\lbp (\svbz_K)_i = 1 | \svbx = i \rbp = \Pr\lbp \svbz_K \in \msf{Cap}_{\svbx} | \svbx = i \rbp  & \stackrel{(a)}{=} \E\lb \Pr\lbp \svbz_K  \in \msf{Cap}_{\svbx} |\svbx=i, \theta \rbp\rb\\
    & = \E\lb \frac{c_1(\varepsilon,d) \theta}{c_1(\varepsilon,d) \theta + (1-\theta)c_2(\varepsilon,d)} \rb \\
    & \stackrel{(b)}{=} \E\lb \frac{e^\varepsilon \theta}{e^\varepsilon \theta + (1-\theta)} \rb, \label{eq:ss_step_2}
\end{align}
where $(a)$ follows by the law of total probability and $(b)$ is due to $c_1(\varepsilon,d)/c_2(\varepsilon,d) = e^\varepsilon$.

To compute $\Pr\lbp (\svbz_K)_i = 1 \mv \svbx = j \rbp$, we decompose it into
\begin{align}
    \Pr\lbp (\svbz_K)_i = 1 \mv \svbx = j \rbp = \Pr\lbp (\svbz_K)_i = 1, (\svbz_K)_j = 1 \mv \svbx = j\rbp + \Pr\lbp (\svbz_K)_i = 1, (\svbz_K)_j = 0 \mv \svbx = j\rbp, \label{eq:Step0}
\end{align}
for any $j \neq i$ and calculate each of the terms separately.

As before, let $\theta$ denotes the fraction of candidates that belong inside the $\msf{Cap}_{\svbx}$ i.e., have $1$ in the same location as $\svbx$. Further, let $\bar{\theta}$ denotes the fraction of candidates that belong inside the $\msf{Cap}_{\svbx}$ i.e., have $1$ in the same location as $\svbx$ as well as have $1$ in the $j^{th}$ location. Since $\svbz_k$ are generated uniformly at random,
$$ \bar{\theta} \sim \frac{1}{N} \msf{Binom}\lp N\theta, \frac{\binom{d-2}{s-2}}{\binom{d-1}{s-1}} \rp = \frac{1}{N} \msf{Binom}\lp N\theta, \frac{s-1}{d-1} \rp,$$ so we have
\begin{align}
    \Pr\lbp (\svbz_K)_i = 1, (\svbz_K)_j = 1 | \svbx = j \rbp 
    & \stackrel{(a)}{=} \E_{\theta}\lb \E_{\bar{\theta}}\lb \Pr\lbp (\svbz_K)_i = 1, (\svbz_K)_j = 1 \mv \svbx = j, \bar{\theta}, \theta \rbp \rb\rb \\
    & = \E_{\theta}\lb \E_{\bar{\theta}}\lb \frac{ c_1(\varepsilon,d)\bar{\theta}}{ c_1(\varepsilon,d) \theta  + (1 - \theta) c_2(\varepsilon,d)} \rb\rb \\ 
    & \stackrel{(b)}{=} \frac{s-1}{d-1}\E_{\theta}\lb \frac{ c_1(\varepsilon,d) \theta}{ c_1(\varepsilon,d)\theta +(1-\theta) c_2(\varepsilon,d)} \rb \\
    & \stackrel{(c)}{=} \frac{s-1}{d-1}\E\lb \frac{e^\varepsilon \theta}{e^\varepsilon \theta + (1-\theta)} \rb \label{eq:Step1}
\end{align}
where $(a)$ follows by the law of total probability, $(b)$ follows because $\E[\bar{\theta}] = \frac{s-1}{d-1} \times \theta$, and $(c)$ is due to $c_1(\varepsilon,d)/c_2(\varepsilon,d) = e^\varepsilon$.

Similarly, to compute the term $\Pr\lbp (\svbz_K)_i = 1, (\svbz_K)_j = 0 | \svbx = j \rbp$, let $\bar{\theta}$ denote the fraction of candidates that belong inside the $\msf{Cap}_{\svbx}$ i.e., have $1$ in the same location as $\svbx$ as well as have $0$ in the $j^{th}$ location. Since $\svbz_k$ are generated uniformly at random,
$$ \bar{\theta} \sim \frac{1}{N} \msf{Binom}\lp N(1-\theta), \frac{\binom{d-2}{s-1}}{\binom{d-1}{s}} \rp = \frac{1}{N} \msf{Binom}\lp N(1-\theta), \frac{s}{d-1} \rp,$$ so we have
\begin{align}
    \Pr\lbp (\svbz_K)_i = 1, (\svbz_K)_j = 0 | \svbx = j \rbp 
    & \stackrel{(a)}{=} \E_{\theta}\lb \E_{\bar{\theta}}\lb \Pr\lbp (\svbz_K)_i = 1, (\svbz_K)_j = 0 \mv \svbx = j, \bar{\theta}, \theta \rbp \rb\rb \\
    & = \E_{\theta}\lb \E_{\bar{\theta}}\lb \frac{ c_2(\varepsilon,d)\bar{\theta}}{ c_1(\varepsilon,d) \theta  + (1 - \theta) c_2(\varepsilon,d)} \rb\rb \\ 
    & \stackrel{(b)}{=} \frac{s}{d-1}\E_{\theta}\lb \frac{ c_2(\varepsilon,d) (1-\theta)}{ c_1(\varepsilon,d)\theta +(1-\theta) c_2(\varepsilon,d)} \rb \\ 
    & \stackrel{(c)}{=} \frac{s}{d-1}\E\lb \frac{(1-\theta)}{e^\varepsilon \theta + (1-\theta)} \rb, \label{eq:Step2}
\end{align}
where $(a)$ follows by the law of total probability, $(b)$ follows because $\E[\bar{\theta}] = \frac{s}{d-1} \times (1-\theta)$, and $(c)$ is due to $c_1(\varepsilon,d)/c_2(\varepsilon,d) = e^\varepsilon$.
Using \eqref{eq:Step1} and \eqref{eq:Step2} in \eqref{eq:Step0}, we have
\begin{align}
    \Pr\lbp (\svbz_K)_i = 1 | \svbx = j \rbp 
    & = \Pr\lbp (\svbz_K)_i = 1, (\svbz_K)_j = 1 | \svbx = j \rbp+\Pr\lbp (\svbz_K)_i = 1, (\svbz_K)_j = 0 | \svbx = j \rbp \\
    & = \frac{s-1}{d-1}\E\lb \frac{e^\varepsilon \theta}{e^\varepsilon \theta + (1-\theta)} \rb + \frac{s}{d-1}\E\lb \frac{(1-\theta)}{e^\varepsilon \theta + (1-\theta)} \rb \\
    & = \frac{1}{d-1}\lp s - \E\lb \frac{e^\varepsilon \theta}{e^\varepsilon \theta +(1-\theta)}\rb \rp \label{eq:ss_step_3}
\end{align}
Combining everything, we have
\begin{align}
    \qMRC_i & = \Probability\lbp(\svbz_K)_i = 1\rbp \\
    & \stackrel{(a)}{=} p_i \times \lb \Pr\lbp (\svbz_K)_i = 1 \mv \svbx = i\rbp - \Pr\lbp (\svbz_K)_i = 1 \mv \svbx = j\rbp \rb + \Pr\lbp (\svbz_K)_i = 1 \mv \svbx = j\rbp. \\
    & \stackrel{(b)}{=} p_i \times \lb \E\lb \frac{e^\varepsilon \theta}{e^\varepsilon \theta + (1-\theta)} \rb - \frac{1}{d-1}\lp s - \E\lb \frac{e^\varepsilon \theta}{e^\varepsilon \theta +(1-\theta)}\rb \rp \rb \\
    &  + \frac{1}{d-1}\lp s - \E\lb \frac{e^\varepsilon \theta}{e^\varepsilon \theta +(1-\theta)}\rb \rp \\
    & \stackrel{(c)}{=} p_i  m_\texttt{mrc} + b_\texttt{mrc}
\end{align}
where $(a)$ follows from \eqref{eq:ss_step_1}, $(b)$ follows from \eqref{eq:ss_step_2} and \eqref{eq:ss_step_3}, and $(c)$ follows from the definitions of $m_\texttt{mrc}$ and $b_\texttt{mrc}$. 

Note that the above conclusion holds for all prior distribution $\svbp = (p_1, ..., p_d)$ such that $\svbx \sim \svbp$. Thus by setting $\svbp = \svbx$ (here $\svbx$ is viewed as a one-hot vector), i.e., letting $\svbp$ be the point mass distribution at $\svbx$, we have
\begin{align}
    \E[\hat{\svbx}_\texttt{mrc}] = (\E[ \svbz_K ] - b_\texttt{mrc}\cdot\mb{1}_d)/m_\texttt{mrc}
    & = ( \qMRC - b_\texttt{mrc}\cdot\mb{1}_d)/m_\texttt{mrc} \\
    & = \lp\lp m_\texttt{mrc}\cdot \svbp + b_\texttt{mrc}\cdot\mb{1}_d\rp -b_\texttt{mrc}\cdot\mb{1}_d\rp/m_\texttt{mrc} = \svbp \stackrel{(a)}{=} \svbx,
\end{align}
where $(a)$ is due to our construction of $\svbp$.
\end{proof}

 \subsection{Utility of Minimal Random Coding simulating \texorpdfstring{$\SubsetSelection$}{Subset Selection}}\label{appendix:mrc_ss_ut}
 
 \subsubsection{The scaling factors of \texorpdfstring{$\SubsetSelection$}{Subset Selection} and \texorpdfstring{$\MRC$}{MRC} are close when \texorpdfstring{$N$}{N} is of the right order}\label{appendix:scaling_mrc_ss}
 In the following Lemma, we show that when the number of candidates $N$ is exponential in $\varepsilon$, then the scaling parameters associated with $\SubsetSelection$ and the $\MRC$ scheme simulating $\SubsetSelection$ are close.
 
\begin{lemma}\label{lemma:mrc_ss_approximation_error}
Let $N$ denote the number of candidates used in the \emph{$\MRC$} scheme. Let $K \sim \piMRC$ where $\piMRC$ is the distribution over the indices $[N]$ associated the \emph{$\MRC$} scheme simulating \emph{$\SubsetSelection$}. Consider any $\lambda > 0$.
Then, the scaling factors $m_{\texttt{ss}}$ and $b_{\texttt{ss}}$ associated with \emph{$\SubsetSelection$} and the scaling factors $m_\texttt{mrc}$ and $b_{\texttt{mrc}}$ associated with the \emph{$\MRC$} scheme simulating \emph{$\SubsetSelection$} are such that
\begin{align}
    m_{\texttt{ss}} - m_\texttt{mrc} \leq \lambda\cdot m_\texttt{mrc}
\end{align}
and $b_{\texttt{ss}} \leq b_\texttt{mrc}$
as long as
\begin{align}
    N\geq \frac{2(e^{\varepsilon}+3)^2(1+\lambda)^2}{0.24^2\lambda^2}\ln\lp \frac{8(1+\lambda)}{0.24\lambda}\rp.
\end{align}
\end{lemma}
\begin{proof}
First, we will obtain convenient expressions for $m_{\texttt{ss}}$ and $b_{\texttt{ss}}$ defined in \eqref{eq:scaling_factors_ss}. We can write 
\begin{align}
& m_{\texttt{ss}} \coloneqq  \lp \frac{\E[\theta]e^\varepsilon}{e^\varepsilon \E[\theta]+(1-\E[\theta])}\rp - \frac{1}{d-1}\lp s-\frac{e^\varepsilon \E[\theta]}{e^\varepsilon \E[\theta]+(1-\E[\theta])} \rp \label{eq:check_mb_def_1}\\
& b_{\texttt{ss}} \coloneqq \frac{1}{d-1}\lp s-\frac{e^\varepsilon\E[\theta]}{e^\varepsilon \E[\theta] + (1-\E[\theta])} \rp. \label{eq:check_mb_def_2}   
\end{align}
To verify these, we simply plug  $\E[\theta] = \frac{s}{d}$ into \eqref{eq:check_mb_def_1} resulting in:
\begin{align*}
    m_{\texttt{ss}} & = \frac{d}{d-1}\frac{se^\varepsilon}{se^\varepsilon + (d-s)} - \frac{s}{d-1}  = \frac{d s e^\varepsilon - s^2 e^\varepsilon - s(d-t)}{(d-1)\lp s e^\varepsilon+d-s\rp}  = \frac{s(d-s)(e^\varepsilon-1)}{(d-1)\lp se^\varepsilon+d-s\rp}.
\end{align*}
and into \eqref{eq:check_mb_def_2} resulting in:
\begin{align*}
    b_{\texttt{ss}} = \frac{1}{d-1}\lp s - \frac{se^\varepsilon}{se^\varepsilon + d - s} \rp 
    & =  \frac{1}{d-1}\lp \frac{s^2e^\varepsilon + s(d -s)-se^\varepsilon}{se^\varepsilon + d - s} \rp \\
    & = \frac{1}{d-1}\lp \frac{s(s-1)e^\varepsilon + s(d -s)}{se^\varepsilon + d - s} \rp.
\end{align*}

Recall the definitions of
$b_{\texttt{ss}}$ and $m_{\texttt{ss}}$ from Lemma \ref{thm:mrc_ss_bias}. Applying Jensen's inequality on the concave function $x\mapsto \frac{x}{x+c}$ for some $c >0$ yields $m_{\texttt{mrc}}\leq m_{\texttt{ss}}$ and $b_{\texttt{mrc}} \geq b_{\texttt{ss}}$.

Now, we will bound $\lba m_{\texttt{mrc}} - m_{\texttt{ss}} \rba$:
\begin{align}
    \lba m_{\texttt{mrc}} - m_{\texttt{ss}} \rba 
    & = \lp\frac{d}{d-1}\rp\lp \frac{\E[\theta]e^\varepsilon}{e^\varepsilon \E[\theta]+(1-\E[\theta])} - \E\lb \frac{\theta e^\varepsilon}{e^\varepsilon \theta+(1-\theta)} \rb\rp \\
    &\stackrel{(a)}{\leq} 2 \lp \frac{\E[\theta]e^\varepsilon}{e^\varepsilon \E[\theta]+(1-\E[\theta])} - \E\lb \frac{\theta e^\varepsilon}{e^\varepsilon \theta+(1-\theta)} \rb\rp \\
    & = 2 \lp \E\lb \frac{\lp \E[\theta]-\theta\rp e^\varepsilon}{\lp(e^\varepsilon-1) \E[\theta]+1\rp\lp(e^\varepsilon-1) \theta+1\rp} \rb\rp\label{eq:ss_m_hat_minus_m_eq_1},
\end{align}
where $(a)$ holds since $d\geq 2$.
Next, we condition on the event $\mcal{E} \coloneqq \lbp \lba \E[\theta]-\theta\rba \leq \sqrt{\frac{\ln(2/\beta)}{2N}} \rbp$, which has probability $\Pr_\theta\lbp \mcal{E} \rbp \geq 1-\beta$ by Hoeffding's inequality. We continue to upper bound \eqref{eq:ss_m_hat_minus_m_eq_1}:
\begin{align}
\lba m_{\texttt{mrc}}    -    m_{\texttt{ss}} \rba 
&= 2 \bigg( \Pr\lbp \mcal{E} \rbp\E\lb \frac{\lp  \E[\theta]-\theta\rp e^\varepsilon}{\lp(e^\varepsilon    -    1) \E[\theta]    +    1\rp\lp(e^\varepsilon    -    1) \theta    +    1\rp} \mv \mcal{E} \rb    \\
& \qquad +    \Pr\lbp \mcal{E}^c \rbp\E\lb \frac{\lp  \E[\theta]-\theta\rp e^\varepsilon}{\lp(e^\varepsilon    -    1) \E[\theta]   +    1\rp\lp(e^\varepsilon    -    1) \theta    +    1\rp} \mv \mcal{E}^c \rb \bigg)\\
& \stackrel{(a)}{\leq} 2 \lp \E\lb \frac{\lp  \E[\theta]-\theta\rp e^\varepsilon}{\lp(e^\varepsilon-1) \E[\theta]+1\rp\lp(e^\varepsilon-1) \theta+1\rp} \mv \mcal{E} \rb +\beta \rp\\
& \stackrel{(b)}{\leq} 2 \lp \E\lb \frac{\lp  \E[\theta]-\theta\rp e^\varepsilon}{\lp(e^\varepsilon-1) \E[\theta]+1\rp\lp(e^\varepsilon-1) E[\theta]/2+1\rp} \mv \mcal{E} \rb +\beta \rp\\
    & \leq 4\E\lb \frac{\lp \E[\theta] - \theta \rp e^\varepsilon}{\lp(e^\varepsilon-1) \E[\theta]+1\rp^2} \mv \mcal{E} \rb + 2\beta \stackrel{(c)}{\leq} 4\sqrt{\frac{\ln(2/\beta)}{2N}} \frac{e^\varepsilon(1+e^\varepsilon)^2}{4e^{2\varepsilon}}+ 2\beta \\
    & = \sqrt{\frac{\ln(2/\beta)}{2N}}\lp e^\varepsilon+2+\frac{1}{e^\varepsilon} \rp + 2\beta \leq \sqrt{\frac{\ln(2/\beta)}{2N}}(e^\varepsilon+3)+2\beta, \label{eq:mrc_ss_nrel1}
\end{align}
 where $(a)$ holds since 
 $$ \frac{\lp  \E[\theta]-\theta\rp e^\varepsilon}{\lp(e^\varepsilon-1) \E[\theta]+1\rp\lp(e^\varepsilon-1) \theta+1\rp} = \frac{\E[\theta]e^\varepsilon}{e^\varepsilon \E[\theta]+(1-\E[\theta])} -  \frac{\theta e^\varepsilon}{e^\varepsilon \theta+(1-\theta)} \leq 1,$$
 $(b)$ holds if we pick $N$ large enough so that $\lba \theta - \E[\theta]\rba \leq \frac{\E[\theta]}{2} $ for which a sufficient condition is $ \sqrt{\frac{\ln(2/\beta)}{2N}} \leq \frac{\E[\theta]}{2}  $ i.e., $N\geq \frac{2\ln(2/\beta)}{\E[\theta]^2} = 2(d/s)^2\ln(2/\beta)$,
and $(c)$ holds since $\E[\theta] = s/d \geq 1/(1+e^\varepsilon)$. Notice that the constraint $N\geq  2(d/s)^2 \ln(2/\beta)$ in inequality $(b)$ can be further satisfied as long as $ N \geq 2\ln(2/\beta)(1+e^\varepsilon)^2 $
since $s/d \geq 1/(1+e^\varepsilon)$.

Next, we lower bound $m_{\texttt{ss}}$ in \eqref{eq:check_mb_def_1}:
\begin{align}
    m_{\texttt{ss}} & = \lp\frac{d}{d-1}\rp\lp \frac{\E[\theta]e^\varepsilon}{e^\varepsilon \E[\theta]+(1-\E[\theta])} - \frac{s}{d} \rp  \geq \frac{\E[\theta]e^\varepsilon}{e^\varepsilon \E[\theta]+(1-\E[\theta])} - \frac{s}{d}\\
    &\stackrel{(a)}{=}\frac{s}{d}\lb \frac{(e^\varepsilon-1)(d-s)}{(e^\varepsilon-1)\cdot s +d} \rb =\frac{s}{d}\lb \frac{(e^\varepsilon-1)(1-s/d)}{(e^\varepsilon-1)\cdot s/d +1} \rb \\
    & \stackrel{(b)}{\geq} \frac{1}{1+e^\varepsilon}\lb \frac{(e^\varepsilon-1)\lp\frac{e^\varepsilon}{1+e^\varepsilon}-\frac{1}{d}\rp}{(e^\varepsilon-1)\lp\frac{1}{1+e^\varepsilon}+\frac{1}{d}\rp+1} \rb \\
    & \stackrel{(c)}{\geq} \frac{1}{1+e^\varepsilon}\lb \frac{(e^\varepsilon-1)\frac{e^\varepsilon-1}{1+e^\varepsilon}}{(e^\varepsilon-1)\lp\frac{2}{1+e^\varepsilon}\rp+1} \rb = \frac{\lp e^\varepsilon-1\rp^2}{(3e^\varepsilon-1)(e^\varepsilon+1)} \\
    &\stackrel{(d)}{\geq} \frac{(e-1)^2}{(3e-1)(e+1)} \geq 0.24, \label{eq:mrc_ss_nrel2}
\end{align}
where $(a)$ holds by plugging in $\E[\theta] = s/d$, $(b)$ holds since $s = \lceil d/(1+e^\varepsilon)\rceil$ (so $\frac{1}{1+e^\varepsilon} \leq \frac{s}{d} \leq \frac{1}{1+e^\varepsilon}+\frac{1}{d}$), $(c)$ holds since we only focus on the regime where $\varepsilon \leq d-1$ (so $\frac{1}{d} \leq \frac{1}{1+\varepsilon}$), and $(d)$ holds by observing that $f(x) \coloneqq \frac{(x-1)^2}{(3x-1)(x+1)}$ is an increasing function for $x \geq 1$ and we have $\varepsilon \geq 1$. Putting things together, we obtain
\begin{align}
    \frac{m_{\texttt{ss}}-m_{\texttt{mrc}}}{m_{\texttt{mrc}}} =  \frac{m_{\texttt{ss}}-m_{\texttt{mrc}}}{m_{\texttt{ss}}-(m_{\texttt{ss}}-m_{\texttt{mrc}})} \stackrel{(a)}{\leq} \frac{\sqrt{\frac{\ln(2/\beta)}{2N}}(e^\varepsilon+3)+2\beta}{0.24-\lp \sqrt{\frac{\ln(2/\beta)}{2N}}(e^\varepsilon+3)+2\beta \rp} \stackrel{(b)}{\leq} \lambda,
\end{align}
where $(a)$ follows from \eqref{eq:mrc_ss_nrel1} and \eqref{eq:mrc_ss_nrel2} and $(b)$ follows as long as
\begin{align}
    \sqrt{\frac{\ln(2/\beta)}{2N}}(e^\varepsilon+3)+2\beta \leq \frac{0.24 \lambda}{1+\lambda}. \label{eq:mrc_ss_nrel3}
\end{align}
To ensure \eqref{eq:mrc_ss_nrel3}, we let
\begin{align*}
    \beta \leq \frac{0.24\lambda}{4(1+\lambda)} \qquad \text{and} \qquad N \geq \frac{1}{2}\lp\frac{(e^\varepsilon+3)}{\frac{0.24\lambda}{(1+\lambda)}-2\beta}\rp^2\ln(2/\beta) = \frac{2(e^{\varepsilon}+3)^2(1+\lambda)^2}{0.24^2\lambda^2}\ln\lp \frac{8(1+\lambda)}{0.24\lambda}\rp.
\end{align*}
It is easy to verify that this choice of $N$ satisfies $ N \geq 2\ln(2/\beta)(1+e^\varepsilon)^2 $.
\end{proof}

\subsubsection{Relationship between mean squared errors associated with \texorpdfstring{$\SubsetSelection$}{Subset Selection} and \texorpdfstring{$\MRC$}{MRC} simulating \texorpdfstring{$\SubsetSelection$}{Subset Selection}}\label{appendix:mrc_ss_scaling_mse}
In the following Proposition,
we show that if $m_{\texttt{mrc}}$ is close to $m_{\texttt{ss}}$ and $b_{\texttt{mrc}} \geq b_{\texttt{ss}}$, then the mean squared error associated with $\MRC$ simulating $\SubsetSelection$ (i.e., $\Expectation_{\qMRC} \big[ \lV  \hat{\svbx}^\texttt{mrc} - \svbx \rV^2_2  \big]$) is close to the mean squared error associated with $\SubsetSelection$ (i.e., $\Expectation_{\qSS} \big[ \lV  \hat{\svbx}^\texttt{ss} - \svbx \rV^2_2  \big]$).
\begin{proposition}\label{proposition:mrc_mse_wrt_ss}
Let $\qSS(\svbz | \svbx)$ be the  $\varepsilon$-LDP \emph{$\SubsetSelection$} mechanism with estimator $\hat{\svbx}^{\texttt{ss}}$. Let $\qMRC(\svbz|\svbx)$ denote the \emph{$\MRC$} privatization mechanism simulating \emph{$\SubsetSelection$} with $N$ candidates and estimator $\hat{\svbx}^{\texttt{mrc}}$.
Let $m_{\texttt{ss}}$ and $b_{\texttt{ss}}$ denote the scaling factors  associated with \emph{$\SubsetSelection$} and $m_\texttt{mrc}$ and $b_\texttt{mrc}$ denote the scaling factors associated with the \emph{$\MRC$} scheme simulating \emph{$\SubsetSelection$}. Consider any $\lambda > 0$. If $m_{\texttt{pu}} - m_\texttt{mrc} \leq \lambda \cdot m_\texttt{mrc}$ and $b_{\texttt{mrc}} \geq b_{\texttt{ss}}$, then 
  \begin{align}
    \Expectation_{\qMRC} \big[ \lV  \hat{\svbx}^\texttt{mrc} - \svbx \rV^2_2  \big]  \leq  \lp 1+4\lambda +5\lambda^2 + 2\lambda^3 \rp  \Expectation_{\qSS}\big[\|\hat{\svbx}^{\texttt{ss}} - \svbx\|^2\big]
\end{align}
\end{proposition}
\begin{proof}
We have
$$  \Expectation_{\qMRC} \big[ \lV  \hat{\svbx}^\texttt{mrc} - \svbx \rV^2_2  \big]  \stackrel{(a)}{=} \sum_{i=1}^d \Var\lp \hat{\svbx}^\texttt{mrc}_i \rp \stackrel{(b)}{=} \lp\frac{1}{m_\texttt{mrc}}\rp^2 \sum_i \Var \lp \lp \svbz_K \rp_i\rp \stackrel{(c)}{=} \lp\frac{1}{m_\texttt{mrc}}\rp^2 \sum_i \qMRC_i (1-\qMRC_i). $$
where $(a)$ follows because $\svbx$ is a constant, $(b)$ follows because $\hat{\svbx}_\texttt{mrc} =(\svbz_K - b_\texttt{mrc})/m_\texttt{mrc}$, and $(c)$ follows because $\lp \svbz_K \rp_i\sim\msf{Ber}(\qMRC_i)$. Similarly, we have
We have
$$  \Expectation_{\qSS} \big[ \lV  \hat{\svbx}^\texttt{ss} - \svbx \rV^2_2  \big]  \stackrel{(a)}{=} \sum_{i=1}^d \Var\lp \hat{\svbx}^\texttt{ss}_i \rp \stackrel{(b)}{=} \lp\frac{1}{m_\texttt{ss}}\rp^2 \sum_i \Var \lp z_i\rp \stackrel{(c)}{=} \lp\frac{1}{m_\texttt{ss}}\rp^2 \sum_i \qSS_i (1-\qSS_i). $$
where $(a)$ follows because $\svbx$ is a constant, $(b)$ follows because $\hat{\svbx}^\texttt{ss} =(\svbz -  b_\texttt{ss})/m_\texttt{ss}$, and $(c)$ follows because $z_i\sim\msf{Ber}(\qSS_i)$.

Now, let us look at the difference i.e.,
\begin{align*}
    & \Expectation_{\qMRC} \big[ \lV  \hat{\svbx}^\texttt{mrc} - \svbx \rV^2_2  \big] - \Expectation_{\qSS} \big[ \lV  \hat{\svbx}^\texttt{ss} - \svbx \rV^2_2  \big] \\
    & = \lp\frac{1}{m_\texttt{mrc}}\rp^2 \sum_i \qMRC_i (1-\qMRC_i) - \lp\frac{1}{m_\texttt{ss}}\rp^2 \sum_i \qSS_i (1-\qSS_i)\\
    & \leq \lp\frac{1}{m_\texttt{mrc}}\rp^2\sum_{i}\lp \qMRC_i (1-\qMRC_i) - \qSS_i (1-\qSS_i) \rp + \lb \frac{1}{m_\texttt{mrc}^2}-\frac{1}{m_\texttt{ss}^2}\rb\lp \sum_i \qSS_i (1-\qSS_i) \rp.
\end{align*}



Now, first, we will bound $\lp\frac{1}{m_\texttt{mrc}}\rp^2\sum_{i}\lp \qMRC_i (1-\qMRC_i) - \qSS_i (1-\qSS_i) \rp$. To that end, observe that $m_{\texttt{pu}} - m_\texttt{mrc} \leq \lambda \cdot m_\texttt{mrc}$ implies 
\begin{equation}\label{eq:m_hat_bdd_2}
    \frac{1}{m_\texttt{mrc}} \leq (1+\lambda)\frac{1}{m_{\texttt{ss}}}.
\end{equation} 
Further, we have
\begin{align}
\qMRC_i & \stackrel{(a)}{=} m_\texttt{mrc} p_i+b_\texttt{mrc} \stackrel{(b)}{=} \qSS_i + (m_\texttt{mrc}-m_\texttt{ss})p_i+(b_\texttt{mrc}-b_\texttt{ss}) \stackrel{(c)}{\geq} \qSS_i - \lambda \cdot m_\texttt{mrc} \cdot p_i +(b_\texttt{mrc}-b_\texttt{ss}) \\ & 
\stackrel{(d)}{\geq} \qSS_i - \lambda \cdot m_\texttt{mrc} \cdot p_i
\stackrel{(e)}{\geq} \qSS_i - \lambda \cdot m_\texttt{ss} \cdot p_i
\stackrel{(f)}{\geq} (1-\lambda)\qSS_i, \label{eq:q_rel}
\end{align}
where $(a)$ follows from Lemma \ref{thm:mrc_ss_bias}, $(b)$ follows from \eqref{eq:ss_marginal}, $(c)$ follows because $m_{\texttt{pu}} - m_\texttt{mrc} \leq \lambda \cdot m_\texttt{mrc}$, $(d)$ follows because $b_\texttt{mrc} \geq b_\texttt{ss}$, $(e)$ follows because $m_\texttt{ss} \geq m_\texttt{mrc}$ as seen in Lemma \ref{lemma:mrc_ss_approximation_error}, and $(f)$ follows because $b_\texttt{ss} \geq 0$. Next, we have
\begin{align}
    \frac{\qMRC_i(1-\qMRC_i) - \qSS_i(1-\qSS_i)}{\qSS_i(1-\qSS_i)} 
     = \frac{(\qSS_i - \qMRC_i)(\qSS_i+\qMRC_i -1)}{\qSS_i(1-\qSS_i)} 
     \stackrel{(a)}{\leq} \frac{\lambda \qSS_i(\qSS_i+\qMRC_i-1)}{\qSS_i(1-\qSS_i)}
     \stackrel{(b)}{\leq} \frac{\lambda }{1-\qSS_i} \label{eq:mrc_ss_inter1}
\end{align}
where $(a)$ follows from \eqref{eq:q_rel} and $(b)$ follows since $\qSS_i \leq 1$ and $\qMRC_i \leq 1$.

Let us now upper bound $\qSS_i$. We have 
\begin{align}
    \qSS_i  = m_\texttt{ss} \cdot p_i +
b_\texttt{ss} \stackrel{(a)}{\leq} m_\texttt{ss} + b_\texttt{ss} \stackrel{(b)}{=} \lp \frac{\E[\theta]e^\varepsilon}{e^\varepsilon \E[\theta]+(1-\E[\theta])}\rp  \stackrel{(c)}{\leq} \frac{1}{2} \label{eq:mrc_ss_inter2}
\end{align}
where $(a)$ follows because $p_i \leq 1$, $(b)$ follows from \eqref{eq:check_mb_def_1} and \eqref{eq:check_mb_def_2}, and $(c)$ follows because $\E[\theta] = \frac{s}{d} \geq \frac{1}{e^\varepsilon+ 1}$. Combining \eqref{eq:mrc_ss_inter1} and \eqref{eq:mrc_ss_inter2}, and then re-arranging results in
$$ \sum_i \qMRC_i(1-\qMRC_i) - \sum_i \qSS_i(1-\qSS_i)\leq 2\lambda \sum_i \qSS_i(1-\qSS_i). $$
Together with \eqref{eq:m_hat_bdd_2}, we obtain
$$ \lp\frac{1}{m_\texttt{mrc}}\rp^2\sum_{i}\lp \qMRC_i (1-\qMRC_i) - \qSS_i (1-\qSS_i) \rp \leq  \frac{2\lambda(1+\lambda)^2}{m_\texttt{ss}^2}\sum_i \qSS_i(1-\qSS_i).$$

To bound $\lb \frac{1}{m_\texttt{mrc}^2}-\frac{1}{m_\texttt{ss}^2}\rb\lp \sum_i \qSS_i (1-\qSS_i) \rp$, simply note that \eqref{eq:m_hat_bdd_2} implies $\frac{1}{m_\texttt{mrc}^2} \leq (1+\lambda)^2\frac{1}{m_{\texttt{ss}}^2}$ resulting in
$$\lb \frac{1}{m_\texttt{mrc}^2}-\frac{1}{m_\texttt{ss}^2}\rb\lp \sum_i \qSS_i (1-\qSS_i) \rp \leq \frac{2\lambda + \lambda^2}{m_\texttt{ss}^2} \lp \sum_i \qSS_i(1-\qSS_i) \rp.$$

Combining everything, we have
\begin{align}
    \Expectation_{\qMRC} \big[ \lV  \hat{\svbx}^\texttt{mrc}    -     \svbx \rV^2_2  \big]   & \leq    \lp 1+2\lambda(1+\lambda)^2 + 2\lambda + \lambda^2 \rp\frac{1}{m_\texttt{mrc}^2}\sum_i \qSS_i(1-\qSS_i)\\
    & =     \lp 1+4\lambda +5\lambda^2 + 2\lambda^3 \rp  \Expectation_{\qSS}\big[\|\hat{\svbx}^{\texttt{ss}}    -    \svbx\|^2\big]
\end{align}
\end{proof}

In the following Lemma, we show that with on the order of $\varepsilon$-bits of communication, the mean squared error associated with $\MRC$ simulating $\SubsetSelection$ (i.e., $\Expectation_{\qMRC} \big[ \lV  \hat{\svbx}^\texttt{mrc} - \svbx \rV^2_2  \big]$) is close to the mean squared error associated with $\SubsetSelection$ (i.e., $\Expectation_{\qSS} \big[ \lV  \hat{\svbx}^\texttt{ss} - \svbx \rV^2_2  \big]$).
\begin{restatable}{lemma}{mrcss}\label{theorem:mrc_ss}
Let $\qSS(\svbz | \svbx)$ be the  $\varepsilon$-LDP \emph{$\SubsetSelection$} mechanism with estimator $\hat{\svbx}^{\texttt{ss}}$. Let $\qMRC(\svbz|\svbx)$ denote the \emph{$\MRC$} privatization mechanism simulating \emph{$\SubsetSelection$} with $N$ candidates and estimator $\hat{\svbx}^{\texttt{mrc}}$. 
Consider any $\lambda > 0$. Then,
\begin{align}
   \Expectation_{\qMRC} \big[ \lV  \hat{\svbx}^\texttt{mrc} - \svbx \rV^2_2  \big]  \leq  \lp 1+4\lambda +5\lambda^2 + 2\lambda^3 \rp  \Expectation_{\qSS}\big[\|\hat{\svbx}^{\texttt{ss}} - \svbx\|^2\big]
\end{align}
as long as 
\begin{align}
    N \geq  \frac{2(e^{\varepsilon}+3)^2(1+\lambda)^2}{0.24^2\lambda^2}\ln\lp \frac{8(1+\lambda)}{0.24\lambda}\rp.
\end{align}
\end{restatable}
\begin{proof}
The proof follows from Proposition \ref{proposition:mrc_mse_wrt_ss} and Lemma \ref{lemma:mrc_ss_approximation_error}.
\end{proof}

\subsubsection{Simulating \texorpdfstring{$\SubsetSelection$}{Subset Selection} using Minimal Random Coding}\label{appendix:mrc_ss_utility}
The following Theorem shows that, for frequency estimation, $\MRC$ can simulate $\SubsetSelection$ in a near-lossless manner (when $\lambda$ is small) while only using on the order of $\varepsilon$ bits of communication.

\begin{restatable}{theorem}{mrcssut}
Let $r_{\msf{FE}} \lp \hat{\Pi}^\texttt{ss}, \qSS \rp$ and $r_{\msf{FE}} \lp \hat{\Pi}^\texttt{mrc}, \qMRC \rp$ be the empirical frequency estimation error for \emph{$\SubsetSelection$} and \emph{$\MRC$} simulating \emph{$\SubsetSelection$} with $N$ candidates respectively. Consider any $\lambda > 0$. Then
 \begin{equation}
     r_{\msf{FE}} \lp \hat{\Pi}^{ \texttt{mrc}}, \qMRC \rp \leq  
     \lp 1+4\lambda+5\lambda^2+2\lambda^3 \rp r_{\msf{FE}} \lp \hat{\Pi}^{ \texttt{ss}}, \qSS \rp,
 \end{equation}
as long as 
\begin{align}
   N \geq  \frac{2(e^{\varepsilon}+3)^2(1+\lambda)^2}{0.24^2\lambda^2}\ln\lp \frac{8(1+\lambda)}{0.24\lambda}\rp.
\end{align}
\end{restatable}
\begin{proof}
The proof follows directly from Lemma~\ref{theorem:mrc_ss} since for all $i \in [n]$, $\hat{\svbx}^{\texttt{mrc}}_i$ are independent of each other as well as unbiased.
\end{proof}

\subsection{Empirical Comparisons}
\label{appendix:mrc_ss_emp}
In this section, we compare $\MRC$ simulating $\SubsetSelection$ (using its approximate DP guarantee) against $\SubsetSelection$ and RHR for frequency estimation with $d = 500$ and $n = 5000$. We use the same data generation scheme described in Section \ref{subsec:mmrc_ss_empirical} and set $\delta = 10^{-6}$. As before, RHR uses $\#$-bits $= \varepsilon$ because it leads to a poor performance if $\#$-bits $ > \varepsilon$. We show the privacy-accuracy tradeoffs for these three methods in Figure \ref{fig:a_freq}. We see that $\MRC$ simulating $\SubsetSelection$ can attain the accuracy of the uncompressed $\SubsetSelection$ for the range of  $\varepsilon$'s typically considered by LDP mechanisms while only using $(3\varepsilon/ \ln 2) + 6$ bits. In comparison with the results from Section \ref{subsec:mmrc_ss_empirical}, the results in this section come with an approximate guarantee ($\delta = 10^{-6}$) and with a higher number of bits of communication. In other words, along with the obvious gains of pure privacy instead of approximate privacy, $\MMRC$ results in a lower communication cost (and therefore a lower computation cost) compared to $\MRC$.

\begin{figure}[h]
\centering
\includegraphics[width=0.45\linewidth]{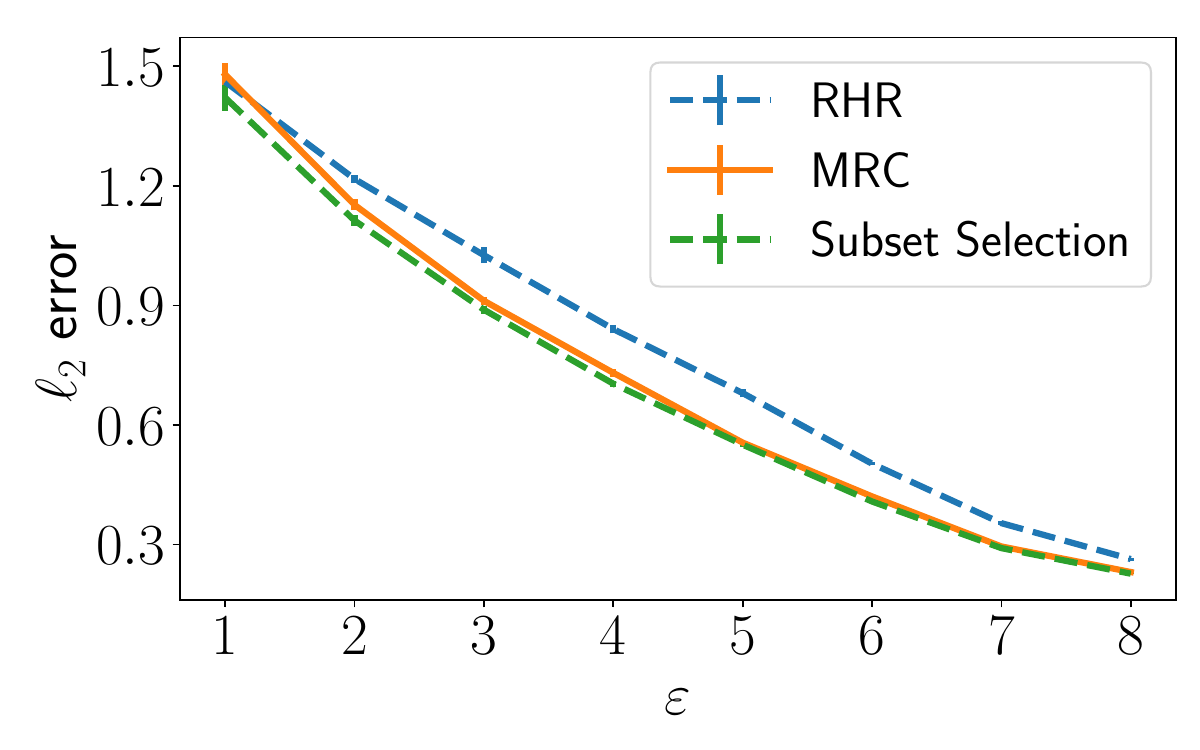}
\caption{Comparing $\SubsetSelection$, $\MRC$ simulating $\SubsetSelection$ and SQKR for frequency estimation in terms of $\ell_2$ error vs $\varepsilon$ with $d = 500$, $n = 5000$, and $\#$bits $= (3\varepsilon/ \ln 2) + 6$.}
\label{fig:a_freq}
\end{figure}
\section{Modified Minimal Random Coding Simulating \texorpdfstring{$\SubsetSelection$}{Subset Selection}}
\label{sec:mmrc_ss_app}

In this section, we prove Lemma \ref{lemma:mmrc_ss_bias} (in Appendix \ref{appendix:mmrc_ss_debias}) and Theorem \ref{thm:fe_mmrc_ss} (in Appendix \ref{appendix:mmrc_ss_utility}). To prove Theorem \ref{thm:fe_mmrc_ss}, first, in Appendix \ref{appendix:scaling_mmrc_ss}, we show that when the number of candidates $N$ is exponential in $\varepsilon$, the scaling factor $m_{\texttt{mmrc}}$ is close to the scaling parameter associated with $\SubsetSelection$ (i.e., $m_{\texttt{ss}}$). Next, in Appendix \ref{appendix:mmrc_ss_scaling_mse}, we provide the relationship between the mean squared error associated with $\MMRC$ simulating $\SubsetSelection$ and the mean squared error associated with $\SubsetSelection$. Finally, in Appendix \ref{appendix:mmrc_ss_emp}, we provide some empirical comparisons in addition to the ones in Section \ref{subsec:mmrc_ss_empirical} between $\MMRC$ simulating $\SubsetSelection$ and $\SubsetSelection$.

\subsection{Unbiased Modified Minimal Random Coding simulating \texorpdfstring{$\SubsetSelection$}{Subset Selection}}\label{appendix:mmrc_ss_debias}
Consider the $\SubsetSelection$ $\varepsilon$-LDP mechanism $\qSS$ described in Section~\ref{sec:preliminaries} with $s \coloneqq \lceil \frac{d}{1+e^\varepsilon}\rceil$. $\SubsetSelection$ is cap-based mechanism as discussed in Section \ref{sec:main_results} and Appendix \ref{appendix:ss} with $\msf{Cap}_{\svbx} = \cZ_{\svbx}$ and $\Probability_{\svbz \sim \Unif(\cZ)}\lp \svbz \in \msf{Cap}_{\svbx} \rp = s/d$.
Let $\piMMRC$ be the distribution and $\svbz_1, \svbz_2,...,\svbz_N$ be the candidates obtained from Algorithm \ref{alg:mmrc} when the reference distribution is $\Unif(\cZ)$ where $\cZ$ is as defined in \eqref{eq:z_ss}.
Let $\theta$ denote the fraction of candidates inside $\msf{Cap}_{\svbx} = \cZ_{\svbx}$ where $\cZ_{\svbx}$ is the set of elements in $\cZ$ with $1$ in the same location as $\svbx$. 
It is easy to see that $\theta \sim \frac{1}{N}\msf{Binom}\lp N, \frac{s}{d} \rp$.
Let $\qMMRC_i = \Probability(z_i = 1)$ where $\svbz \sim \qMMRC(\cdot|\svbx)$ i.e., $\qMMRC_i = \Probability\lbp(\svbz_K)_i = 1\rbp$ where $K \sim \piMMRC(\cdot)$. 

\begin{restatable}{lemma}{mrmcssbias}\label{thm:mmrc_ss_bias_0}
Let $K \sim \piMMRC(\cdot)$ and $\qMMRC_i = \Probability\lbp(\svbz_K)_i = 1\rbp$ for $i \in [d]$. Then, 
\begin{align}
    \qMMRC_i = p_i  m_\texttt{mmrc} + b_\texttt{mmrc}
\end{align}
where
\begin{align}
m_\texttt{mmrc} &\coloneqq \frac{d}{d-1} \E\lb \frac{e^\varepsilon\theta}{e^\varepsilon \E\lb\theta\rb + (1-\E \lb\theta\rb)}  \cdot \Indicator \lp\theta\leq \E\lb\theta\rb \rp  + \frac{e^\varepsilon \E \lb\theta\rb +\theta- \E \lb\theta\rb}{e^\varepsilon \E\lb\theta\rb + (1-\E \lb\theta\rb)}  \cdot \Indicator \lp\theta> \E\lb\theta\rb \rp \rb -\frac{s}{d-1} \label{eq:m_mmrc_ss}\\
b_\texttt{mmrc} &\coloneqq \frac{1}{d-1}\lp s - \E\lb \frac{e^\varepsilon\theta}{e^\varepsilon \E\lb\theta\rb + (1-\E \lb\theta\rb)}  \cdot \Indicator \lp\theta\leq \E\lb\theta\rb \rp  + \frac{e^\varepsilon \E \lb\theta\rb +\theta- \E \lb\theta\rb}{e^\varepsilon \E\lb\theta\rb + (1-\E \lb\theta\rb)}  \cdot \Indicator \lp\theta> \E\lb\theta\rb \rp \rb \rp. \label{eq:b_mmrc_ss}
\end{align}
\end{restatable}
\begin{proof}
Following the proof of Lemma ~\ref{thm:mrc_ss_bias}, we compute 
$\Pr\lbp (\svbz_{K})_i = 1 | \svbx = i \rbp $ and $\Pr\lbp (\svbz_{K})_i = 1 | \svbx = j  \rbp $ separately.

To compute $\Pr\lbp (\svbz_K)_i = 1 \mv \svbx = i\rbp$, recall that $\theta$ denotes the fraction of candidates that belong inside the $\msf{Cap}_{\svbx}$ i.e., have $1$ in the same location as $\svbx$. From Appendix \ref{appendix:ss_cap}, recall that $c_1(\varepsilon,d) \coloneqq \dfrac{e^\varepsilon}{\binom{d-1}{s-1}e^\varepsilon + \binom{d-1}{s}}$, $c_2(\varepsilon,d) \coloneqq \dfrac{1}{\binom{d-1}{s-1}e^\varepsilon + \binom{d-1}{s}}$. Further, since $\svbz_k$ are generated uniformly at random,
$$ \theta \sim \frac{1}{N} \msf{Binom}\lp N, \frac{\binom{d-1}{s-1}}{\binom{d}{s}} \rp = \frac{1}{N} \msf{Binom}\lp N, \frac{s}{d} \rp,$$ so we have
\begin{align}
     & \Pr\lbp (\svbz_K)_i = 1 | \svbx = i \rbp \\
     &\qquad  = \Pr\lbp \svbz_K \in \msf{Cap}_{\svbx} | \svbx = i \rbp  \stackrel{(a)}{=} \E\lb \Pr\lbp \svbz_K  \in \msf{Cap}_{\svbx} |\svbx=i, \theta \rbp\rb \\
    &\qquad \stackrel{(b)}{=} \E\lb \frac{e^\varepsilon \theta}{e^\varepsilon \E\lb \theta \rb + (1-\E \lb \theta \rb)}  \cdot \Indicator \lp \theta \leq \E\lb \theta \rb \rp  + \frac{e^\varepsilon \E \lb \theta \rb + \theta - \E \lb \theta \rb}{e^\varepsilon \E\lb \theta \rb + (1-\E \lb \theta \rb)}  \cdot \Indicator \lp \theta > \E\lb \theta \rb \rp \rb  \label{eq:mmrc_ss_step_1}
\end{align}
where $(a)$ follows by the law of total probability and $(b)$ is due to Algorithm \ref{alg:mmrc} and $c_1(\varepsilon,d)/c_2(\varepsilon,d) = e^\varepsilon$.


To compute $\Pr\lbp (\svbz_K)_i = 1 \mv \svbx = j \rbp$, we decompose it into
\begin{align}
    \Pr\lbp (\svbz_K)_i = 1 \mv \svbx = j \rbp = \Pr\lbp (\svbz_K)_i = 1, (\svbz_K)_j = 1 \mv \svbx = j\rbp + \Pr\lbp (\svbz_K)_i = 1, (\svbz_K)_j = 0 \mv \svbx = j\rbp, \label{eq:mmrc_Step0}
\end{align}
for any $j \neq i$ and calculate each of the terms separately.

As before, let $\theta$ denotes the fraction of candidates that belong inside the $\msf{Cap}_{\svbx}$ i.e., have $1$ in the same location as $\svbx$. Further, let $\bar{\theta}$ denotes the fraction of candidates that belong inside the $\msf{Cap}_{\svbx}$ i.e., have $1$ in the same location as $\svbx$ as well as have $1$ in the $j^{th}$ location. Since $\svbz_k$ are generated uniformly at random,
$$ \bar{\theta} \sim \frac{1}{N} \msf{Binom}\lp N\theta, \frac{\binom{d-2}{s-2}}{\binom{d-1}{s-1}} \rp = \frac{1}{N} \msf{Binom}\lp N\theta, \frac{s-1}{d-1} \rp,$$ so we have
\begin{align}
    &\Pr\lbp (\svbz_K)_i = 1, (\svbz_K)_j = 1 | \svbx = j \rbp 
    \stackrel{(a)}{=} \E_{\theta}\lb \E_{\bar{\theta}}\lb \Pr\lbp (\svbz_K)_i = 1, (\svbz_K)_j = 1 \mv \svbx = j, \bar{\theta}, \theta \rbp \rb\rb \\
    & \stackrel{(b)}{=} \E_{\theta}\lb \E_{\bar{\theta}}\lb \frac{e^\varepsilon}{e^\varepsilon \E\lb \theta \rb + (1-\E \lb \theta \rb)} \times \bar{\theta} \rb \Indicator \lp \theta \leq \E\lb \theta \rb \rp
    +\E_{\bar{\theta}}\lb\frac{e^\varepsilon \E \lb \theta \rb + \theta - \E \lb \theta \rb}{\theta\lp e^\varepsilon \E\lb \theta \rb + (1-\E \lb \theta \rb)\rp }\times \bar{\theta}\rb \Indicator \lp \theta > \E\lb \theta \rb \rp\rb \\ 
    & \stackrel{(c)}{=} \frac{s-1}{d-1}\E\lb \frac{e^\varepsilon \theta}{e^\varepsilon \E\lb \theta \rb + (1-\E \lb \theta \rb)}  \cdot \Indicator \lp \theta \leq \E\lb \theta \rb \rp  + \frac{e^\varepsilon \E \lb \theta \rb + \theta - \E \lb \theta \rb}{e^\varepsilon \E\lb \theta \rb + (1-\E \lb \theta \rb)}  \cdot \Indicator \lp \theta > \E\lb \theta \rb \rp \rb. \label{eq:mmrc_Step1}
\end{align}
where $(a)$ follows by the law of total probability, $(b)$ follows 
from Algorithm \ref{alg:mmrc}, and $(c)$ is because $\E[\bar{\theta}] = \frac{s-1}{d-1} \times \theta$.

Similarly, to compute the term $\Pr\lbp (\svbz_K)_i = 1, (\svbz_K)_j = 0 | \svbx = j \rbp$, let $\bar{\theta}$ denote the fraction of candidates that belong inside the $\msf{Cap}_{\svbx}$ i.e., have $1$ in the same location as $\svbx$ as well as have $0$ in the $j^{th}$ location. Since $\svbz_k$ are generated uniformly at random,
$$ \bar{\theta} \sim \frac{1}{N} \msf{Binom}\lp N(1-\theta), \frac{\binom{d-2}{s-1}}{\binom{d-1}{s}} \rp = \frac{1}{N} \msf{Binom}\lp N(1-\theta), \frac{s}{d-1} \rp,$$ so we have
\begin{align}
    & \Pr\lbp (\svbz_K)_i = 1, (\svbz_K)_j = 0 | \svbx = j \rbp 
    \stackrel{(a)}{=} \E_{\theta}\lb \E_{\bar{\theta}}\lb \Pr\lbp (\svbz_K)_i = 1, (\svbz_K)_j = 0 \mv \svbx = j, \bar{\theta}, \theta \rbp \rb\rb \\
    & \stackrel{(b)}{=} \E_{\theta}\lb \E_{\bar{\theta}}\lb \frac{\bar{\theta}}{e^\varepsilon \E\lb \theta \rb + (1-\E\lb \theta \rb)} \cdot \Indicator \lp \theta > \E\lb \theta \rb \rp  + \frac{(1-\E\lb \theta \rb) + \lp \E \lb \theta \rb - \theta \rp e^{\varepsilon}}{(1-\theta) (e^\varepsilon \E\lb \theta \rb + (1-\E\lb \theta \rb))} \bar{\theta} \cdot \Indicator \lp \theta \leq \E\lb \theta \rb \rp  \rb \rb\\ 
    & \stackrel{(c)}{=} \frac{s}{d-1}\E\lb \frac{(1-\theta)}{e^\varepsilon \E\lb \theta \rb + (1-\E\lb \theta \rb)} \cdot \Indicator \lp \theta > \E\lb \theta \rb \rp  + \frac{(1-\E\lb \theta \rb) + \lp \E \lb \theta \rb - \theta \rp e^{\varepsilon}}{e^\varepsilon \E\lb \theta \rb + (1-\E\lb \theta \rb)} \cdot \Indicator \lp \theta \leq \E\lb \theta \rb \rp  \rb. \label{eq:mmrc_Step2}
\end{align}
where $(a)$ follows by the law of total probability, $(b)$ follows 
from Algorithm \ref{alg:mmrc}, and $(c)$ is because $\E[\bar{\theta}] = \frac{s}{d-1} \times \theta$. Using \eqref{eq:mmrc_Step1} and \eqref{eq:mmrc_Step2} in \eqref{eq:mmrc_Step0}, we have
\begin{align}
    & \Pr\lbp (\svbz_K)_i = 1 | \svbx = j \rbp \\
    & \qquad = \Pr\lbp (\svbz_K)_i = 1, (\svbz_K)_j = 1 | \svbx = j \rbp+\Pr\lbp (\svbz_K)_i = 1, (\svbz_K)_j = 0 | \svbx = j \rbp \\
    &\qquad  = \frac{1}{d-1}\lp s - \E\lb \frac{e^\varepsilon \theta}{e^\varepsilon \E\lb \theta \rb + (1-\E \lb \theta \rb)}  \cdot \Indicator \lp \theta \leq \E\lb \theta \rb \rp  + \frac{e^\varepsilon \E \lb \theta \rb + \theta - \E \lb \theta \rb}{e^\varepsilon \E\lb \theta \rb + (1-\E \lb \theta \rb)}  \cdot \Indicator \lp \theta > \E\lb \theta \rb \rp \rb \rp.\label{eq:mmrc_ss_step_2}
\end{align}
Combining everything, we have
\begin{align}
    \qMMRC_i  & =   \Probability\lbp(\svbz_K)_i  =   1\rbp \\
    & =   p_i \times \lb \Pr\lbp (\svbz_K)_i   =   1 \mv \svbx = i\rbp - \Pr\lbp (\svbz_K)_i   =   1 \mv \svbx = j\rbp \rb + \Pr\lbp (\svbz_K)_i   =   1 \mv \svbx = j\rbp  \\
    & \stackrel{(a)}{=}   p_i  m_\texttt{mmrc} + b_\texttt{mmrc}
\end{align}
where $(a)$ follows from \eqref{eq:mmrc_ss_step_1} and \eqref{eq:mmrc_ss_step_2}, and the definitions of $m_\texttt{mmrc}$ and $b_\texttt{mmrc}$. 
\end{proof}

\mmrcssbias*
\begin{proof}
Given Lemma \ref{thm:mmrc_ss_bias_0}, the proof follows from the proof of Lemma \ref{thm:mrc_ss_bias}.
\end{proof}

 \subsection{Utility of Modified Minimal Random Coding simulating \texorpdfstring{$\SubsetSelection$}{Subset Selection}}\label{appendix:mmrc_ss_ut}
 
 \subsubsection{The scaling factors of \texorpdfstring{$\SubsetSelection$}{Subset Selection} and \texorpdfstring{$\MMRC$}{MMRC} are close when \texorpdfstring{$N$}{N} is of the right order}\label{appendix:scaling_mmrc_ss}
 In the following Lemma, we show that when the number of candidates $N$ is exponential in $\varepsilon$, then the scaling parameters associated with $\SubsetSelection$ and the $\MMRC$ scheme simulating $\SubsetSelection$ are close.
 
\begin{lemma}\label{lemma:MMRC_SS_N}
Let $N$ denote the number of candidates used in the \emph{$\MMRC$} scheme. Let $K \sim \piMMRC$ where $\piMMRC$ is the distribution over the indices $[N]$ associated the \emph{$\MMRC$} scheme simulating \emph{$\SubsetSelection$}. Consider any $\lambda > 0$.
Then, the scaling factors $m_{\texttt{ss}}$ and $b_{\texttt{ss}}$ associated with \emph{$\SubsetSelection$} and the scaling factors $m_\texttt{mmrc}$ and $b_{\texttt{mmrc}}$ associated with the \emph{$\MMRC$} scheme simulating \emph{$\SubsetSelection$} are such that
\begin{align}
    m_{\texttt{ss}} - m_\texttt{mmrc} \leq \lambda\cdot m_\texttt{mmrc}
\end{align}
and $b_{\texttt{ss}} \leq b_\texttt{mmrc}$
as long as
\begin{align}
    N\geq \frac{2(e^{\varepsilon}+1)^2(1+\lambda)^2}{0.24^2\lambda^2}\ln\lp \frac{8(1+\lambda)}{0.24\lambda}\rp.
\end{align}
\end{lemma}
\begin{proof}
The proof is similar to the proof of Lemma \ref{lemma:mrc_ss_approximation_error}. We only show the key steps here.

From \eqref{eq:b_mmrc_ss} and \eqref{eq:check_mb_def_2}, we have
\begin{align}
    b_\texttt{mmrc} - b_{\texttt{ss}} & =  \frac{1}{d-1} \cdot \frac{1}{e^\varepsilon \E [\theta] +(1-\E [\theta] )} \cdot \E\lb e^\varepsilon (\E [\theta]  - \theta) \cdot \Indicator \lp \theta \leq \E\lb \theta \rb \rp + (\E [\theta]  - \theta) \cdot \Indicator \lp \theta > \E\lb \theta \rb \rp \rb \\
    & \stackrel{(a)}{\geq} \frac{1}{d-1} \cdot \frac{1}{e^\varepsilon \E [\theta] +(1-\E [\theta] )} \cdot \E\lb  (\E [\theta]  - \theta) \cdot \Indicator \lp \theta \leq \E\lb \theta \rb \rp + (\E [\theta]  - \theta) \cdot \Indicator \lp \theta > \E\lb \theta \rb \rp \rb \\
    & = \frac{1}{d-1} \cdot \frac{1}{e^\varepsilon \E [\theta] +(1-\E [\theta] )} \cdot \E\lb (\E [\theta]  - \theta) \rb = 0.
\end{align}
where $(a)$ follows because $e^\varepsilon \geq 1$.
From \eqref{eq:m_mmrc_ss} and \eqref{eq:check_mb_def_1}, we have
\begin{align}
    m_{\texttt{ss}} - m_\texttt{mmrc} & = \frac{d}{d-1} \cdot \frac{1}{e^\varepsilon \E [\theta] +(1-\E [\theta] )} \cdot \E\lb e^\varepsilon (\E [\theta]  - \theta) \cdot \Indicator \lp \theta \leq \E\lb \theta \rb \rp + (\E [\theta]  - \theta) \cdot \Indicator \lp \theta > \E\lb \theta \rb \rp \rb \\
    & \leq \frac{d}{d-1} \cdot \frac{1}{e^\varepsilon \E [\theta] +(1-\E [\theta] )} \cdot \E\lb e^\varepsilon (\E [\theta]  - \theta) \cdot \Indicator \lp \theta \leq \E\lb \theta \rb \rp \rb\\
    & \stackrel{(a)}{\leq}  \frac{2}{e^\varepsilon \E [\theta] +(1-\E [\theta] )} \cdot \E\lb e^\varepsilon (\E [\theta]  - \theta) \cdot \Indicator \lp \theta \leq \E\lb \theta \rb \rp \rb \label{eq:ss_mmrc_minus_m1}
\end{align}
where $(a)$ holds since $d\geq 2$. Next, we condition on the event $\mcal{E} \coloneqq \lbp \lba \E[\theta]-\theta\rba \leq \sqrt{\frac{\ln(2/\beta)}{2N}} \rbp$, which has probability $\Pr_\theta\lbp \mcal{E} \rbp \geq 1-\beta$ by Hoeffding's inequality. We continue to upper bound \eqref{eq:ss_mmrc_minus_m1}:
\begin{align}
m_{\texttt{ss}} - m_\texttt{mmrc}
&= 2 \bigg( \Pr\lbp \mcal{E} \rbp\E\lb \frac{ e^\varepsilon (\E [\theta]  - \theta) \cdot \Indicator \lp \theta \leq \E\lb \theta \rb \rp }{e^\varepsilon \E [\theta] +(1-\E [\theta] )} \mv \mcal{E} \rb \\
& +\Pr\lbp \mcal{E}^c \rbp\E\lb \frac{ e^\varepsilon (\E [\theta]  - \theta) \cdot \Indicator \lp \theta \leq \E\lb \theta \rb \rp }{e^\varepsilon \E [\theta] +(1-\E [\theta] )} \mv \mcal{E}^c \rb \bigg)\\
& \stackrel{(a)}{\leq} 2 \lp \E\lb \frac{ e^\varepsilon (\E [\theta]  - \theta) \cdot \Indicator \lp \theta \leq \E\lb \theta \rb \rp }{e^\varepsilon \E [\theta] +(1-\E [\theta] )} \mv \mcal{E} \rb +\beta \rp\\
& \stackrel{(b)}{\leq} (1+e^\varepsilon) \sqrt{\frac{\ln(2/\beta)}{2N}}  + 2\beta
\end{align}
where $(a)$ holds since 
 $$ \frac{ e^\varepsilon (\E [\theta]  - \theta) \cdot \Indicator \lp \theta \leq \E\lb \theta \rb \rp }{e^\varepsilon \E [\theta] +(1-\E [\theta] )} = \frac{ e^\varepsilon \E [\theta]  \cdot \Indicator \lp \theta \leq \E\lb \theta \rb \rp }{e^\varepsilon \E [\theta] +(1-\E [\theta] )} -  \frac{ e^\varepsilon \theta \cdot \Indicator \lp \theta \leq \E\lb \theta \rb \rp }{e^\varepsilon \E [\theta] +(1-\E [\theta] )} \leq 1,$$ and $(b)$ holds since $\E[\theta] = s/d \geq 1/(1+e^\varepsilon)$. 

The rest of the proof is similar to the proof of Lemma \ref{lemma:mrc_ss_approximation_error}.
\end{proof}

\subsubsection{Relationship between the mean squared errors associated with \texorpdfstring{$\SubsetSelection$}{Subset Selection} and \texorpdfstring{$\MMRC$}{MMRC} simulating \texorpdfstring{$\SubsetSelection$}{Subset Selection}}\label{appendix:mmrc_ss_scaling_mse}
In the following Proposition,
we show that if $m_{\texttt{mmrc}}$ is close to $m_{\texttt{ss}}$ and $b_{\texttt{mmrc}} \geq b_{\texttt{ss}}$, then the mean squared error associated with $\MMRC$ simulating $\SubsetSelection$ (i.e., $\Expectation_{\qMMRC} \big[ \lV  \hat{\svbx}^\texttt{mmrc} - \svbx \rV^2_2  \big]$) is close to the mean squared error associated with $\SubsetSelection$ (i.e., $\Expectation_{\qSS} \big[ \lV  \hat{\svbx}^\texttt{ss} - \svbx \rV^2_2  \big]$).
\begin{proposition}\label{proposition:mmrc_mse_wrt_ss}
Let $\qSS(\svbz | \svbx)$ be the  $\varepsilon$-LDP \emph{$\SubsetSelection$} mechanism with estimator $\hat{\svbx}^{\texttt{ss}}$. Let $\qMMRC(\svbz|\svbx)$ denote the \emph{$\MMRC$} privatization mechanism simulating \emph{$\SubsetSelection$} with $N$ candidates and estimator $\hat{\svbx}^{\texttt{mmrc}}$.
Let $m_{\texttt{ss}}$ and $b_{\texttt{ss}}$ denote the scaling factors  associated with \emph{$\SubsetSelection$} and $m_\texttt{mmrc}$ and $b_\texttt{mmrc}$ denote the scaling factors associated with the \emph{$\MMRC$} scheme simulating \emph{$\SubsetSelection$}. Consider any $\lambda > 0$. If $m_{\texttt{pu}} - m_\texttt{mmrc} \leq \lambda \cdot m_\texttt{mmrc}$ and $b_{\texttt{mmrc}} \geq b_{\texttt{ss}}$, then 
  \begin{align}
    \Expectation_{\qMMRC} \big[ \lV  \hat{\svbx}^\texttt{mmrc} - \svbx \rV^2_2  \big]  \leq  \lp 1+4\lambda +5\lambda^2 + 2\lambda^3 \rp  \Expectation_{\qSS}\big[\|\hat{\svbx}^{\texttt{ss}} - \svbx\|^2\big]
\end{align}
\end{proposition}
\begin{proof}
The proof is similar to the proof of Proposition \ref{proposition:mrc_mse_wrt_ss}.
\end{proof}

In the following Lemma, we show that with on the order of $\varepsilon$-bits of communication, the mean squared error associated with $\MMRC$ simulating $\SubsetSelection$ (i.e., $\Expectation_{\qMMRC} \big[ \lV  \hat{\svbx}^\texttt{mmrc} - \svbx \rV^2_2  \big]$) is close to the mean squared error associated with $\SubsetSelection$ (i.e., $\Expectation_{\qSS} \big[ \lV  \hat{\svbx}^\texttt{ss} - \svbx \rV^2_2  \big]$).
 \begin{restatable}{lemma}{ssmmrc}\label{theorem:mmrc_accuracy_ss}
 Let $\qSS(\svbz | \svbx)$ be the  $\varepsilon$-LDP \emph{$\SubsetSelection$} mechanism with parameters $d$ and $s=\lceil \frac{d}{1+e^\varepsilon} \rceil$ and estimator $\hat{\svbx}^{\texttt{ss}}$. Let $\qMMRC(\svbz|\svbx)$ denote the \emph{$\MMRC$} privatization mechanism simulating \emph{$\SubsetSelection$} with $N$ candidates and estimator $\hat{\svbx}^{\texttt{mmrc}}$ as defined above. Consider any $\lambda > 0$. Then,
 \begin{align}
      & \E_{\qMMRC} \lb \lV \hat{\svbx}^{\texttt{mmrc}}  -  \svbx\rV^2_2\rb   \leq  (1 + 4\lambda + 5\lambda^2 + 2\lambda^3) \E_{\qSS} \lb \lV \hat{\svbx}^{\texttt{ss}} - \svbx\rV^2_2\rb,
 \end{align}
  as long as 
 \begin{align}
      N\geq \frac{2(e^{\varepsilon}+1)^2(1+\lambda)^2}{0.24^2\lambda^2}\ln\lp \frac{8(1+\lambda)}{0.24\lambda}\rp.
  \end{align}
 \end{restatable}
 \begin{proof}
The proof follows from Proposition \ref{proposition:mmrc_mse_wrt_ss} and Lemma \ref{lemma:MMRC_SS_N}.
\end{proof}
 
\subsubsection{Simulating \texorpdfstring{$\SubsetSelection$}{Subset Selection} using Modified Minimal Random Coding}\label{appendix:mmrc_ss_utility}
Now, we provide a proof of Theorem \ref{thm:fe_mmrc_ss}.
\mmrcss*
\begin{proof}
The proof follows directly from Lemma~\ref{theorem:mmrc_accuracy_ss} since for all $i \in [n]$, $\hat{\svbx}^{\texttt{mmrc}}_i$ are independent of each other as well as unbiased.
\end{proof}

\subsection{Additional Empirical Comparisons}\label{appendix:mmrc_ss_emp}
In Section \ref{subsec:mmrc_ss_empirical}, we empirically demonstrated the privacy-accuracy-communication tradeoffs of $\MMRC$ simulating $\SubsetSelection$ against $\SubsetSelection$ and RHR in terms of $\ell_2$ error vs $\#$bits and $\ell_2$ error vs $\varepsilon$ (see Figure \ref{fig:freq}). In this section, we provide comparisons between these methods in terms of $\ell_2$ error vs $d$ (see Figure \ref{fig:freq_app} (left)) and $\ell_2$ error vs $n$ (see Figure \ref{fig:freq_app} (right)) for a fixed $\varepsilon$ (=6) and a fixed $\#$bits (=14). As before, RHR uses $\#$bits $= \varepsilon$ for both because it leads to a poor performance if $\#$bits $ > \varepsilon$.
\begin{figure}[h]
\centering
\includegraphics[width=0.45\linewidth]{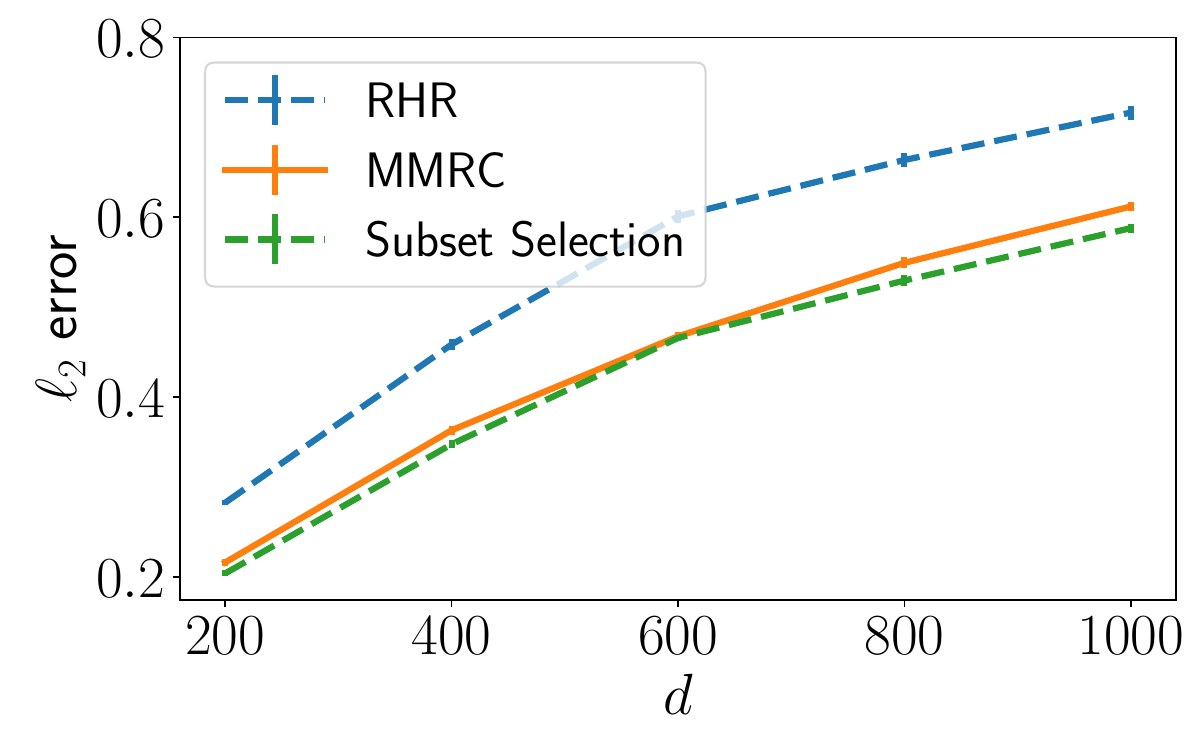} \qquad \includegraphics[width=0.45\linewidth]{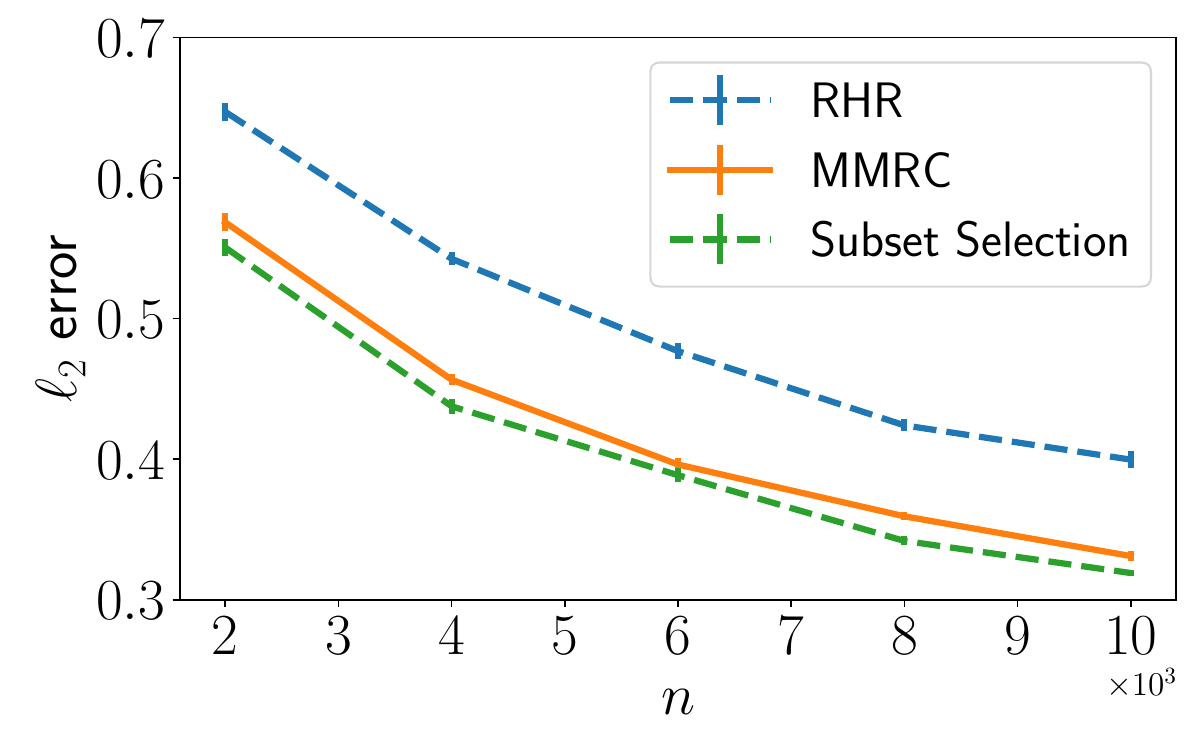}%
\caption{Comparing $\SubsetSelection$, $\MMRC$ simulating $\SubsetSelection$ and RHR for frequency estimation with $\varepsilon=6$ and $\#$bits $=14$. \textbf{Left:} $\ell_2$ error vs $d$ for $n = 5000$. \textbf{Right:} $\ell_2$ error vs $n$ for $d = 500$.}
\label{fig:freq_app}
\end{figure}

\end{document}